\documentclass[12pt]{article}

\usepackage{amssymb,amsmath,amsfonts,eurosym,geometry,ulem,graphicx,color,setspace,sectsty,comment,footmisc,caption,natbib,pdflscape,array}
\usepackage{subcaption}

\newcolumntype{L}[1]{>{\raggedright\let\newline\\arraybackslash\hspace{0pt}}m{#1}}
\newcolumntype{C}[1]{>{\centering\let\newline\\arraybackslash\hspace{0pt}}m{#1}}
\newcolumntype{R}[1]{>{\raggedleft\let\newline\\arraybackslash\hspace{0pt}}m{#1}}

\usepackage{bbm}				
\usepackage{mathtools}			
\usepackage{threeparttable}		
\usepackage{booktabs}			
\usepackage{here}				
\usepackage{multirow}			
\usepackage{xr}				
\usepackage{afterpage}			
\usepackage{amsthm}			
\usepackage[colorlinks=true,citecolor=blue,urlcolor=blue,pdfpagemode=UseNone,pdfstartview=FitH]{hyperref}	
\usepackage{arydshln} 			
\newtheorem{theorem}{Theorem}[section]
\newtheorem{assumption}{Assumption}[section]
\newtheorem{proposition}{Proposition}[section]
\newtheorem{corollary}{Corollary}[section]
\newtheorem{lemma}{Lemma}[section]
\newtheorem{example}{Example}[section]
\newtheorem{remark}{Remark}[section]

\newtheorem{definition}{Definition}	

\usepackage{chngcntr}
\usepackage{apptools}
\AtAppendix{\counterwithin{assumption}{section}}
\AtAppendix{\counterwithin{proposition}{section}}
\AtAppendix{\counterwithin{corollary}{section}}
\AtAppendix{\counterwithin{lemma}{section}}
\AtAppendix{\counterwithin{example}{section}}
\AtAppendix{\counterwithin{remark}{section}}
\AtAppendix{\counterwithin{definition}{section}}

\usepackage{etoolbox}
\makeatletter
\patchcmd{\NAT@test}{\else \NAT@nm}{\else \NAT@nmfmt{\NAT@nm}}{}{}
\DeclareRobustCommand\citepos
{\begingroup
	\let\NAT@nmfmt\NAT@posfmt
   	\NAT@swafalse\let\NAT@ctype\z@\NAT@partrue
   	\@ifstar{\NAT@fulltrue\NAT@citetp}{\NAT@fullfalse\NAT@citetp}
}
\let\NAT@orig@nmfmt\NAT@nmfmt
\def\NAT@posfmt#1{\NAT@orig@nmfmt{#1's}}

\begin{document}

\begin{titlepage}
	\title{Inference in Linear Dyadic Data Models \\ with Network Spillovers}
	\author{Nathan Canen\thanks{University of Houston, University of Warwick and Research Economist at NBER. Department of Economics, University of Warwick, Coventry CV4 7AL, U.K. Email: ncanen@uh.edu} \and Ko Sugiura\thanks{Department of Economics, University of Houston, Science Building 3581 Cullen Boulevard Suite 230, Houston, 77204-5019, Texas, United States. Email: ksugiura@uh.edu}}
	\date{\today}
	\maketitle
	\begin{abstract}
		\noindent When using dyadic data (i.e., data indexed by pairs of units), researchers typically assume a linear model, estimate it using Ordinary Least Squares and conduct inference using ``dyadic-robust" variance estimators. The latter assumes that dyads are uncorrelated if they do not share a common unit (e.g., if the same individual is not present in both pairs of data). We show that this assumption does not hold in many empirical applications because indirect links may exist due to network connections, generating correlated outcomes. Hence, ``dyadic-robust'' estimators can be biased in such situations. We develop a consistent variance estimator for such contexts by leveraging results in network statistics. Our estimator has good finite sample properties in simulations, while allowing for decay in spillover effects. We illustrate our message with an application to politicians' voting behavior when they are seating neighbors in the European Parliament.\\ 
		\vspace{0in}\\
		\noindent\textbf{Keywords:} Dyadic data, Networks, Inference, Cross-sectional dependence, Congressional Voting.\\
		\vspace{0in}\\
		\bigskip
	\end{abstract}
	\setcounter{page}{0}
	\thispagestyle{empty}
\end{titlepage}

\pagebreak \newpage



\section{Introduction}\label{sec:Introduction}


Dyadic data is categorized by the dependence between two sets of sampled units (dyads). For example, exports between the U.S. and Canada depend on both countries (and, plausibly, their characteristics). This contrasts to classical data in the social sciences that only depends on a single unit of observation (e.g., the GDP of the U.S., or a politician's vote in a roll-call). 

The empirical relevance of dyadic data is showcased by its widespread use, which has increased over the past two decades (\citet{Graham-2020a} provides an extensive review). For example, applications are found in political economy (correlation in voting behavior in Parliament across seating neighbors, \citet{Harmon_et_al-2019}), international political economy and trade (export-import outcomes across countries, \citet{Anderson_and_van_Wincoop-2003}), international relations (\citet{hoff04}, for a salient example), among many others. In fact, dyadic data is considered to be dominant in quantitative international relations \citep{poast}. In these examples, applied researchers typically model the dependence between dyadic outcomes and observable characteristics using a linear model, which they then estimate using Ordinary Least Squares (OLS). However, inference on such estimators for the linear parameters is more complex. 

The main approach in recent applied work has been the use of the so-called ``dyadic-robust" estimators (e.g., \citet{Cameron_et_al-2011}, \citet{Aronow_et_al-2015}, and \citet{Tabord-Meehan-2019}, among others). Such estimators build on the widely-used assumption in dyadic data that the error terms for dyad $(i, j)$ and for dyad $(k, l)$ can only be correlated if they share a unit (see \citet{Aronow_et_al-2015} and \citet{Tabord-Meehan-2019} for a discussion; \citet{Cameron_and_Miller-2014wp} for a review). 

In this paper, we first argue that such an assumption does not hold in many applications using dyadic data where dyads may be indirectly connected along a network.\footnote{This is a concrete class of applied examples where the assumption fails. The possibility that cross-sectional dependence in dyadic data might be more extensive than assumed has been pointed out by \citet{Cameron_and_Miller-2014wp} and \cite{cranmerdemarais}.} Figure \ref{fig:ExampleOfPolNetwork} presents a simple example in the context of politicians in Congress, whose votes or decisions depend on their seating neighbors. It is completely possible that behavior across dyads $(A,B)$ and $(C,D)$ might be correlated along unobservables because they have many indirect connections (in the figure, through $A$ sitting next to $B$, who sits next to $C$).\footnote{We expand on these examples in the next section. Such spillovers could be further rationalized as individual-level unobserved heterogeneity: e.g., an unmeasured preference for voting Yes, or a preference for trading with a certain country (see \citet{Graham-2020b} and references therein for details).} We show that such spillovers invalidate the assumptions for consistency of existing ``dyadic-robust" variance estimators through generating interdependence, implying that they are biased for the true asymptotic variance when dyads may be correlated even when they do not share a common unit.

\begin{figure}[htbp]
	\begin{center}
	\caption{Hypothetical Example of a Network in Parliament}
	\includegraphics[scale=0.4, trim = 6cm 5cm 12cm 3.5cm, clip]{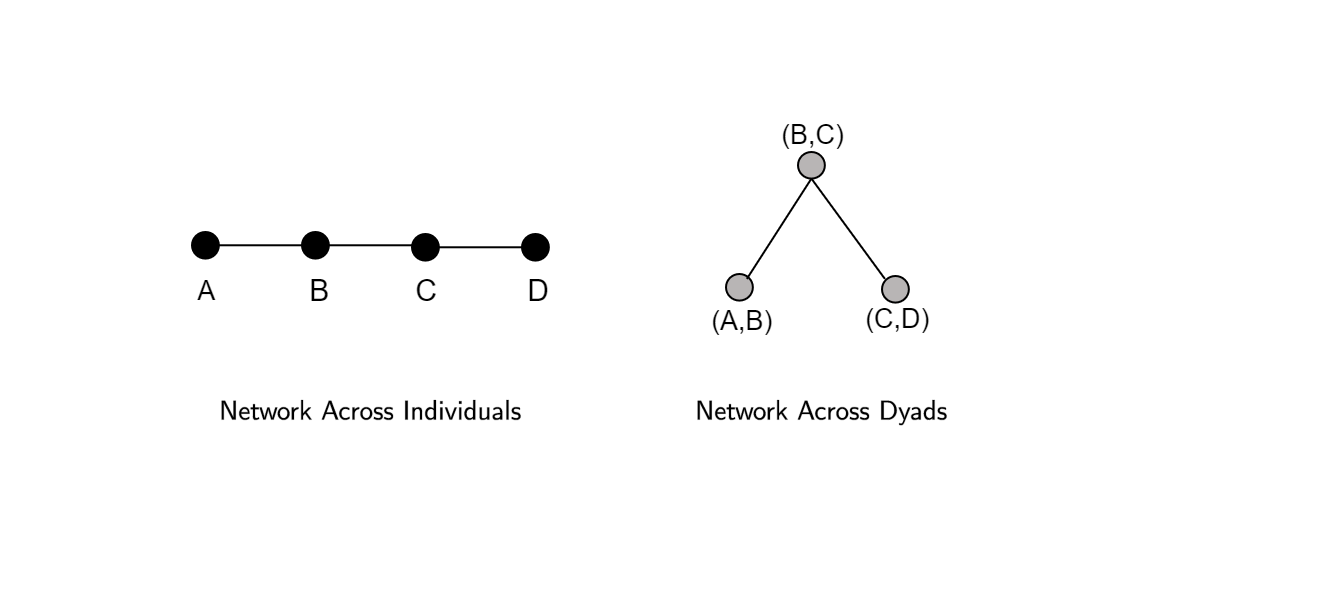}
		\label{fig:ExampleOfPolNetwork}
		\end{center}
		\begin{tablenotes}[para,flushleft]
			\item[] {\footnotesize {\sl Notes}: The left figure shows a hypothetical example of politician networks based on seating arrangements: $A$ sits beside $B$, who sits beside $C$, who sits beside $D$. The right-hand figure illustrates the resulting network among active dyads. As dyads $(A,B)$ and $(B,C)$ share a unit, they are indirectly linked in the dyadic network. However, though dyads $(A,B)$ and $(C,D)$ do not have a politician in common, they might still be correlated through two indirect links: namely, $B$ sits beside $C$, who sits beside $D$. Hence, $D$'s actions can affect politician $A$.}
		\end{tablenotes}
\end{figure}

To deal with these issues, we develop a consistent variance estimator that explicitly accounts for such network spillovers even with dyadic data, thereby complementing existing approaches (e.g., \cite{Aronow_et_al-2015}).\footnote{We provide an extensive comparison of the relative benefits of each approach in the next section. We note here, though, that neither approach subsumes the other, as they depend on different assumptions and may be more appropriate for different applications.} We prove that our proposed variance estimator is consistent for the true variance of the OLS estimator in linear models with dyadic data when the cross-sectional dependence follows an observed (exogenous) network. Our main insight is that the dependence across all dyads, including indirect spillovers, can be rewritten as correlations across a specific network over dyads. This allows us to apply the framework of \citet{Kojevnikov_et_al-2021} to such network random variables, although here it is a network over dyads, rather than individuals. Monte Carlo simulations show that our proposed estimator has good finite sample properties and outperforms other estimators for the relevant contexts.

To help practitioners, we then provide a step-by-step guideline on whether our estimator may be appropriate to their context. As we describe, this choice depends on: (i) whether spillovers from indirectly connected dyads are likely to be present, (ii) whether the researcher observes/constructs the network among dyads through which spillovers propagate, and (iii) whether those spillovers are likely to be persistent. Our variance estimator is consistent for the asymptotic variance of the OLS estimator even under (i)-(iii). And our estimator can account for decay in propagation, as Corollary \ref{coro:ExogenousRegression:NetworkHACEstimators:InconsistencyResult} and Example \ref{example:MaximumAdmissibleBiasInTheDyadic-RobustVarianceEstimator} illustrate. On the other hand, if decay is very high, or spillovers along the network do not exist, then those same results imply that the estimator of \cite{Aronow_et_al-2015} will also be consistent.

Finally, we illustrate the extent to which neglecting network spillovers with dyadic data may bias inference results. Beyond Monte Carlo simulations, we revisit the application in \citet{Harmon_et_al-2019} of voting in the European Parliament. The authors study whether random seating arrangements (based on naming conventions) induce neighboring politicians to agree with one another in policy votes. The outcome, whether politicians $i$ and $j$ vote the same way on a policy, is dyadic in nature. However, $i$ and $j$'s votes may be positively correlated even if they are not neighbors: for instance, $i$ and $j$ may sit on either side of common neighbors $k$ and $l$, who influence them both, and this seating arrangement is observed. This chain of influences is sufficient to induce strong positive correlation across non-dyads. We show that neglecting such higher order spillovers has significant empirical consequences: their estimated variance using the estimator in \citet{Aronow_et_al-2015} is roughly $22\%$ smaller than using our consistent estimator accounting for such spillovers; while the estimate based on the Eicker-Huber-White estimator ignoring spillovers is approximately 73\% smaller than our proposal, consistent with the arguments of \cite{pinto}.


\subsection{Related Literature}\label{subsec:RelatedLiterature}

The use of dyadic data in Political Science has a rich history, particularly in International Relations. However, empirical challenges with such models are well known -- see \cite{poast} for a historical overview. Early on, the concerns were mostly about model specification, including the error term. This includes the 2001 special issue of \textit{International Organization}, mostly focusing on the use of fixed effects. More recently, \cite{pinto} pointed out that ignoring dependence across dyads can lead to erroneous hypothesis testing, as computed standard errors would be too small. \citet{hoff04} and \citet{minhas19, minhas22} suggest including random coefficients and latent variables to account for dependencies across dyads. Our approach explicitly accounts for the whole network of interdependencies across dyads, which can go beyond third-order dependences (assumed in \cite{minhas19,minhas22}). It does so by using asymptotic inference, rather than Bayesian (\citealp{minhas19,minhas22}) or randomized inference (\citealp{pinto}).

As a result, our paper is directly related to the literature on (asymptotic) inference in regression analysis with dyadic random variables. \citet{Aronow_et_al-2015} and \citet{Tabord-Meehan-2019} consider OLS estimation and inference in a linear dyadic regression model. Meanwhile, \citet{Graham-2020a} and \citet{Graham-2020b} explore a likelihood-based approach to dyadic regression models, while \citet{Graham_et_al-2022wp} and \citet{Chiang_and_Tan-2022} provide results for kernel density estimation in dyadic regression models.\footnote{Similarly, our paper is related to the recent literature on inference for multiway clustering, whether OLS estimation with multiway clustering (\citet{Cameron_et_al-2011}); a clustering method in high-dimensional set-ups (\citet{Chiang_et_al-2021}); clustering within the time dimension (\citet{Chiang_et_al-2022wpa}); clustered inference with empirical likelihood (\citet{Chiang_et_al-2022wpb}); bootstrap methods in multiway clustering \citep{Davezies_et_al-2021, Menzel-2021, Mackinnon_et_al-2022wpa}; clustering in the context of average treatment effects (\citet{Abadie_et_al-2022}), to name but a few (see \citet{MacKinnon_et_al-2022} for review). Again, one of the common assumptions in this literature is that individual observations are divided into disjoint groups -- clusters -- and observations in different clusters are not correlated. To that end, \citet{Mackinnon_et_al-2022wpb} propose measures of cluster-level influence that can be used to assess whether the underlying assumption of cluster-robust variance estimation is satisfied.} While useful to allow for correlations along time and within such groups, this separable structure may be inappropriate for environments where spillovers follow a complex form of dependence along a network. 
As alluded to previously, all such papers assume that dyads are uncorrelated unless they share a common unit, thereby ignoring dependence across indirectly linked dyads, while our estimator explicitly captures the higher order correlation across dyads along the dyadic network. 

However, we emphasize that neither approach subsumes the other. The papers cited above leave the dependence within ``clusters" (groups of dyads that share units) unrestricted, but assume independence across such clusters. This is akin to the literature with one-way clustering (e.g., \citet{Hansen_and_Lee-2019}). By comparison, our approach restricts such dependence among groups of dyads that share units (i.e., dependence is assumed to follow the observed network), but allows for dependence across such ``clusters" of dyads along the dyadic network.\footnote{As a result, our approach complements the existing methods similar to how inference with spatial data (e.g., \citet{Conley-1999} and \citet{Jenish_and_Prucha-2009}) complements one-way clustering inference. Our approach still differs from such inference with spatial data, since the latter routinely assumes the index set to be a Euclidean metric space, whose metric relies solely on the nature of the space, and uses it to define the dependence between variables. See also \citet{Ibragimov_and_Muller-2010} and references therein. In our context, however, the index set of dyads alone does not suffice to dictate the dependence structure because indices themselves do not inform us of the network topology. Instead, we first introduce a metric on a network among dyads and our mixing condition is based on dependence as dyads grow further apart along the network.} Which approach to pursue using dyadic data depends on the researchers' applications and how they fit such assumptions.

\section{Set-up}\label{sec:Set-up}


Assume that we observe a cross section of $N\in \mathbbm{N}$ individuals located along a network -- the latter interpretable as politicians, countries, firms, or other observation units depending on the context. The dyads present in the $N$-individual network (i.e., among the $\tbinom{N}{N-1}$ possible dyads), are called active dyads, so that the dyad for two units $i$ and $j$ (e.g., politicians, countries) is denoted as some $m$. The set of active dyads is denoted $\mathcal{M}_{N}$ and $M$ denotes the cardinality of that set.

We assume that each dyad $m$ is endowed with a triplet of dyad-specific variables, forming a triangular array $\{(y_{M, m}, x_{M, m}, \varepsilon_{M, m})\}_{m\in\mathcal{M}_{N}}$ with respect to $M$, where $y_{M,m}\in \mathbbm{R}$ is a one-dimensional observable outcome, $x_{M,m}\in \mathbbm{R}^{K}$ is a $K$-dimensional vector of observable characteristics with $K\in\mathbbm{N}$, and $\varepsilon_{M,m}\in\mathbbm{R}$ is a one-dimensional random error term that is not observable to the researcher. We only consider exogenous network formation and the network is assumed to be observable. These conditions are summarized in the following assumption:

\begin{assumption}[Exogenous and Observable Dyadic Networks]\label{assum:exog_networks}
The network among dyads is assumed to be conditionally independent of $\{\varepsilon_{M,m}\}_{m\in\mathcal{M}_{N}}$. Furthermore, this network among the $N$ individuals is assumed to be observable.
\end{assumption}

While such assumptions are standard in models of dyadic networks, they seem particularly appropriate when units or dyad pairs are linked across geographical, physical, or ex-ante social relations (e.g., family ties). This includes capturing neighboring and regional spillovers across countries, as often done in international relations, or exogenous seating arrangements in Parliament, as illustrated in the examples in the next section. 

The subsequent arguments require us to distinguish between a pair of dyads who share a member (i.e., who are directly linked -- which we call, \textit{adjacent}) and a pair of dyads who are directly or indirectly linked (which we call, simply, \textit{connected}).

\begin{definition}[Adjacent \& Connected Dyads]\label{defn:AdjacentAndConnectedDyads}
	Two active dyads $m$ and $m'$ are said to be {\sl adjacent} if they have an individual in common; and they are called {\sl connected} if they are linked through pairs of adjacent dyads.
 \end{definition}
 
In Figure \ref{fig:ExampleOfPolNetwork}, dyad (A,B) is adjacent to (B,C), and connected with, though not adjacent to, (C,D). Hence, the adjacency relationship constitutes a network structure among active dyads, and thus, a network over individuals can be transformed to one over active dyads. For example, the right-hand side panel of Figure \ref{fig:ExampleOfPolNetwork} provides a network over pairs of voting politicians (i.e., active dyads).\footnote{This corresponds to thinking about the {\sl line graph} of the original graph over individuals.} We define the geodesic distance between two connected dyads $m$ and $m'$ to be the smallest number of adjacent dyads between them. Note that adjacent dyads are a special case of connected dyads with geodesic distance equal to one. 

\subsection{The Linear Model}

\subsubsection{Set-up \& Identification}\label{subsubsec:Set-upAndIdentification}

The cross-sectional model of interest takes the form of the linear network-regression model: for any $N\in\mathbbm{N}$,
\begin{align}\label{eq:ExogenousRegression:PopulationNetworkRegression_a}
	y_{M,m} =  x_{M,m}' \beta + \varepsilon_{M,m} \ \ \ \ \ \forall m\in\mathcal{M}_{N}^{},
\end{align}
where 
\begin{align}\label{eq:ExogenousRegression:PopulationNetworkRegression_b}
	Cov(\varepsilon_{M, m}, \varepsilon_{M, m'}\mid X_{M})
	= 0	\quad \text{unless $m$ and $m'$ are connected},
\end{align}
and $\beta$ is a $K\times 1$ vector of the regression coefficients and $X_{M}$ denotes the $M\times K$ matrix that records the observed dyad-specific characteristics, i.e., $X_{M}\coloneqq[x_{M,1},\ldots, x_{M,M}]'$. 

In this paper, we assume that $\beta$ is identified, which follows from standard assumptions on strict exogeneity, lack of multicollinearity and the existence of finite second moments of $y_{M,m}$ and $x_{M,m}$. (For completeness, see Assumption \ref{assm:ExogenousRegression:IdentificationCondition} and Proposition \ref{prop:ExogenousRegression:Identification} in Appendix). 

We note that equation (\ref{eq:ExogenousRegression:PopulationNetworkRegression_b}) allows for there to be spillovers across the error terms even when dyads $m$ and $m'$ are not adjacent, as long as they are connected through indirect links. By comparison, applied researchers such as \citet{Harmon_et_al-2019} and \citet{Lustig_and_Richmond-2020} (and the estimators of \citet{Aronow_et_al-2015} and \citet{Tabord-Meehan-2019}) consider a variant of the linear regression \eqref{eq:ExogenousRegression:PopulationNetworkRegression_a} under the assumption
	\begin{align}\label{eq:ExogenousRegression:Example:DyadicRegression_b}
		Cov( \varepsilon_{M, d(i,j)}, \varepsilon_{M, d(k,l)}\mid x_{M, d(i,j)}, x_{M, d(k,l)} ) 
		= 0 \hspace{5mm}	\text{unless \quad $\{i,j\}\cap \{k,l\} \neq \emptyset$},
	\end{align}
	with $m = d(i,j)$ representing the dyad between $i$ and $j$. This specific assumption would be equivalent to setting:
	\begin{align}\label{eq:ExogenousRegression:Example:DyadicRegression_c}
		Cov( \varepsilon_{M, m}, \varepsilon_{M, m'}\mid X_{M} )
		= 0 \hspace{5mm}	\text{unless $m$ and $m'$ are adjacent}.
	\end{align}

\subsubsection{Examples}	
Whether to allow indirect spillovers as in (\ref{eq:ExogenousRegression:PopulationNetworkRegression_b}) or not (as in \eqref{eq:ExogenousRegression:Example:DyadicRegression_c}) depends on the researchers' applications. We now present two examples where our approach may be preferable. 

\begin{example}[Gravity Model of Bilateral Trade Flow]\label{ex:GravityModelOfBilateralTradeFlow}
A researcher is studying the trade flow from country $i$ to $j$, with (log) exports from $i$ to $j$ denoted $y_{ij}$. Following the literature, (s)he assumes $y_{ij}$ follows the structural gravity equation (e.g., \cite{Eaton_and_Kortum-2002, Anderson_and_van_Wincoop-2003, Melitz-2003, Helpman_et_al-2008}):
	\begin{align}\label{eq:EndogenousRegression:Example:GravityModel_1}
		y_{ij} = \alpha+\beta z_{ij} + \gamma \sum_{k\neq i}g_{ki} y_{ki} + \eta_{ij},
	\end{align}
where $z_{ij}$ represents a dyadic characteristic of $i$ and $j$, such as the shipping cost, whether both countries are democratic (e.g., \citealp{mansfield00}), or whether both participate in WTO/GATT (e.g., \citealp{gowa}); $\sum_{k}g_{ki} y_{ki}$ is the amount $i$ spends on imports ($g_{ki}$ equals one if country $i$ purchases goods from country $k$ and zero otherwise), and $\eta_{ij}$ captures unobserved heterogeneity pertaining to the trade flow between countries $i$ and $j$.

To see our main point, suppose there are only four countries ($1$, $2$, $3$ and $4$) which trade, where country $1$ exports to country $2$, which in turn exports to country $3$, and country $3$ exports to country $4$. Equation (\ref{eq:EndogenousRegression:Example:GravityModel_1}) then simplifies to: $y_{12} = \alpha + \beta z_{12} + \eta_{12}$, \\ $y_{23} = \alpha + \beta z_{23} + \gamma y_{12} + \eta_{23}$, and henceforth.
	Rearranging these equations implies that the trade flow from country $3$ to $4$ can be written as:
	\begin{align*}
		y_{34}	&= \alpha + \alpha \gamma + \alpha \gamma^2 + \gamma^{2} \beta z_{12} + \gamma \beta z_{23} + \beta z_{34} + \gamma^{2} \eta_{12} + \gamma \eta_{23} + \eta_{34}.
	\end{align*}
	Therefore, $Cov(y_{12}, y_{34}\mid z) = \gamma^{2} Var(\eta_{12}\mid z) \neq 0$, where $z \equiv \{z_{12}, z_{23}, z_{34}\}$. Hence, there can be non-zero correlation between trade flows $y_{12}$ and $y_{34}$ even if they do not have a country in common. This is because an idiosyncratic shock to an upstream country can propagate through the trade network.
\end{example}

\begin{example}[Legislative Voting]\label{ex:LegislativeVoting}
A researcher is interested in whether seating arrangements in legislatures can affect a politician's behavior, $y_i$ (e.g., propensity to vote ``Yes" on a roll-call, as \citet{Harmon_et_al-2019}, or the amount of co-sponsoring, as \cite{Saia, Lowe}, among others). For concreteness, suppose there are four politicians with the seating arrangements given by Figure \ref{fig:ExampleOfPolNetwork}. 

The researchers posit that $i$'s behavior can be influenced by the (average) of its seating neighbors' own voting behavior through a parameter $\gamma$ as follows:
\begin{eqnarray}
	&y_{A} = \alpha + \gamma y_{B} + \eta_{A}, \quad &y_{B} = \alpha + \gamma \frac{y_{A}+y_{C}}{2} + \eta_{B} \\
		&y_{C} = \alpha + \gamma \frac{y_{B}+y_{D}}{2} + \eta_{C}, \quad	&y_{D} = \alpha + \gamma y_{C} + \eta_{D}
\end{eqnarray}
If $\gamma \neq 0$, $A$ is affected by their neighbor $B$, while $B$ is affected by both of its neighbors ($A$ and $C$) and so forth. The researcher is interested in whether neighbors' decisions are more highly correlated than the decisions among non-neighbors. 

Denote $y_{ij}$ as the dyadic outcome of interest (e.g., a measure of correlation between $i$ and $j$'s decisions). Both $y_{AB}$ and $y_{CD}$ involve $y_{B}$ and $y_{C}$, which are themselves a function of $\eta_{B}$ and $\eta_{C}$. Hence, $Cov(y_{AB}, y_{CD}) \neq 0$, even if the two pairs of legislators do not share a common member.

\end{example}

\subsubsection{Estimation}\label{subsubsec:Estimation}
Throughout this paper, we focus on the Ordinary Least Squares (OLS) estimator of $\beta$, denoted by $\hat{\beta}$. Under the assumptions above, we can write
\begin{align}\label{eq:ExogenousRegression:OLSEstimator_2}
	\hat{\beta} - \beta
	&= \Bigg( \sum_{j\in\mathcal{M}_{N}^{}} x_{M, j}x_{M, j}' \Bigg)^{-1}
		\sum_{m\in\mathcal{M}_{N}^{}} x_{M, m} \varepsilon_{M, m}.
\end{align}
It is straightforward to verify that $\hat{\beta}$ is unbiased for $\beta$ under our identification conditions (Assumption \ref{assm:ExogenousRegression:IdentificationCondition}). However, a consistency result is by no means trivial due to the dependence along the network which induces a complex form of cross-sectional dependence, hindering a na\"{i}ve application of the standard theory for independently and identically distributed ($i.i.d.$) random vectors. This is clarified in Section \ref{subsec:NetworkDependentProcesses}. Before that, we provide a heuristic overview of the main point of our paper in this setting.

\subsection{Outline of Our Procedure}

\subsubsection{Inference}

Inference about $\beta$ is based on a normal approximation of the distribution of $\hat{\beta}$ around $\beta$. We focus on hypothesis testing conducted using the expression:
\begin{align}\label{eq:TestStatistics}
	\big(\widehat{Var}(\hat{\beta})\big)^{-\frac{1}{2}} (\hat{\beta}-\beta),
\end{align}
where $\widehat{Var}(\hat{\beta})$ is a consistent estimator of the asymptotic variance of $\hat{\beta}$. Our main result in Section \ref{subsec:NetworkHACVarianceEstimation} is providing such an appropriate estimator, which takes the form:
{\footnotesize
\begin{align}\label{eq:ExogenousRegression:NetworkVarianceEstimator}
	\widehat{Var}(\hat{\beta})
	\coloneqq \Bigg( \sum_{k\in\mathcal{M}_{N}^{}} x_{k}x_{k}' \Bigg)^{-1}
			\Bigg( \sum_{m\in\mathcal{M}_{N}^{}} \sum_{m'\in\mathcal{M}_{N}^{}} \kappa_{m,m'} h_{m,m'} \hat{\varepsilon}_{m}\hat{\varepsilon}_{m'} x_{m}x_{m'}' \Bigg)
			\Bigg( \sum_{k\in\mathcal{M}_{N}^{}} x_{k}x_{k}' \Bigg)^{-1},
\end{align}
}
\noindent where $\kappa_{m,m'}$ is an appropriate kernel function that will formally be defined in Section \ref{subsec:NetworkHACVarianceEstimation}; $h_{m,m'}$ represents an indicator function that takes one if dyads $m$ and $m'$ are connected and zero otherwise; and $\hat{\varepsilon}_{m} \coloneqq y_{m}- x_{m}' \hat{\beta}$. 

This paper derives conditions under which $\widehat{Var}(\hat{\beta})$ is consistent for the asymptotic variance of $\hat{\beta}$. Before doing so, let us compare the variance estimator (\ref{eq:ExogenousRegression:NetworkVarianceEstimator}) with an often used estimator based on one-way clustering of dyad groupings. 

\begin{remark}[Dyadic-Robust Variance Estimator]\label{ex:ExogenousRegression:Dyadic-RobustVarianceEstimator}
An increasing number of applied researchers, such as \citet{Harmon_et_al-2019} and \citet{Lustig_and_Richmond-2020}, estimate model (\ref{eq:ExogenousRegression:PopulationNetworkRegression_a}) and conduct inference, using the following dyadic-robust variance estimators proposed by \citet{Aronow_et_al-2015} and \citet{Tabord-Meehan-2019}:
	{\small
	\begin{align}\label{eq:ExogenousRegression:Example:Dyadic-RobustVarianceEstimator:Estimator_1}
		\widehat{Var}(\hat{\beta})
		\coloneqq \Bigg( \sum_{k\in\mathcal{M}_{N}^{}} x_{k}x_{k}' \Bigg)^{-1}
				\Bigg( \sum_{m\in\mathcal{M}_{N}^{}} \sum_{m'\in\mathcal{M}_{N}^{}} \mathbbm{1}_{m,m'} \hat{\varepsilon}_{m}\hat{\varepsilon}_{m'} x_{m}x_{m'}' \Bigg)
				\Bigg( \sum_{k\in\mathcal{M}_{N}^{}} x_{k}x_{k}' \Bigg)^{-1},
	\end{align}}
	where $\mathbbm{1}_{m,m'}$ equals one if dyads $m$ and $m'$ are adjacent and zero otherwise. 
	
Note that the use of the dyadic-robust variance estimator sets cases in which two dyads are not adjacent, but connected, to zero. Meanwhile, our estimator (\ref{eq:ExogenousRegression:NetworkVarianceEstimator}) accounts for network spillovers by accommodating the correlation across both adjacent and connected dyads.\footnote{See Definition \ref{defn:AdjacentAndConnectedDyads} and the subsequent discussion. The choice of kernel and lag-truncation is discussed in Section \ref{subsec:NetworkHACVarianceEstimation}.} As our examples above suggest,  the structure of the variance estimator  (\ref{eq:ExogenousRegression:Example:Dyadic-RobustVarianceEstimator:Estimator_1}) may not be compatible with indirect spillovers in some settings, which should assume the specification (\ref{eq:ExogenousRegression:PopulationNetworkRegression_b}) instead. This suggests that the dyadic-robust variance estimator may be inconsistent when there is cross-sectional dependence along a network over dyads: i.e., when non-adjacent dyads can still affect the correlation structure and outcomes of dyad $m$.\footnote{Clustering estimators may be inappropriate when the correlation structure has network spillovers as in (\ref{eq:ExogenousRegression:PopulationNetworkRegression_b}), because each agent has a complex (i.e., non-separable) structure of connections, reflected in a non-separable network across dyads. If the network model features positive spillovers, then the dyadic-robust variance estimator will likely underestimate the true variance, leading to conservative hypothesis testing. Meanwhile, it is likely to overstate the true variance when there are negative spillovers. We expand on this point in our numerical simulations.} This conjecture is formally proven in Corollary \ref{coro:ExogenousRegression:NetworkHACEstimators:InconsistencyResult} and illustrated in Monte Carlo simulations in Section \ref{sec:MonteCarloSimulation}. We note that this is a feature of applying such dyadic-robust variance estimators to network spillovers, and not a feature of those estimators per se.
\end{remark}

In order to validate the test statistic (\ref{eq:TestStatistics}) for inference, this paper establishes the consistency and asymptotic normality of $\hat{\beta}$ and develops a consistent network-robust estimator for the asymptotic variance when the correlation structure is given by \eqref{eq:ExogenousRegression:PopulationNetworkRegression_b}.

\subsubsection{A Step-by-Step Guide on Implementation}\label{sec:guide}

Researchers interested in the linear specification (\ref{eq:ExogenousRegression:PopulationNetworkRegression_a}) may find the following guidelines useful when deciding whether to use our proposed estimator (\ref{eq:ExogenousRegression:NetworkVarianceEstimator}).
\begin{enumerate}
\item
The researcher should first ask whether spillovers from indirectly connected dyads are likely to be present (and not decay immediately) in their set-up: i.e., is equation (\ref{eq:ExogenousRegression:PopulationNetworkRegression_b}) a more appropriate assumption than equation (\ref{eq:ExogenousRegression:Example:DyadicRegression_c})? 

While this depends on the specific application, Examples \ref{ex:GravityModelOfBilateralTradeFlow}-\ref{ex:LegislativeVoting} illustrate models where that is likely to be the case. And condition (\ref{eq:ExogenousRegression:NetworkHACEstimators:InconsistencyResult}) provides a notion of how much persistence is needed for a bias to appear. As we show below, these insights are robust to decaying spillover effects (see Corollary \ref{coro:ExogenousRegression:NetworkHACEstimators:InconsistencyResult}, Example \ref{example:MaximumAdmissibleBiasInTheDyadic-RobustVarianceEstimator}, and the associated simulation results).

\item
If such spillovers of unobservables are likely to exist, are they governed by an exogenous and observable network (Assumption \ref{assum:exog_networks}), such as physical, geographical or social (e.g., family ties)?\footnote{Alternatively, one would have to consider recovering unobserved exogenous networks (see \cite{paula19} for an example with panel data) or modeling endogenous network formation, possibly with unobserved links (see \cite{battaglini22, canen22} for two examples in political economy).} If so, the proposed estimator is appropriate under regularity conditions.
\item
One can implement our estimator by: (i) choosing a kernel (e.g., rectangular, see Section \ref{subsec:NetworkHACVarianceEstimation}), (ii) setting the lag-truncation, $b_M$ (either by a known value, or adaptively by $b_M = 2 \log(M)/\log(\max(average~degree, 1.05))$, where $M$ is the number of dyads and we use the average degree of the dyadic network, (iii) plugging-in those choices into equation (\ref{eq:ExogenousRegression:NetworkVarianceEstimator}). \medskip

This estimator is consistent under regularity conditions, even when spillovers decay, and shows good finite-sample properties in the simulations below.
\end{enumerate}


\section{Theoretical Results}\label{sec:MainResults}

\subsection{Network Dependent Processes}\label{subsec:NetworkDependentProcesses}

Let $Y_{M,m}$ be a random vector defined as
\begin{align}\label{eq:ExogenousRegression:DefinitonOfY}
	Y_{M,m}
	\coloneqq x_{M, m} \varepsilon_{M, m},\footnotemark
\end{align}
\footnotetext{By construction, the collection of $Y_{M,m}$'s constitutes a triangular array of random vectors.}
and denote $\mathcal{C}_{M} \coloneqq \{x_{M, m}\}_{m\in\mathcal{M}_{N}}$.\footnote{For the case of stochastic networks, it is defined to include information about the network topology as well as the collection of the dyad-specific attributes $\{x_{M, m}\}_{m\in\mathcal{M}_{N}}$. See \citet{Kojevnikov_et_al-2021}.} From equations \eqref{eq:ExogenousRegression:OLSEstimator_2} and \eqref{eq:ExogenousRegression:DefinitonOfY} we can write
\begin{align}\label{eq:ExogenousRegression:ExpressionOfBetaHat}
	\hat{\beta}-\beta 
	= \Bigg( \frac{1}{M} \sum_{j\in\mathcal{M}_{N}^{}} x_{M, j}x_{M, j}' \Bigg)^{-1} \frac{1}{M} \sum_{m\in\mathcal{M}_{N}} Y_{M,m}.
\end{align}

Our interest lies in proving the asymptotic properties of $\hat{\beta}$ taking advantage of the expression \eqref{eq:ExogenousRegression:ExpressionOfBetaHat}. However, the presence of $\varepsilon_{M,m}$ in $Y_{M,m}$, which is allowed to be correlated along the network over active dyads, renders our approach nonstandard. The canonical results about independently and identically distributed ($i.i.d.$) random variables are by no means applicable due to the complex form of dependence inherent to $\varepsilon_{M,m}$. Dyadic-robust estimators (e.g., \citet{Aronow_et_al-2015} and \citet{Tabord-Meehan-2019}) only account for correlations between adjacent dyads. This may be appropriate in some settings, but not when there may be indirect spillovers, as in specification \eqref{eq:ExogenousRegression:PopulationNetworkRegression_a} -- \eqref{eq:ExogenousRegression:PopulationNetworkRegression_b}. Thence, the right hand side of \eqref{eq:ExogenousRegression:ExpressionOfBetaHat} cannot simply be embedded in other widely-studied variants of dependent random vectors, including spatially-correlated and time-series data.

The main insight of this paper is that the spillovers across connected -- even if not adjacent -- dyads, can be rewritten as the dependence of $Y_{M,m}$'s along the network of active dyads (hereby, referred to as the ``network''). This transformation allows us to embrace such complex cross-sectional dependence and appropriately rewrite the problem so that recent results on network dependent random variables \citep{Kojevnikov_et_al-2021} can be applied. Our results then allow for correlations of random variables given a common shock as explicitly dictated by distances between them on the dyadic network. Our notion of dependence restricts the covariance between the random variables for any two sets of dyads $A$ and $B$, $Y_{M,A}$ and $Y_{M,B}$, which are at a distance $s$ from one another, to be bounded and to decrease as $s$ goes to infinity. Indeed, this covariance is controlled by a sequence of (random) coefficients $\theta_{M,s}$, which capture the strength of covariation net of scaling, analogous to correlation coefficients. As the minimum distance, $s$, between $A$ and $B$ grows, the dependence $\{\theta_{M,s}\}$ between $Y_{M,A}$ and $Y_{M,B}$, tends to zero. A formal description is provided in Appendix \ref{appdx:MathmaticalSetup}.



\subsection{Definitions}

As will become transparent shortly, asymptotic theories for $\hat{\beta}$ rest on tradeoffs between the correlation of the network-dependent random vectors (i.e., the dependence coefficients) and the denseness of the underlying network. To measure the denseness, we first define two concepts of neighborhoods: for each $m\in\mathcal{M}_{N}$ and $s\in\mathbbm{N}\cup \{0\}$,
\begin{align*}
	&\mathcal{M}_{N}(m; s) \coloneqq \{ m'\in\mathcal{M}_{N}: \rho_{M}(m,m') \leq s \}, \\
	&\mathcal{M}_{N}^{\partial}(m; s) \coloneqq \{ m'\in\mathcal{M}_{N}: \rho_{M}(m,m') = s \},
\end{align*}
where $\rho_{M}(m,m')$ denotes the geodesic distance between dyads $m$ and $m'$.\footnote{Recall that we define the geodesic distance between two connected dyads $m$ and $m'$ to be the smallest number of adjacent dyads between them.} The former set collects all the $m$'s neighbors whose distance from $m$ is no more than $s$ (which we call a neighborhood), whilst the latter registers all the $m$'s neighbors whose distance from $m$ is exactly $s$ (which we call a neighborhood shell). 

Next, we define two types of density measures of a network: for $k, r >0$, 
\begin{align}\label{eq:DefinitionsOfNetworkDensities}
	\begin{split}
	&\Delta_{M}^{}(s,r; k) \coloneqq \frac{1}{M} \sum_{m\in\mathcal{M}_{N}} \max_{m'\in\mathcal{M}_{N}^{\partial}(m;s)} |\mathcal{M}_{N}(m;r)\backslash \mathcal{M}_{N}(m'; s-1)|^{k}, \\
	&\delta_{M}^{\partial}(s; k) \coloneqq \frac{1}{M} \sum_{m\in\mathcal{M}_{N}} |\mathcal{M}_{N}^{\partial}(m;s)|^{k},
	\end{split}
\end{align}
where it is assumed that $\mathcal{M}_{N}(m'; -1) = \emptyset$. The former measure gauges the denseness of a network in terms of the average size of a version of the  neighborhood. The latter expresses the denseness of a network as the average size of the neighborhood shell raised by $k$. \citet{Kojevnikov_et_al-2021} show that controlling the asymptotic behavior of an appropriate composite of these two measures (denoted by $c_M$ and defined in Assumption \ref{assm:Kojevnikov_et_al-2021:Assm4.1}) is sufficient for the Law of Large Numbers (LLN) and Central Limit Theorem (CLT) of the network dependent random variables (Condition ND).

It is worth noting that there are two different units at play: the number of sampling units (i.e., individuals), $N$, and the number of dyads, $M$. To bridge these two units, we assume that the number of active dyads also goes to infinity as the number of sampling units is taken to infinity, eliminating the possibility of extremely sparse networks among individuals. This is empirically relevant. For example, in international trade, the entry of a new country/firm to a market will most likely increase the number of trade flows in the economy; in political economy, the more members of parliaments (MEPs) there are, the more pairs of the MEPs sitting next to each other there will be (see Section \ref{sec:EmpiricalIllustration}).\footnote{This assumption is similar in spirit to Assumption 2.3 of \citet{Tabord-Meehan-2019} in which the minimum degree is assumed to grow at some constant rate relative to the number of individuals. It is milder than Assumption 2.3 of \citet{Tabord-Meehan-2019} since the latter does not allow any individual to be isolated, while Assumption \ref{assm:ExogenousRegression:Denseness} merely constrains the average degree. Similarly, this assumption is weaker than the assumption that the maximum degree in a network is bounded even when $N\rightarrow\infty$  (e.g., \citet{Penrose_and_Yukich-2003} and \citet{de_Paula_et_al-2018}).} 

\begin{assumption}\label{assm:ExogenousRegression:Denseness}
	$M \to \infty$ as $N \to \infty$.
\end{assumption}

\subsection{Asymptotic Properties of $\hat{\beta}$}\label{subsec:LimitTheoremsforNetworkDependentProcesses}

We make use of the following two regularity assumptions for the proof of consistency of $\hat \beta$ for $\beta$ (Theorem \ref{theo:ConsistencyofBetaHat}) and to derive its asymptotic distribution (Theorem \ref{theo:ExogenousRegression:CLT}).\footnote{These assumptions are required for Theorem \ref{theo:ExogenousRegression:CLT}, but as usual, the proof of consistency (Theorem \ref{theo:ConsistencyofBetaHat})can be derived under weaker conditions. (See Assumptions \ref{assm:ExogenousRegression:ConditionalFiniteMoment1} and \ref{assm:Kojevnikov_et_al-2021:Assm3.2} and their associated discussion, in Appendix).}   

\begin{assumption}[Conditional Finite Moment of $\varepsilon_{m}$]\label{assm:ExogenousRegression:ConditionalFiniteMoment2}
	There exists $p>4$ such that \\ $\sup_{N\geq 1} \max_{m\in\mathcal{M}_{N}} E\big[ | \varepsilon_{m} |^{p} \mid \mathcal{C}_{M} \big] < \infty$ a.s.
\end{assumption}

\begin{assumption}[\citet{Kojevnikov_et_al-2021}, Assumption 3.4]\label{assm:Kojevnikov_et_al-2021:Assm3.4}
	There exists a positive sequence $r_{M}\rightarrow \infty$ such that for $k=1,2$,
	\begin{align*}
		& \frac{M^{2}\theta_{M,r_{M}}^{1-1/p}}{\sigma_{M}} \stackrel{a.s.}{\rightarrow} 0, \quad \text{and} \quad \frac{M}{\sigma_{M}^{2+k}} \sum_{s\geq 0} c_{n}(s,r_{M};k)\theta_{M,s}^{1-\frac{2+k}{p}} \stackrel{a.s.}{\rightarrow} 0,
	\end{align*}
	as $M\rightarrow \infty$, where $p>4$ is the same as the in Assumption \ref{assm:ExogenousRegression:ConditionalFiniteMoment2}.
\end{assumption}

Assumption \ref{assm:ExogenousRegression:ConditionalFiniteMoment2} requires that the errors are not too large once conditioned on common shocks. Together with the standard full rank assumption that guarantees identification of $\beta$\footnote{See Assumption \ref{assm:ExogenousRegression:IdentificationCondition}(c) in Appendix, for completeness.}, this implies Assumption 3.1 of \citet{Kojevnikov_et_al-2021} for each $u$-th element of $Y_{M,m}^{}$, denoted by $Y_{M,m}^{u}$ with $u\in\{1,\ldots, K\}$.   

Assumption \ref{assm:Kojevnikov_et_al-2021:Assm3.4} is a condition that controls the tradeoff between the denseness of the underlying network and the covariability of the random vectors. If the network becomes dense, then the dependence of the associated random variables has to decay much faster. This requirement is consistent with the typical assumption in the spatial econometrics literature upon which the empirical analyses of the aforementioned examples (i.e., international trade, legislative voting, and exchange rates) are based, because it embodies the idea that spillovers decay as they propagate farther (see, e.g., \citet{Kelejian_and_Prucha-2010}). For instance, \citet{Acemoglu_et_al-2015} assume that network spillovers are zero if agents are sufficiently distantly connected on a geographical network. This decay in the propagation of spillovers through the terms $\theta_{M,\cdot}$. This assumption may be violated for very dense networks with low decay of spillovers. 

\subsubsection{Consistency}

\begin{theorem}[Consistency of $\hat{\beta}$]\label{theo:ConsistencyofBetaHat}
	Under Assumptions \ref{assm:ExogenousRegression:Denseness}-\ref{assm:Kojevnikov_et_al-2021:Assm3.4}, $\| \hat{\beta} - \beta \|_{2} \stackrel{p}{\rightarrow} 0$
	as $N\rightarrow\infty$.
\end{theorem}
\begin{proof}
	See Appendix \ref{appndx:ProofOfMainTheorems:ExogenousLinearRegression:ConsistencyOfBetaHat}. 
\end{proof}

Theorem \ref{theo:ConsistencyofBetaHat} establishes the consistency of $\hat{\beta}$ under the scenario where the number of sampling units $N$ goes to infinity. When Assumption \ref{assm:ExogenousRegression:Denseness} is dropped, the result continues to hold in terms of the number of active dyads $M$. 


\subsubsection{Asymptotic Normality}

Let $S_{M}^{} \coloneqq \sum_{m\in\mathcal{M}_{N}} Y_{M,m}^{}$, which is present in $\hat \beta$ in equation (\ref{eq:ExogenousRegression:ExpressionOfBetaHat}). Let $S_{M}^{u}$ be the $u$-th entry of $S_{M}$ for $u\in\{1,\ldots,K\}$ and denote the unconditional variance of $S_{M}^{u}$ by $\tau_{M}^{2} \coloneqq Var (S_{M}^{u} )$. Since $S_{M}^{u}$ is not a sum of independent variables, its variance cannot be simply expressed as a sum of the variances of $Y_{M,m}$. We thus need to explicitly take into account covariance between the random variables $\{Y_{M,m}^{u}\}_{m\in\mathcal{M}_{N}}$. We study the CLT for the normalized sum of $Y_{M,m}^{u}$, which is given by $\frac{S_{M}^{u}}{\tau_{M}}$.

Assumptions \ref{assm:ExogenousRegression:ConditionalFiniteMoment2} and \ref{assm:Kojevnikov_et_al-2021:Assm3.4} are written in terms of conditional expectations, whereas we are interested in the unconditional distribution of $\frac{S_{M}}{\tau_{M}}$. Assumption \ref{assm:AsymptRateOfVariance} below bridges the conditional and unconditional variances. The final assumption for the asymptotic normality result is a standard regularity condition guaranteeing that the asymptotic variance is well-defined. This assumption follows from both matrices in the expression being well-defined (e.g., $x_{M,k}$ having bounded support).\footnote{Further note that Theorem \ref{theo:ExogenousRegression:CLT} is proved under a weaker condition than Assumption \ref{assm:AsymptRateOfVariance}.}

\begin{assumption}[Growth Rates of Variances]\label{assm:AsymptRateOfVariance}
	$\frac{\sigma_{M}^{2}}{\tau_{M}^{2}} \stackrel{a.s.}{\rightarrow} 1$ as $N\rightarrow\infty$.
\end{assumption}

\begin{assumption}\label{assm:ExogenousRegression:BoundedSupport}
	(a) For all $N\geq 1$, $\{ x_{M, m} \}_{m\in\mathcal{M}_{N}}$ have uniformly bounded support.\\
	(b) $\lim_{N\rightarrow\infty} \Big(\frac{1}{M} \sum_{k\in\mathcal{M}} E \big[ x_{M,k} x_{M,k}' \big]\Big)$ is positive definite.\\
	(c)	$\lim_{N\rightarrow\infty} \frac{N}{M^{2}} \sum_{m\in\mathcal{M}_{N}^{}} \sum_{m'\in\mathcal{M}_{N}^{}} E\big[ \varepsilon_{M, m}\varepsilon_{M, m'} x_{M, m}x_{M, m'}' \big]$ exists with finite elements.
\end{assumption}


Under these assumptions, the asymptotic distribution of $\hat{\beta}$ is given by:

\begin{theorem}[Asymptotic Normality of $\hat \beta$]\label{theo:ExogenousRegression:CLT}
	Under Assumptions \ref{assm:ExogenousRegression:Denseness}-\ref{assm:ExogenousRegression:BoundedSupport}, \\ $\sqrt{N} \big( \hat{\beta}-\beta \big) \stackrel{d}{\rightarrow} \mathcal{N}(0, AVar(\hat{\beta}))$ as $N\rightarrow\infty$,
		where 
	{\footnotesize
	\begin{align}\label{eq:ExogenousRegression:AsymptoticVariance}
		AVar(\hat{\beta})
		= \lim_{N\rightarrow\infty} N \Big( \sum_{k\in\mathcal{M}_{N}^{}} E\big[ x_{M, k}x_{M, k}' \big]\Big)^{-1}
				\Big( \sum_{m\in\mathcal{M}_{N}^{}} \sum_{m'\in\mathcal{M}_{N}^{}} E\big[ \varepsilon_{M, m}\varepsilon_{M, m'} x_{M, m}x_{M, m'}' \big] \Big)  
				\Big( \sum_{k\in\mathcal{M}_{N}^{}} E\big[ x_{M, k}x_{M, k}' \big] \Big)^{-1},
	\end{align}
	} 
	which is positive semidefinite with finite elements.
\end{theorem}
\begin{proof}
	See Appendix \ref{appndx:ProofOfMainTheorems:ExogenousLinearRegression:AsymptoticNormalityOfBetaHat}. 
\end{proof}



\subsection{Consistent Estimation of the Asymptotic Variance of $\hat{\beta}$ under Network Spillovers}\label{subsec:NetworkHACVarianceEstimation}

Our objective is to consistently estimate $AVar(\hat{\beta})$ defined in Theorem \ref{theo:ExogenousRegression:CLT}. As errors are mean zero, $Y_{M,m}$ is centered, i.e., $E\big[ Y_{M,m} \big]=0$ for each $m\in \mathcal{M}_{N}$. 

\subsubsection{The Estimator}

The proposed estimator is a type of kernel estimator. Let $b_{M}$ denote the bandwidth, or the lag truncation (its choice is described in Section \ref{sec:bandwidth} below) and $\omega:\overline{\mathbbm{R}}\rightarrow[-1,1]$ a kernel function such that $\omega(0)=1$, $\omega(z)=0$ whenever $|z|>1$, and $\omega(z)=\omega(-z)$ for all $z\in\overline{\mathbbm{R}}$. The feasible variance estimator of interest is
{\small
\begin{align}\label{eq:ExogenousRegression:NetworkHACEstimation:VarHatBetaHat}
	\widehat{Var}(\hat{\beta}) = \Big( \frac{1}{M}\sum_{k\in\mathcal{M}_{N}} x_{M,k}x_{M,k}' \Big)^{-1} 
			\Big( \frac{1}{M^{2}} \sum_{s\geq 0}\sum_{m\in\mathcal{M}_{N}} \sum_{m'\in\mathcal{M}_{N}^{\partial}(m;s)} \omega_{M}(s) \hat{Y}_{M,m}\hat{Y}_{M,m'}' \Big)
			\Big( \frac{1}{M}\sum_{k\in\mathcal{M}_{N}} x_{M,k}x_{M,k}' \Big)^{-1},
\end{align}}
with $\omega_{M}(s) \coloneqq \omega\big(\frac{s}{b_{M}}\big)$ for all $s \geq 0$ and $\hat{Y}_{M,m} \coloneqq x_{M, m} \hat{\varepsilon}_{M, m}$, where $\hat{\varepsilon}_{M, m}\coloneqq y_{M, m} - x_{M, m}' \hat{\beta}$.

\subsubsection{Choice of Lag Truncation, $b_M$:}\label{sec:bandwidth}

There are two approaches for the choice of the associated lag truncation parameter. First, the researcher may already know (or is willing to impose) the truncation, perhaps due to a theoretical/institutional motivation. For instance, \cite{Acemoglu_et_al-2015} set the lag to two in a related problem. Then, the thought exercise is that this choice will adapt as $M \to \infty$ according to the assumptions below. Alternatively, the researcher could use a data-driven choice. Assumption \ref{assm:Kojevnikov_et_al-2021:Assm4.1} (c) below suggests it should depend on both the sample size and the network topology, including the average degree of the dyadic network. One such selection rule is suggested in \citet{Kojevnikov_et_al-2021} based on their proofs: $b_M = 2 \log(M)/\log(\max(average~degree, 1.05))$. 

\subsection{Consistency of the Proposed Estimator}

The consistency of the variance estimator requires two sets of additional assumptions. The first set is Assumption 4.1 of \citet{Kojevnikov_et_al-2021}, but stated here in terms of the network over dyads.

\begin{assumption}[\citet{Kojevnikov_et_al-2021}, Assumption 4.1]\label{assm:Kojevnikov_et_al-2021:Assm4.1}
	There exists $p>4$ such that (a) $\sup_{N\geq 1} \max_{m\in\mathcal{M}_{N}} E\big[ | \varepsilon_{m} |^{p} \mid \mathcal{C}_{M} \big] < \infty$ $a.s.$; \\ (b) $\lim_{M\rightarrow\infty} \sum_{s\geq 1} |\omega_{M}(s)-1| \delta_{M}^{\partial}(s)\theta_{M,s}^{1-\frac{2}{p}} = 0$ $a.s.$; \\ (c) $\lim_{M\rightarrow \infty}\frac{1}{M} \sum_{s\geq 0} c_{M}(s,b_{M};2) \theta_{M,s}^{1-\frac{4}{p}} = 0$ $a.s.$, where 
	\begin{align*}
		c_{M}(s,r; k) \coloneqq \inf_{\alpha>1} \big(\Delta_{M}^{}(s,r; k\alpha)\big)^{\frac{1}{\alpha}} \Big(\delta_{M}^{\partial}\big(s; \frac{\alpha}{\alpha-1}\big)\Big)^{\frac{\alpha-1}{\alpha}}.
	\end{align*}
\end{assumption}

Assumption \ref{assm:Kojevnikov_et_al-2021:Assm4.1} (a) is a stronger counterpart to Assumption \ref{assm:ExogenousRegression:ConditionalFiniteMoment2}, as it requires that a higher-order (i.e., higher than fourth order) conditional moment be well-defined. Assumption (b) posits a tradeoff between the kernel function, the denseness of a network and the dependence coefficients. Specifically, the kernel function $\omega_{M}$ is required to converge to one sufficiently fast. \citet{Kojevnikov_et_al-2021} demonstrate primitive conditions under which this requirement is fulfilled (Proposition 4.2). Assumption (c) requires that the correlation coefficients decay much faster relative to the denseness of the network. This is satisfied in the suggested choice for $b_M$ above. 

Another set of conditions restricts the denseness of the network, ruling out the situation where the network becomes progressively dense: most notably, the case where every single individual unit is directly linked to every other individual. 

\begin{assumption}\label{assm:ExogenousRegression:Sparseness} \phantom{}
	(a) $\sup_{N\geq 1} \sum_{s\geq 0} \delta_{M}^{\partial}(s; 1)<\infty$; (b) $\lim_{M\rightarrow\infty}\frac{1}{M^{}} \sum_{s\geq 0} c_{M}(s,b_{M};2) = 0$.
\end{assumption}

The following theorem is the main theoretical contribution of this paper.

\begin{theorem}[Consistency of the Network-Robust Variance Estimator]\label{thm:ExogenousRegression:ConsistencyOfTheNetwork-HACVarianceEstimator}
Under the conditions for Theorem \ref{theo:ExogenousRegression:CLT}, and Assumptions \ref{assm:Kojevnikov_et_al-2021:Assm4.1} and \ref{assm:ExogenousRegression:Sparseness}, $\big\| N \widehat{Var}(\hat{\beta}) - Var(\hat{\beta}) \big\|_{F} \stackrel{p}{\rightarrow} 0$ as $N\rightarrow \infty$,
	where $\|\cdot\|_{F}$ indicates the Frobenius norm.
\end{theorem}
\begin{proof}
	See Appendix \ref{appndx:subsec:ProofOfMainTheorems:ExogenousLinearRegression:ConsistencyOfVHat}. 
\end{proof}

Theorem \ref{thm:ExogenousRegression:ConsistencyOfTheNetwork-HACVarianceEstimator} establishes the consistency of our proposed variance estimator accounting for network spillovers across dyads in the sense of the Frobenius norm. 

\subsection{When to Use the Proposed Estimator and the Role of Decaying Spillover Effects}

It follows from Theorem \ref{thm:ExogenousRegression:ConsistencyOfTheNetwork-HACVarianceEstimator} that the dyadic-robust variance estimator (\ref{eq:ExogenousRegression:Example:Dyadic-RobustVarianceEstimator:Estimator_1}) is inconsistent for the true variance when the underlying network involves a non-negligible degree of far-away correlations, as suggested in the examples of the previous section.\footnote{If the network is such that there are only adjacent dyads (i.e., when equation (\ref{eq:ExogenousRegression:Example:DyadicRegression_c}) holds), then the result above implies consistency of this estimator for dyadic dependence. By comparison, Lemma 1 of \citet{Aronow_et_al-2015} and Proposition 3.1 and 3.2 of \citet{Tabord-Meehan-2019} also provide consistent variance estimators for the dyadic dependence case without higher-order network spillovers. However, these results and ours do not subsume one another. Indeed, their estimators can accommodate flexible dependence within clusters of dyads that share common units, while we assume that the network of spillovers is observed even if there are only adjacent connections.} Specifically, the following corollary states that the dyadic-robust variance estimator of \citet{Aronow_et_al-2015} may not necessarily be consistent when it is na\"ively applied to the network-regression model with non-zero correlations beyond direct neighbors. 

\begin{corollary}[Inconsistency of Dyadic-Robust Estimators with Network Spillovers]\label{coro:ExogenousRegression:NetworkHACEstimators:InconsistencyResult}
	Suppose that the assumptions required in Theorem \ref{thm:ExogenousRegression:ConsistencyOfTheNetwork-HACVarianceEstimator} hold. Assume, in addition, that
	\begin{align}\label{eq:ExogenousRegression:NetworkHACEstimators:InconsistencyResult}
		\inf_{N\geq 1} \frac{1}{M} \bigg\| \sum_{s\geq 2} \sum_{m\in\mathcal{M}_{N}} \sum_{m'\in\mathcal{M}_{N}^{\partial}(m;s)} E\big[ \varepsilon_{M, m} \varepsilon_{M, m'} x_{M, m} x_{M, m'}' \big] \bigg\|_{F} > 0.
	\end{align}
	Then, the dyadic-robust estimator (\ref{eq:ExogenousRegression:Example:Dyadic-RobustVarianceEstimator:Estimator_1}) applied to the network-regression model (\ref{eq:ExogenousRegression:PopulationNetworkRegression_a}) and (\ref{eq:ExogenousRegression:PopulationNetworkRegression_b}) is inconsistent.
\end{corollary}
\begin{proof}
	See Appendix \ref{appndx:subsec:ProofOfMainTheorems:ExogenousLinearRegression:InConsistencyOfVHatDyad}. 
\end{proof}

The added condition \eqref{eq:ExogenousRegression:NetworkHACEstimators:InconsistencyResult} in Corollary \ref{coro:ExogenousRegression:NetworkHACEstimators:InconsistencyResult} pertains to both the network configuration of active dyads and the regression variables. It represents a setting where the spillovers from far-away neighbors are non-negligible even when $N$ is large. For instance, \eqref{eq:ExogenousRegression:NetworkHACEstimators:InconsistencyResult} can hold even if there are not many neighbors, as long as the covariances between the error terms are sufficiently large. This builds on \cite{pinto} that inference with dyadic data may be biased if one only partially accounts for such spillovers. On the other hand, if far-away neighbors in the network have only a negligible effect on the cross-sectional dependence, the dyadic-robust variance estimator remains a good approximation for the asymptotic variance of linear dyadic data models with network spillovers across dyads. These insights are investigated further in Section \ref{sec:MonteCarloSimulation} using numerical simulations. 

These observations extend to settings where the spillovers decay along the network (i.e., when the correlation along unobservables decreases with the geodesic distance among dyads). Indeed, our estimator already accounts for such decay through the indirect covariances in its expression (\ref{eq:ExogenousRegression:NetworkHACEstimation:VarHatBetaHat}). When such spillovers propagate and decay is not too high, then condition (\ref{eq:ExogenousRegression:NetworkHACEstimators:InconsistencyResult}) is satisfied, as such spillovers are non-neglible. On the other hand, if they decay at a very high rate (in the limit, a $100\%$ decay from adjacent to connected dyads), then our estimator will become very similar to the dyadic-robust variance estimator.

However, some researchers may be willing to tolerate some asymptotic bias to still implement the dyadic-robust estimator even with some asymptotic bias. In such cases, when should they prefer our proposed estimator? While a general answer is complex because the bias depends on both the strength of indirect spillovers and on the network configuration, we provide an intuitive criterion which is verified using: (i) a simple, yet informative, analytical example where spillovers decay exponentially with distance along the network (shown below), and (ii) extensive Monte Carlo simulations with decay and different network configurations (shown in Section \ref{sec:MonteCarloSimulation}). They illustrate that researchers should opt for the network-robust estimator above when they have a low tolerance for bias, when there are persistent indirect spillovers and/or when the network is not too sparse.

\begin{example}[Maximum Admissible Bias in the Dyadic-Robust Variance Estimator]\label{example:MaximumAdmissibleBiasInTheDyadic-RobustVarianceEstimator}
Suppose that spillovers decay exponentially with distance along the network: i.e., \\$E\big[ \varepsilon_{M, m} \varepsilon_{M, m'} x_{M, m} x_{M, m'} \big] = \gamma^{s}$, where $s$ is the geodesic distance between dyads $m$ and $m'$, $\gamma\in (0,1)$ and $S$ the longest path in the network. 

Let $B>0$ denote the maximum tolerance for condition (\ref{eq:ExogenousRegression:NetworkHACEstimators:InconsistencyResult}) that the researcher is willing to allow when using the dyadic-robust variance estimator (\ref{eq:ExogenousRegression:Example:Dyadic-RobustVarianceEstimator:Estimator_1}). Then, a sufficient condition for the researcher to prefer the proposed estimator (\ref{eq:ExogenousRegression:NetworkHACEstimation:VarHatBetaHat}) over (\ref{eq:ExogenousRegression:Example:Dyadic-RobustVarianceEstimator:Estimator_1}) is that the decay rate $\gamma$ is higher than a threshold $\bar{\gamma}$, where
	\begin{align*}
		\ln \bar{\gamma} = \frac{2}{S+2} \bigg\{ \ln B - \ln (S-1) - \frac{1}{S-1} \sum_{s\geq 2} \ln \delta_{M}^{\partial}(s) \bigg\}.
	\end{align*}
\end{example}
\begin{proof}
	See Appendix \ref{appndx:example_proof}. 
\end{proof}

When the tolerated bias is small ($B \to 0$), dependence is large enough and does not decay too fast (large $\gamma$), or the network is more dense, our approach is preferable because it provides a consistent estimator even with non-negligible spillovers. Since the network is observed, $S$ and the last term are estimable and can be used for such diagnostics. Appendix \ref{appdx:AdditionalResults:S2gamma0.2} provides further discussions, while the next section presents for our baseline results and discusses simulation exercises.

\newpage
\section{Monte Carlo Simulations}\label{sec:MonteCarloSimulation}


\subsection{Simulation Design}

We compare three types of variance estimators across different specifications and network configurations. We use the Eicker-Huber-White estimator as a benchmark,\footnote{It is used in \citet{bliss98} and \citet{mansfield00}, for instance.} the dyadic-robust estimator of \citet{Tabord-Meehan-2019} as a comparison accounting for the dyadic nature of the data (when inappropriately used in the presence of network spillovers), and our proposed estimator which is robust to network spillovers across dyads.

We first generate networks on which random variables are assigned. We follow \citet{Canen_et_al-2020b} among others by employing two models of random graph formations. They are referred to as Specifications 1 and 2. Specification 1 uses the \citepos{Barabasi_and_Albert-1999} model of preferential attachment, with the fixed number of edges $\nu \in\{1,2,3\}$ being established by each new node.\footnote{In generating the Barab\'asi-Albert random graphs, we follow \citet{Canen_et_al-2020b} by choosing the seed to be the Erd\"os-Renyi random graph with the number of nodes equal the smallest integer above $5\sqrt{N}$, where $N$ denotes the number of nodes.} Specification 2 is based on the Erd\"os-Renyi random graph \citep{Erdos_and_Renyi-1959, Erdos_and_Renyi-1960} with probability $p= \frac{\lambda}{N}$ for $N$ denoting the number of nodes and $\lambda\in\{1,2,3\}$ being a parameter that governs the probability relative to the node size. The summary statistics for the networks generated by Specification 1 and 2 are given in Appendix \ref{appdx:AdditionalResults:SummaryStatistics}. The maximum degree and the average degree increase monotonically as we increase the parameters in both specifications. The number of active edges (i.e., dyads) also increases with the sample size regardless of the specification. This reflects that each node tends to have more direct links as the network becomes denser. In our exercises, the number of indirectly linked dyads also increases with network denseness. However, this is due to our simulated networks being relatively sparse. In other settings, the number of indirect connections may decrease with network density.

For each of the randomly generated networks, the simulation data is generated from the following simple network-linear regression:
\begin{align*}
	y_{m} = x_{m} \beta + \varepsilon_{m}, \footnotemark
\end{align*}
\footnotetext{To simplify notation, we drop the $M$ subscript, making the triangular array structure implicit.} 
\noindent with $m\coloneqq d(i,j)$ representing the dyad between agent $i$ and $j$. The dyad-specific regressor $x_{m}$ is defined as $x_{m}\coloneqq |z_{i} - z_{j}|$, where both $z_{i}$ and $z_{j}$ are drawn independently from $\mathcal{N}(0,1)$. The regression coefficient is fixed to $\beta=1$ across specifications. 

The dyad-specific error term $\varepsilon_{m}$ is constructed to exhibit non-zero correlation with $\varepsilon_{m'}$ as long as dyads $m$ and $m'$ are connected (i.e., in the network terminology, there exists a path in the simulated network), while the strength of the correlation is assumed to decay as they are more distant. To that end, we draw $\varepsilon_{m} \coloneqq \sum_{m'} \gamma_{m,m'} \eta_{m,m'}$, where $\gamma_{m,m'}$ equals $\gamma^{s}$ if the distance between $m$ and $m'$ is $s$, and $0$ otherwise, for $\gamma\in[0,1]$\footnote{In this simulation, we focus on cases of positive spillovers, as negative spillovers can be analyzed analogously.} and $s\in\{1,\ldots,S\}$ with $S$ being the maximum geodesic distance that the spillover propagates to. Each $\eta_{m,m'}$ is drawn $i.i.d.$ from $\mathcal{N}(0,1)$. Hence, $\gamma$ controls the strength of spillover effects, representing their decay rate. If $\gamma=1$, then spillover effects are the same no matter how far the agents are apart, i.e., the spillover effects do not decay. If $\gamma = 0$, there are no spillover effects, so the dyadic-robust variance estimator should be consistent. The case of $S=2$ corresponds to a situation where up to friends of friends may matter for spillovers.

We consider three scenarios for each type of network. In the main text, we set $S=2$ and $\gamma=0.8$. The results for $S=2$ with $\gamma=0.2$ are given in Appendix \ref{appdx:AdditionalResults:S2gamma0.2}, and the ones for $S=1$ with $\gamma=0.8$ are in Appendix \ref{appdx:AdditionalResults:S1}. Throughout the experiments, we employ the mean-shifted (by one) rectangular kernel with the lag truncation equal to two for the sake of comparison.\footnote{These choices for the kernel and lag truncation correspond to placing unit weights on pairs of adjacent dyads and pairs of connected dyads of distance equal two, and zero otherwise. This combination is sufficient for us to accommodate all the correlations across the network since it is constructed to have $S=2$. This provides a transparent comparison to other estimators. However, it is rare for a researcher to have such information. Appendix \ref{appdx:AdditionalResults:SInfty} provides further exercises where $S>2$ (so spillovers reach further dyads) and where the lag-truncation is chosen adaptively by the rule recommended in Section \ref{subsec:NetworkHACVarianceEstimation}. Those results are discussed below.}

\subsection{Results}

In Table \ref{tbl:MonteCarloSimulation:Results:CPACI_S2gamma0.8woKernel} we present the coverage probability for $\beta$ and the average length of the confidence interval across simulations. To do so, we compute the $t$-statistic using the OLS estimator for $\beta$ and different variance estimators under a Normal distribution approximation.\footnote{It is well known that estimates of a variance-covariance matrix may be negative semidefinite when the sample size is very small. This occurs in four out of 5000 simulations when $N=500$. Rather than dropping such observations, we follow \citet{Cameron_et_al-2011} and augment the eigenvalues of the matrix by adding a small constant, say $0.005$, thereby obtaining a new variance estimate that is more conservative.} The finite-sample properties of the three variance estimators are further illustrated in Figure \ref{fig:MonteCarloSimulation:Results:Boxplots08} in Appendix \ref{appdx:AdditionalResults:S2gamma0.8}.

\begin{table}[ht!]
	\centering
	\caption{
	The empirical coverage probability and average length of confidence intervals for $\beta$ at $95\%$ nominal level: $S=2$, $\gamma=0.8$.}
	\begin{threeparttable}
		\begin{tabular}{lcccccccc} \toprule \toprule
							& $$ 		& \multicolumn{3}{c}{Specification 1} && \multicolumn{3}{c}{Specification 2} \\ \cmidrule{3-5}\cmidrule{7-9}
							& $N$ 	& $\nu=1$	& $\nu=2$	& $\nu=3$	& & $\lambda=1$ 	& $\lambda=2$ & $\lambda=3$ \\ \midrule
							&		& \multicolumn{6}{c}{Coverage Probability} 					\\ 
			Eicker-Huber-White	& 500	& 0.877 & 0.868 & 0.871 & & 0.891	& 0.870	& 0.875     \\ 
  							& 1000	& 0.880 & 0.873 & 0.873 & & 0.892	& 0.881	& 0.888     \\ 
  							& 5000	& 0.879 & 0.871 & 0.877 & & 0.890	& 0.882	& 0.880     \\
			Dyadic-robust 		& 500	& 0.922 & 0.898 & 0.894 & & 0.932	& 0.921	& 0.917     \\ 
  							& 1000	& 0.929 & 0.913 & 0.901 & & 0.937	& 0.927	& 0.924     \\ 
  							& 5000	& 0.934 & 0.912 & 0.909 & & 0.939	& 0.933	& 0.922     \\
			Network-robust 	& 500 	& 0.930 & 0.917 & 0.915 & & 0.937	& 0.937	& 0.941     \\ 
  							& 1000 	& 0.939 & 0.934 & 0.933 & & 0.946	& 0.945	& 0.948     \\ 
  							& 5000 	& 0.949 & 0.944 & 0.943 & & 0.947	& 0.948	& 0.948     \\ 
							\\
							&		& \multicolumn{6}{c}{Average Length of the Confidence Intervals}	\\ 
			Eicker-Huber-White	& 500	& 0.368 & 0.409 & 0.482 & & 0.287	& 0.285	& 0.296     \\ 															& 1000	& 0.266 & 0.302 & 0.331 & & 0.205	& 0.201	& 0.207     \\ 
							& 5000 	& 0.132 & 0.159 & 0.176 & & 0.092	& 0.090	& 0.094     \\
			Dyadic-robust 		& 500 	& 0.426 & 0.454 & 0.520 & & 0.328	& 0.329	& 0.337     \\ 
  							& 1000 	& 0.312 & 0.339 & 0.361 & & 0.236	& 0.232	& 0.237     \\ 
  							& 5000 	& 0.158 & 0.178 & 0.192 & & 0.106	& 0.104	& 0.108     \\ 
			Network-robust 	& 500 	& 0.441 & 0.493 & 0.568 & & 0.337	& 0.349	& 0.366     \\ 
  							& 1000 	& 0.326 & 0.373 & 0.408 & & 0.244	& 0.248	& 0.259     \\ 
  							& 5000 	& 0.167 & 0.199 & 0.222 & & 0.110	& 0.112	& 0.118     \\ \bottomrule
  		\end{tabular}
	%
		\begin{tablenotes}[para,flushleft]
			\item[] {\footnotesize {\sl Note}: The upper-half of the table displays the empirical coverage probability of the asymptotic confidence interval for $\beta$, and the lower-half showcases the average length of the estimated confidence intervals. As the sample size ($N$) increases, the empirical coverage probability for our estimator accounting for network spillovers approaches $0.95$, the correct nominal level. However, that is not the case for alternative estimators.}
		\end{tablenotes}
	\end{threeparttable}
	\label{tbl:MonteCarloSimulation:Results:CPACI_S2gamma0.8woKernel}
\end{table}

The results for the empirical coverage probabilities depend on two dimensions: the sample size ($N$) and the denseness of the underlying network (parametrized by $\nu$ and $\lambda$). The coverage probability for each estimator improves with the sample size. However, when spillovers are high ($\gamma = 0.8$), only our proposed network-robust variance estimator has coverage close to 95\%, consistent with Theorem \ref{thm:ExogenousRegression:ConsistencyOfTheNetwork-HACVarianceEstimator}. Meanwhile, in this set-up, both the Eicker-Huber-White and the dyadic-robust variance estimators perform poorly as the underlying network becomes denser, no matter which specification of the network is involved. For example, in Specification 1 with $\nu=3$ and the largest sample size ($N=5000$), the confidence intervals based on the Eicker-Huber-White and the dyadic-robust variance estimators do not cover the true parameter 615 and 455 times out of 5000 simulations (12.3\% and 9.1\%), respectively. On the other hand, the network-robust variance estimator is designed to capture higher-order correlations and, thus, its coverage remains stable across network configurations.

A similar conclusion is drawn from the average length of the confidence intervals. The confidence intervals for the Eicker-Huber-White and dyadic-robust variance estimators are typically shorter than those for our proposed estimator when $\gamma$ is large and $S=2$. This means the former undercovers the true parameter (in the presence of positive spillovers), confirming our theoretical results (Corollary \ref{coro:ExogenousRegression:NetworkHACEstimators:InconsistencyResult}). The difference is often around 10 -- 20\% of the length of the network-robust estimator and it increases with even denser networks (see Table \ref{tbl:MonteCarloSimulation:Results:CPACI_S2gamma0.8woKernelwHigherDensenessParameters} in Appendix).

However, as the magnitude of spillovers decreases (i.e. $\gamma$ tends to zero), higher-order spillovers are less pronounced, so that the biases from using the Eicker-Huber-White and dyadic-robust variance estimators disappear. This is shown in Table \ref{tbl:MonteCarloSimulation:Results:CPACI_S2gamma0.2woKernel} of Appendix \ref{appdx:AdditionalResults:S2gamma0.2} for the case of $S=2$ and $\gamma = 0.2$. When $S=1$, the dyadic-robust variance estimator coincides with our proposed estimator (i.e., there are no spillovers from non-adjacent links, or spillovers fully decay immediately). This is shown in Table \ref{tbl:MonteCarloSimulation:Results:CPACI_S1gamma0.8woKernel} of Appendix \ref{appdx:AdditionalResults:S1}. 

Finally, Appendix \ref{appdx:AdditionalResults:SInfty} shows that the results are robust to spillovers that can reach the most distantly connected neighbors ($S=\infty$) and to choosing the lag-truncation adaptively. Rather than assuming a given value, we set $b_M = 2 \log(M)/\log(\max(average~degree, 1.05))$ as suggested above. The results are qualitatively and quantitatively very similar: when the decay is not too high (i.e., spillovers propagate), the confidence interval associated with the dyadic-robust estimator will undercover the true parameter even in moderate sample sizes.

\section{Empirical Illustration: Legislative Voting in the European Parliament}\label{sec:EmpiricalIllustration}

We now turn to an empirical application to demonstrate the performance of our variance estimator with real-world data. In doing so, we revisit the work of \citet{Harmon_et_al-2019} on whether legislators who sit next to each other in Parliament tend to vote more alike on policy proposals. 

They focus on the European Parliament, whose Members (MEPs) are voted in through elections in each European Union (EU) member country every five years. The Parliament convenes once or twice a month, in either Brussels or Strasbourg, to debate and vote on a series of proposals. Once elected to the European Parliament (EP), these MEPs are organized into European Political Groups (EPGs), which aggregate similar ideological members/parties across countries. As \citet{Harmon_et_al-2019} describe, these EPGs function as parties for many of the traditional party-functions in other legislatures, including coordination on policy and policy votes. Most importantly, MEPs sit within their EPG groups in the chamber. However, within each EPG group, non-party leaders traditionally sit in alphabetical order by last name. See Figure \ref{fig:SeatingPlanAtTheEuropeanParliament} in Appendix \ref{appndx:EmpiricalIllustration:SeatingArrangement} for an example. Hence, to the extent that last names are exogenous (i.e., politicians are not selected or choose last names anticipating seating arrangements in the European Parliament), seating arrangements are used as a quasi-natural experiment.

\subsection{Data}

We adopt the same dataset used in \citet{Harmon_et_al-2019}, which collects the MEP-level data on votes cast in the EP. The dataset records what each MEP voted for (Yes or No), where she was seated, and a number of individual characteristics (e.g. country, age, education, gender, tenure, etc). Its sample period is the plenary sessions held in both Brussels and Strasbourg between October 2006 and November 2010. This corresponds to the sixth and seventh terms of the European Parliament. 

For our empirical illustration, we restrict the sample to the policies voted in Strasbourg during the seventh term and we focus on the seating pattern between July 14th -- July 16th, 2009 (which involved 116 different proposals being voted on). The resulting sample has 2,431,261 observations, which are split into 422 politicians forming 26,099 pairs (i.e., dyads) of MEPs over 116 proposals.\footnote{There are 334 pairs of adjacent dyads and 591 pairs of connected dyads.} Further information on the construction of our sample is detailed in Appendix \ref{appndx:EmpiricalIllustration:DataConstruction}. These restrictions keep the main set-up in \citet{Harmon_et_al-2019}, while allowing us to evaluate variance estimators with smaller sample sizes.

\subsection{Empirical Set-up}


Let $Agree_{d(i,j), t}$ be an indicator that takes one if MEP $i$ and $j$ cast the same vote on proposal $t$, and zero otherwise. Likewise, $SeatNeighbors_{d(i,j), t}$ is a binary variable that equals one if MEP $i$ and $j$ are seated next to each other when the vote for proposal $t$ is taken place, and zero otherwise. We view the seating arrangement as a network among the MEPs: i.e., two MEPs who are seated next to each other within the same political group are treated as an active dyad, which in turn constitutes a network over pairs of MEPs within each political group (see the note below Figure \ref{fig:SeatingPlanAtTheEuropeanParliament} for details).\footnote{This accommodates the row-by-EP-by-EPG clustering implemented by the authors. This is already an extension to \citet{Harmon_et_al-2019} which assume that only seating neighbors would have correlated error terms. One could extend this further by allowing connections across parties.}

We follow \citet{Harmon_et_al-2019} in assuming that such seating arrangements are exogenously determined, i.e., the adjacency relation is exogenously formed. Their main specification is a linear model:
\begin{align}\label{eq:EmpiricalApplication:EmpiricalSetup}
	Agree_{d(i,j), t} = \beta_{0} + \beta_{1} SeatNeighbors_{d(i,j), t} + \varepsilon_{d(i,j), t}.
\end{align}

The authors originally conducted inference using the estimator in \citet{Aronow_et_al-2015}, assuming that dyads $m = d(i,j)$ cannot be correlated with $m' = d(k,l)$ unless they share a common member (i.e., there is no correlation across errors if $i,j$ and $k,l$ do not include a common unit). We will now compare this approach to using the variance estimator introduced in Section \ref{subsec:NetworkHACVarianceEstimation}, which allows the error terms to exhibit non-zero correlations as long as they are connected on the network over dyads represented by the adjacency relation of seating arrangements in Parliament. In this analysis, we use the mean-shifted rectangular kernel with the lag truncation equal the longest path in the constructed network. This combination of the kernel density and the lag truncation enables us to accommodate all the possible correlations across connected dyads (i.e., pairs of MEPs), placing equal weight on each of them. This makes the comparisons across estimators more transparent.\footnote{In Appendix \ref{appndx:EmpiricalIllustration:FullResults}, we replicate this analysis with a different choice of kernel and setting the lag-truncation parameter following the criterion suggested above/in \citet{Kojevnikov_et_al-2021}. The results are very similar.}

Inspired by \citet{Harmon_et_al-2019}, we consider three specifications: (I) a simple linear regression model as given in (\ref{eq:EmpiricalApplication:EmpiricalSetup}); (II) the model (\ref{eq:EmpiricalApplication:EmpiricalSetup}) augmented with a flexible set of other demographic variables;\footnote{Following \citet{Harmon_et_al-2019}, we include indicators whether country of origins, quality of education, freshman status and gender, respectively, are the same, as well as differences in ages and tenures. See the note below Table \ref{tbl:EmpiricalApplication:Results:MainAnalysis:StrasbourgTerm7} for details.} and (III) the model (\ref{eq:EmpiricalApplication:EmpiricalSetup}) with both a flexible set of other demographic variables and day-specific fixed effects. When fixed effects are present in their original estimation, we estimate a within-difference model via OLS.

\subsection{Results}\label{sec:EmpiricalIllustration:Results}

The main results of our empirical analysis are summarized in Table \ref{tbl:EmpiricalApplication:Results:MainAnalysis:StrasbourgTerm7}. Panel A displays the parameter estimates for the three different specifications. This panel shows that our point-estimates are consistent with the original estimates of \citet{Harmon_et_al-2019} (columns 6 and 7 of Table 4), as they are close to 0.006 (their original results) and stable across specifications.\footnote{Note that our dependent variable is equal to one if two MEPs vote the same and zero otherwise, while \citet{Harmon_et_al-2019} code it as one if MEPs vote differently. Hence, to compare our estimates with theirs, the signs on the estimates of $SeatNeighbors$ must be flipped.} Hence, changes to point-estimates are not due to sample selection. The positive coefficient for {\sl SeatNeighbors} suggests that the MEPs sitting together tend to vote more similarly than those sitting apart, providing evidence in favor of their original hypothesis. The coefficients on the covariates (displayed in Panel C of Appendix Table \ref{tbl:EmpiricalApplication:Results:MainAnalysis:StrasbourgTerm7_2}) are also quantitatively and qualitatively similar to those in their original paper. For instance, our estimates for {\sl SameCountry} are 0.056, while their estimates are around 0.051, suggesting politicians from the same country are more likely to vote similarly on policies. 

Panel B shows the standard errors for the regression coefficient of {\sl SeatNeighbors} using different variance estimators. Building on Section \ref{sec:MonteCarloSimulation}, the panel compares three different types of heteroskedasticity robust estimators: namely, the Eicker-Huber-White, the dyadic-robust estimator (used in their original work), and our proposed network-robust estimator. We note that, due to the smaller sample size, the standard errors in our exercise are larger than the original authors'. Hence, the results in this section are not directly comparable to those in \citet{Harmon_et_al-2019}. Rather, we compare the different estimators within our own sample.

As foreshadowed in the Monte Carlo simulations, the Eicker-Huber-White estimates are the smallest, followed by the dyadic-robust estimates, which, in turn, are smaller than the network-robust estimates. In fact, for Specification (III), the Eicker-Huber-White estimate is roughly 73\% smaller than using the estimator accounting for network spillovers across dyads, while the dyadic-robust one is 22\% smaller. This fact entails two implications. First, our finding provides empirical evidence in support of the existence of \textit{indirect} positive spillovers among the MEPs: even distant connections may indirectly generate correlated behavior among politicians $i$ and $j$. Second, the use of alternative estimators not accounting for such spillovers undercovers the true parameter and may generate biased hypothesis testing about the regression coefficient of {\sl SeatNeighbors}. The difference in estimates appears quantitatively meaningful in this empirical example. 

\begin{table}[htbp]
	\centering
	\caption{Spillovers in Legislative Voting -- Main Analysis}
	\begin{threeparttable}
		\centering
		{\scriptsize
		\begin{tabular}{lccc} \toprule \toprule
									& Specification (I) 	& Specification (II) 	& Specification (III) \\ \midrule
			{\sl Panel A: Parameter Estimates}&&&					 \\
			\quad Seat neighbors		& 0.007 		& 0.006		& 0.006  \\ \\
			
			{\sl Panel B: Standard errors} &&& \\
			\quad Eicker-Huber-White			& 0.003 		& 0.003	& 0.003 \\
			\quad Dyadic-robust			& 0.008	 	& 0.008 	& 0.009 \\
			\quad Network-robust		& 0.010		& 0.010	& 0.011 \\ \bottomrule
			
			%
		\end{tabular}
		}
	
	\end{threeparttable}
		\parbox{6.2in}{\footnotesize Notes: Panel A displays the parameter estimates for the variable ``Seat Neighbors" for the three different specifications, and Panel B shows the standard errors for its regression coefficient using different variance estimators. Adjacency of MEPs is defined at the level of a row-by-EP-by-EPG. (See the note below Figure \ref{fig:SeatingPlanAtTheEuropeanParliament}.) Independent variables are as follows: {\sl Seat neighbors} is an indicator variable denoting whether both MEPs sit together; {\sl Same country} represents an indicator for whether both MEPs are from the same country; {\sl Same quality education} is an indicator showing whether both MEPs have the same quality of education background, measured by if both have the degree from top 500 universities; {\sl Same freshman status} encodes whether both MEPs are freshman or not; {\sl Age difference} is the difference in the MEPs' ages; and {\sl Tenure difference} measures the difference in the MEPs' tenures. A full description of the result is provided in Table \ref{tbl:EmpiricalApplication:Results:MainAnalysis:StrasbourgTerm7_2}.}
	\label{tbl:EmpiricalApplication:Results:MainAnalysis:StrasbourgTerm7}
\end{table}


\section{Conclusion}\label{sec:Conclusion}

With dyadic data, researchers typically assume that dyads are uncorrelated if they do not share a common unit, an assumption that is leveraged in inference. We showed that this assumption may be inappropriate in many models where interactions occur on a network: while data may be dyadic, the cross-sectional dependence may be much more complex and spillover beyond pairwise interactions. For instance, trade between countries may depend on trade between auxiliary partners and their unobservables. In political economy, whether a politician votes with a colleague may have spillovers from seating neighbors beyond one's own immediate ones. We verified this using both theoretical results and Monte Carlo simulations. We provided a new consistent variance estimator for parameters in a linear model with dyadic data with correct asymptotic coverage and good finite sample properties. 

To conclude, we clarify that our goal in this exercise is neither to criticize dyadic-robust variance estimators, which are a fundamental part of the empiricist's toolkit, nor to suggest our approach should always be used. Rather, we wish to draw attention that researchers should fully specify the cross-sectional dependence in their model. If the conventional assumption of dyadic dependence correctly specifies the environment in question, or when spillovers beyond immediate neighbors might be negligible, then previous approaches suffice. However, as we have discussed above, existing applications may apply the latter method even if it is seemingly inappropriate to their setting. This includes situations where such network spillovers may be present or persistent (even with decay). In such scenarios, our estimator provides a possible solution. Those choices, though, must be guided by the application that empiricists face. Hence, building on \cite{poast}, we recommend researchers to continue to fully specify their model, including full specification of their covariance structure, thereby clarifying what type of inference procedure is most appropriate for their environment.

\singlespacing
\setlength\bibsep{0pt}
\subsection*{Acknowledgments}
We thank Aimee Chin, Hugo Jales, Taisuke Otsu, Pablo Pinto, V\'{i}tor Possebom, Kevin Song, Bent E. S\o rensen, and seminar participants at the University of Houston, the 2022 Texas Econometrics Camp, the 2022 European Winter Meeting and the 2022 Asia Meeting of the Econometric Society in East and South-East Asia for valuable comments and suggestions. 

\subsection*{Data Availability Statement}
Replication code and data for this article will be made available publicly online upon conditional acceptance.

\subsection*{Competing Interests Statement}
Competing interests: The authors declare none.
\bibliographystyle{chicago}
\bibliography{ListOfReferences}

\clearpage

\onehalfspacing


\appendix
\begin{center}
\section*{\large Supplementary Material (Online-Only Publication) for: \\ ``Inference in Linear Dyadic Data Models \\ with Network Spillovers" \\ \medskip Nathan Canen and Ko Sugiura}
\end{center}

\section{Mathematical Set-up}\label{appdx:MathmaticalSetup}

This section lays out the mathematical set-up of our model in more detail, heavily drawing from \citet{Kojevnikov_et_al-2021}. We conclude with a discussion of the related statistical literature.

We first define a collection of pairs of sets of dyads. For any positive integers $a$, $b$ and $s$, define
\begin{align*}
	\mathcal{P}_{M}(a,b; s) \coloneqq \{ (A,B): A,B\subset\mathcal{M}_{N}, |A|=a, |B|=b , \rho_{M}(A,B) \geq s \},
\end{align*}
where
\begin{align}\label{eq:ExogenousRegression:DistanceBetweenTwoSets}
	\rho_{M}(A,B) \coloneqq \min_{m\in A} \min_{m'\in B} \rho_{M}(m,m'),
\end{align}
with $\rho_{M}(m,m')$ denoting the geodesic distance between dyads $m$ and $m'$, i.e., the smallest number of adjacent dyads between dyads $m$ and $m'$. In words, the set $\mathcal{P}_{M}(a,b; s)$ collects all two distinct sets of active dyads whose sizes are $a$ and $b$ and that have no dyads in common. 

Next we consider a collection of bounded Lipschitz functions. Define
\begin{align*}
	\mathcal{L}_{K} \coloneqq \{ \mathcal{L}_{K,c}: c\in\mathbbm{N} \},
\end{align*}
where
\begin{align*}
	\mathcal{L}_{K,c} \coloneqq \{ f:\mathbbm{R}^{K\times c}\rightarrow\mathbbm{R}: \|f\|_{\infty} < \infty, \text{Lip}(f) < \infty \},
\end{align*}
with $\|\cdot\|_{\infty}$ representing the supremum norm and $\text{Lip}(f)$ being the Lipschitz constant.\footnote{It is immediate to see that $\mathbbm{R}$ is a normed space with respect to the Euclidean norm, while the $\mathbbm{R}^{K\times c}$ can be equipped with the norm $\rho_{c}(x,y) \coloneqq \sum_{\ell=1}^{c} \|x_{\ell}-y_{\ell}\|$ where $x, y \in \mathbbm{R}^{K\times c}$ and $\|z\|\coloneqq (z'z)^{\frac{1}{2}}$, thereby the Lipschitz constant is defined as $\text{Lip}(f)\coloneqq \min\{w\in\mathbbm{R}: |f(x)-f(y)| \leq w \rho_{c}(x,y) \ \forall x,y\in\mathbbm{R}^{K\times c} \}$.} In words, the set $\mathcal{L}_{K,c}$ collects all the bounded Lipschitz functions on $\mathbbm{R}^{K\times c}$ and the set $\mathcal{L}_{K}$ moreover gathers such sets with respect to $c\in\mathbbm{N}$.

Lastly, we write
\begin{align*}
	Y_{M,A} \coloneqq (Y_{M,m})_{m\in A},
\end{align*}
and $Y_{M,B}$ is analogously defined. Let $\{\mathcal{C}_{M}\}_{M\geq 1}$ denote a sequence of $\sigma$-algebras and be suppressed as $\{\mathcal{C}_{M}\}$. 

The network dependent random variables are characterized by the upper bound of their covariances, first defined in Definition 2.2 of \citet{Kojevnikov_et_al-2021}. 

\begin{definition}[Conditional $\psi$-Dependence given $\{\mathcal{C}_{M}\}$]
	A triangular array $\{ Y_{M,m}\in\mathbbm{R}^{K}: M\geq 1, m\in\{1,\ldots, M\}\}$ is called conditionally $\psi$-dependent  given $\{\mathcal{C}_{M}\}$, if for each $M\in\mathbbm{N}$, there exist a $\mathcal{C}_{M}$-measurable sequence $\theta_{M}\coloneqq \{\theta_{M,s}\}_{s\geq 0}$ with $\theta_{M,0}=1$, and a collection of nonrandom function $(\psi_{a,b})_{a,b\in\mathbbm{N}}$ where $\psi_{a,b}:\mathcal{L}_{K,a}\times\mathcal{L}_{K,b}\rightarrow [0.\infty)$, such that for all $(A,B)\in\mathcal{P}_{M}(a,b;s)$ with $s>0$ and all $f\in \mathcal{L}_{K,a}$ and $g\in \mathcal{L}_{K,b}$,
	\begin{align*}
		\big| Cov\big( f(Y_{M,A}), g(Y_{M,B}) \mid \mathcal{C}_{M} \big) \big| \leq \psi_{a,b}(f,g)\theta_{M,s} \ \ \ \ \ a.s.
	\end{align*}
\end{definition}

Intuitively, this definition states that the upper bound must be decomposed into two components. The first part $\psi_{a,b}(f,g)$ is deterministic and depends on nonlinear Lipschitz functions $f$ and $g$. The other component $\theta_{M,s}$ is stochastic and depends only on the distance of the random variables on the underlying network. The former, nonrandom component reflects the scaling of the random variables as well as that of the Lipschitz transformations, while the latter random part stands for the covariability of the two random variables. We call $\theta_{M,s}$ the dependence coefficient. We follow \citet{Kojevnikov_et_al-2021} in assuming boundedness for these two components.

\begin{assumption}[\citet{Kojevnikov_et_al-2021}, Assumption 2.1]\label{assm:Kojevnikov_et_al-2021:Assm2.1}
	The triangular array $\{ Y_{M,m}\in\mathbbm{R}^{K}: M\geq 1, m\in\{1,\ldots, M\}\}$ is conditionally $\psi$-dependent  given $\{\mathcal{C}\}$ with the dependence coefficients $\{\theta_{M,s}\}$ satisfying the following conditions: (a) there exists a constant $C>0$ such that $\psi_{a,b}(f,g) \leq C\times ab \big( \|f\|_{\infty} +\text{Lip}(f) \big) \big( \|g\|_{\infty} +\text{Lip}(g) \big)$; (b) $\sup_{M\geq 1}\max_{s\geq 1} \theta_{M,s} <\infty$ a.s.
\end{assumption}

Assumption \ref{assm:Kojevnikov_et_al-2021:Assm2.1} is maintained throughout the paper and employed to show asymptotic properties of our estimators such as the consistency and asymptotic normality, and the consistency of the network-robust variance estimator for dyadic data.

\subsection{Related Literature}

Our main insight in accommodating indirect spillovers is that we can rewrite the correlation structure among dyads as a dyadic network, where links denote whether they share a common member. As a result, this dyadic network describes how close/far certain dyads are from sharing members with other dyads. In doing so, the transformed problem is amenable to appropriate applications of recent developments in the statistics of random variables which are correlated along an (observable, exogenous) network. In particular, we apply asymptotic results for network-dependent random variables developed by \citet{Kojevnikov_et_al-2021}\footnote{\citet{Vainora-2020wp} provides another such theoretical contribution.} to an appropriately defined dyadic network, with assumptions imposed on the latter. \citet{Leung-2021wpb} and \citet{Leung-2022} also apply the framework of \citet{Kojevnikov_et_al-2021} to study, respectively, cluster-robust inference and causal inference for the case of individual-specific random variables. These papers focus on the correlation along a network over individuals, rather than over dyads. Meanwhile, \citet{Leung_and_Moon-2021wp} derive an asymptotic theory for dyadic variables in the context of networks, primarily for endogenous network formation models.
\section{Proofs of Main Theorems and Results}\label{appndx:ProofOfMainTheorems:ExogenousLinearRegression}

\subsection{Identification of $\beta$}\label{appndx:IdentificationOfBeta}

\begin{assumption}\label{assm:ExogenousRegression:IdentificationCondition}
	For each $N\in\mathbbm{N}_{}^{}$: \\
	(a) $\sup_{m\in\mathcal{M}_{N}} E\big[ | \varepsilon_{M, m} |^{2} \big]$ exists and is finite; \\
	(b) $\sup_{m\in\mathcal{M}_{N}} E\big[ \| x_{M, m} \|_{}^{} \big]$ exists and is finite; \\
	(c) $E\big[ x_{M, m} x_{M, m}' \big]$ exists with finite elements and positive definite for all $m\in\mathcal{M}_{N}$; \\
	(d) $E\big[ \varepsilon_{M, m}^{} \mid X_{M}^{} \big] = 0$ for all $m\in\mathcal{M}_{N}$.
\end{assumption}

Assumption \ref{assm:ExogenousRegression:IdentificationCondition} (a) and (b) are standard and jointly imply the finite existence of the second moment of $y_{M, m}$ for all $m\in\mathcal{M}_{N}$, which in turn implies the finite existence of the cross moment of $y_{M, m}$ and $x_{M, m}$ for all $m\in\mathcal{M}_{N}$. The third and fourth assumptions are also standard in the context of the linear regression models and require no multicolinearity and strict exogeneity, respectively. 

Identification of the linear parameter in equation (\ref{eq:ExogenousRegression:PopulationNetworkRegression_a}) follows from Assumption \ref{assm:ExogenousRegression:IdentificationCondition} (see Proposition \ref{prop:ExogenousRegression:Identification} in Appendix \ref{appndx:IdentificationOfBeta}).

\begin{proposition}[Identification]\label{prop:ExogenousRegression:Identification}
	Under Assumption \ref{assm:ExogenousRegression:IdentificationCondition}, the regression parameter $\beta$ in (\ref{eq:ExogenousRegression:PopulationNetworkRegression_a}) is identified.
\end{proposition}

\begin{proof}\label{proof:prop:ExogenousRegression:Identification}
	For each $m\in\mathcal{M}_{N}$, premultiply the model (\ref{eq:ExogenousRegression:PopulationNetworkRegression_a}) by $x_{M, m}$ to obtain
	\begin{align*}
		x_{M, m} y_{M, m}
		= x_{M, m} x_{M, m}' \beta
			+ x_{M, m} \varepsilon_{M, m} \hspace{10mm} \forall m\in\mathcal{M}_{N}.
	\end{align*}

	Taking the expectation with respect to $\{(x_{M, m}, y_{M, m}, \varepsilon_{M, m})\}_{m\in\mathcal{M}_{N}}$ implies:
	\begin{align*}
E\Big[ x_{M, m} y_{M, m} \Big]
		&= E\Big[x_{M, m} x_{M, m}' \Big] \beta
			+ E\Big[x_{M, m} \varepsilon_{M, m}\Big].
	\end{align*}
		The second term on the right hand side is equal to 0, due to Assumption \ref{assm:ExogenousRegression:IdentificationCondition} (d). Next, Assumption \ref{assm:ExogenousRegression:IdentificationCondition} (c) ensures existence of the inverse of the expectation term in the first term of the right hand side, ensuring identification. 
\end{proof}

\subsection{Consistency of $\hat{\beta}$}\label{appndx:ProofOfMainTheorems:ExogenousLinearRegression:ConsistencyOfBetaHat}
As usual, the Central Limit Theorem for a normalized sum requires us to have stronger conditions than what is required for consistency. Those stronger conditions were introduced in the main text as Assumptions \ref{assm:ExogenousRegression:ConditionalFiniteMoment2} and \ref{assm:Kojevnikov_et_al-2021:Assm3.4}. However, they are only required for Theorem \ref{theo:ExogenousRegression:CLT}. For the consistency proof (Theorem \ref{theo:ConsistencyofBetaHat}),we can replace those two assumptions by the following weaker conditions.

\begin{assumption}\label{assm:ExogenousRegression:ConditionalFiniteMoment1}
	There exists $\eta>0$ such that $\sup_{N\geq1}\max_{m\in\mathcal{M}_{N}} E\big[ | \varepsilon_{M, m} |^{1+\eta} \mid \mathcal{C}_{M} \big] < \infty$.
\end{assumption}

Assumption \ref{assm:ExogenousRegression:ConditionalFiniteMoment1} allows for the same interpretation as Assumption \ref{assm:ExogenousRegression:ConditionalFiniteMoment2}, i.e., the random error term $\varepsilon_{m}$ cannot be too large, conditional on a common component. This assumption, however, is less stringent than the previous one because it now requires the finiteness of a lower moment of $\varepsilon_{m}$. 

\begin{assumption}\label{assm:Kojevnikov_et_al-2021:Assm3.2}
	$\frac{1}{M}\sum_{s\geq 1} \delta_{M}^{\partial}(s; 1)\theta_{M,s} \stackrel{a.s.}{\rightarrow} 0$ as $M\rightarrow \infty$.
\end{assumption}

Similar to Assumption \ref{assm:Kojevnikov_et_al-2021:Assm3.4}, this assumption binds the covariance of the random variables, the dependence reflected in the dependence coefficients, and the underlying network. That is, given $\sigma_{M}$ growing at least at the same rate of $M$, the composite of the density of the network and the magnitude of the correlations of the random variables must decay fast enough.

\medskip

\noindent
\textbf{Proof of Theorem \ref{theo:ConsistencyofBetaHat}:}
	From (\ref{eq:ExogenousRegression:OLSEstimator_2}), (\ref{eq:ExogenousRegression:DefinitonOfY}) and (\ref{eq:ExogenousRegression:ExpressionOfBetaHat}), we can write
	\begin{align*}
		\hat{\beta} - \beta
		&= \Big( \sum_{j\in\mathcal{M}_{N}^{}} x_{M, j}x_{M, j}' \Big)^{-1}
		\sum_{m\in\mathcal{M}_{N}^{}} x_{M, m} \varepsilon_{M, m}  \displaybreak[0]\\
		&= \Big( \frac{1}{M} \sum_{j\in\mathcal{M}_{N}^{}} x_{M, j}x_{M, j}' \Big)^{-1} \frac{1}{M} \sum_{m\in\mathcal{M}_{N}} Y_{M,m} \\
		&= \frac{1}{M} \sum_{m\in\mathcal{M}_{N}} \Big( \frac{1}{M} \sum_{j\in\mathcal{M}_{N}^{}} x_{M, j}x_{M, j}' \Big)^{-1} Y_{M,m}.
	\end{align*}
	Define $\tilde{Y}_{M,m} \coloneqq \Big( \frac{1}{M} \sum_{j\in\mathcal{M}_{N}^{}} x_{M, j}x_{M, j}' \Big)^{-1} Y_{M,m}$ and let $\tilde{Y}_{M,m}^{u}$ be the $u$-th entry of $\tilde{Y}_{M,m}$. That is,
	\begin{align*}
		\tilde{Y}_{M,m}^{u} 
		&= D^{u} Y_{M,m} \\
		&= D^{u} x_{M,m} \varepsilon_{M,m},
	\end{align*}
	where $D^{u}$ stands for the $u$-th row of the matrix $\big( \frac{1}{M} \sum_{j\in\mathcal{M}_{N}^{}} x_{M, j}x_{M, j}' \big)^{-1}$.
	Moreover, let $\hat{\beta}^{u}$ and $\beta^{u}$, respectively, denote the $u$-th entry of $\hat{\beta}$ and $\beta$, so that we can write
	\begin{align*}
		\hat{\beta}^{u} - \beta^{u}
		&= \frac{1}{M} \sum_{m\in\mathcal{M}_{N}} \tilde{Y}_{M,m}^{u},
	\end{align*}
	
	In light of Assumption \ref{assm:ExogenousRegression:IdentificationCondition} (d), it holds that for any $N>0$ and for each $m\in\mathcal{M}_{N}$
	\begin{align*}
		E\big[ \tilde{Y}_{M,m}^{u} \mid \mathcal{C}_{M} \big]
		&= D^{u} x_{M, m} \underbrace{E\big[ \varepsilon_{M, m} \mid \mathcal{C}_{M} \big]}_{0} \displaybreak[0]\\
		&= 0.
	\end{align*}
	By Theorem 3.1 of \citet{Kojevnikov_et_al-2021}, $\Big\| \frac{1}{M} \sum_{m\in\mathcal{M}_{N}} \Big( \tilde{Y}_{M,m}^{u} - \underbrace{E\big[ \tilde{Y}_{M,m}^{u} \mid \mathcal{C}_{M} \big]}_{0} \Big) \Big\|_{\mathcal{C}_{M},1} \stackrel{a.s.}{\rightarrow} 0$. Hence,
	\begin{align*}
		\Big\| \frac{1}{M} \sum_{m\in\mathcal{M}_{N}} \tilde{Y}_{M,m}^{u} \Big\|_{\mathcal{C}_{M},1}\stackrel{a.s.}{\rightarrow} 0  \qquad M\rightarrow \infty,
	\end{align*}
	so that
	\begin{align*}
		E\big[ \big| \hat{\beta}^{u} - \beta^{u} \big| \big]
		&= E\Big[ E\big[ |\hat{\beta}^{u} - \beta^{u} | \mid \mathcal{C}_{M}  \big] \Big] \displaybreak[0]\\
		&= E\Big[ \| \hat{\beta}^{u} - \beta^{u}\|_{\mathcal{C}_{M},1} \Big] \displaybreak[0]\\
		&= E\Big[ \Big\| \frac{1}{M} \sum_{m\in\mathcal{M}_{N}} \tilde{Y}_{M,m}^{u} \Big\|_{\mathcal{C}_{M},1} \Big] \displaybreak[0]\\
		&\rightarrow{} 0 \qquad M\rightarrow \infty,
	\end{align*}
	where the last implication is a consequence of the Dominated Convergence Theorem. In view of Assumption \ref{assm:ExogenousRegression:Denseness}, this is true also with respect to $N$ going to infinity.
	
	Since it holds by the Markov inequality that for any $c>0$
	\begin{align*}
		\mathrm{Pr}\big( | \hat{\beta}^{u} - \beta^{u} | > c\big) \leq \frac{E\big[ | \hat{\beta}^{u} - \beta^{u} | \big]}{c},
	\end{align*}
	it then follows that
	\begin{align*}
		\mathrm{Pr}\big( | \hat{\beta}^{u} - \beta^{u} | > c\big) \rightarrow 0,
	\end{align*}
	as $N\rightarrow\infty$.
	Hence we have
	\begin{align*}
		| \hat{\beta}^{u} - \beta^{u} | \stackrel{p}{\rightarrow} 0 \qquad as \quad N\rightarrow\infty.
	\end{align*}
	Finally, we can invoke the Cram\'er-Wold device to obtain
	\begin{align*}
		\| \hat{\beta} - \beta \|_{2} \stackrel{p}{\rightarrow} 0 \qquad as \quad N\rightarrow\infty,
	\end{align*}
	as desired. \hfill$\square$


\subsection{Lemma}\label{appndx:ProofOfMainTheorems:ExogenousLinearRegression:PreliminaryResultsLemma}

Here we establish a lemma that is used repeatedly throughout the subsequent proofs in this paper.

\begin{lemma}\label{lemma:WellDefinedBread}
	Define $A\coloneqq \lim_{N\rightarrow\infty} \frac{1}{M} \sum_{k\in\mathcal{M}} E\Big[x_{M,k} x_{M,k}'\Big]$ and assume that Assumptions \ref{assm:ExogenousRegression:Denseness} and \ref{assm:ExogenousRegression:BoundedSupport} hold.
	\begin{itemize}
		\item [(i)] $A^{-1}\coloneqq \lim_{N\rightarrow\infty} \Big(\frac{1}{M} \sum_{k\in\mathcal{M}} E \big[ x_{M,k} x_{M,k}' \big]\Big)^{-1}$ exists with finite elements and positive definite.
		\item[(ii)] Suppose, moreover, that Assumption \ref{assm:ExogenousRegression:Sparseness} holds. Then, $\Big\| \Big( \frac{1}{M} \sum_{k\in\mathcal{M}_{N}^{}} x_{M,k}x_{M,k}' \Big)^{-1}
				- \Big( \frac{1}{M} \sum_{k\in\mathcal{M}_{N}^{}} E\big[ x_{M,k}x_{M,k}' \big] \Big)^{-1} \Big\|_{F}
				\stackrel{p}{\rightarrow} 0$.
	\end{itemize}
\end{lemma}

\begin{proof}
	\noindent\textbf{$(i)$} The fact that it is positive definite follows from Assumption \ref{assm:ExogenousRegression:BoundedSupport}. The fact that the elements are finite is proved by considering element-by-element convergence. Let $x_{k,i}$ denote the $i$-th element of $x_{M,k}$. Then the $(i,j)$ entry of $\frac{1}{M} \sum_{k\in\mathcal{M}} E \big[ x_{M,k} x_{M,k}' \big]$ is given by: $\frac{1}{M} \sum_{k\in\mathcal{M}} E \big[ x_{k,i} x_{k,j} \big]$.
	
	We write the $(i,j)$ entry of $A$ as $A_{i,j}$. 
	
	From Assumption \ref{assm:ExogenousRegression:BoundedSupport}, there exists a nonnegative finite constant $C_{0,1}$ such that 
	\begin{align*}
		C_{0,1} = \sup_{N\geq 1} \max_{m\in\mathcal{M}_{N}} E\big[ x_{m,i}x_{m,j} \big],
	\end{align*}
	so that
	\begin{align*}
		A_{i,j} &= \lim_{N\rightarrow\infty} \frac{1}{M} \sum_{k\in\mathcal{M}_{N}^{}} \underbrace{E\big[ x_{k,i}x_{k,j} \big]}_{\leq C_{0,1}} \\
		&\leq \lim_{N\rightarrow\infty} \frac{1}{M} \sum_{k\in\mathcal{M}_{N}^{}} C_{0,1} \\
		&= C_{0,1} \lim_{N\rightarrow\infty} \frac{1}{M} \underbrace{\sum_{k\in\mathcal{M}_{N}^{}} 1}_{M} \\
		&= C_{0,1}.
	\end{align*}
	Hence $A_{i,j}$ exists with being finite. By repeating the same argument for all $i,j=1,\ldots, K$, it holds that $A$ exists with finite elements.
	
	\noindent\textbf{$(ii)$}
	To begin with, observe that
	\begin{align*}
		&\Big\| \Big( \frac{1}{M} \sum_{k\in\mathcal{M}_{N}^{}} x_{M,k}x_{M,k}' \Big)^{-1}
		- \Big( \frac{1}{M} \sum_{k\in\mathcal{M}_{N}^{}} E\big[ x_{M,k}x_{M,k}' \big] \Big)^{-1} \Big\|_{F} \\
		&= \Big\| \Big( \frac{1}{M} \sum_{k\in\mathcal{M}_{N}^{}} x_{M,k}x_{M,k}' \Big)^{-1}
				- A^{-1} + A^{-1} - \Big( \frac{1}{M} \sum_{k\in\mathcal{M}_{N}^{}} E\big[ x_{M,k}x_{M,k}' \big] \Big)^{-1} \Big\|_{F} \\
		&\leq \Big\| \Big( \frac{1}{M} \sum_{k\in\mathcal{M}_{N}^{}} x_{M,k}x_{M,k}' \Big)^{-1}
				- A^{-1} \Big\|_{F} + \Big\| A^{-1} - \Big( \frac{1}{M} \sum_{k\in\mathcal{M}_{N}^{}} E\big[ x_{M,k}x_{M,k}' \big] \Big)^{-1} \Big\|_{F}.
	\end{align*}
Note that convergence of the second term follows from (i). Hence, we wish to prove that:
			\begin{align*}
				\Big\| \frac{1}{M} \sum_{k\in\mathcal{M}_{N}^{}} x_{M,k}x_{M,k}' - A^{} \Big\|_{F} \stackrel{p}{\rightarrow} 0
			\end{align*}
			
			\quad To do so, we follow a strategy employed in \citet{Aronow_et_al-2015} and \citet{Tabord-Meehan-2019}. 
			In light of (i), it remains to show
			\begin{align*}
				Var\Big( \frac{1}{M} \sum_{k\in\mathcal{M}_{N}^{}} x_{M,k}x_{M,k}' \Big) \rightarrow 0.
			\end{align*}
			As in (i), we consider the element-by-element convergence, using the same notation. The variance can be expressed as a sum of covariances:
			\begin{align*}
				Var\Big( \frac{1}{M} \sum_{k\in\mathcal{M}_{N}^{}} x_{k,i}x_{k,j} \Big)
				&= \frac{1}{M^{2}} \sum_{m\in\mathcal{M}_{N}} \sum_{m'\in\mathcal{M}_{N}} Cov\big(x_{m,i}x_{m,j}, x_{m',i}x_{m',j}\big) \\
				&= \frac{1}{M^{2}} \sum_{s\geq 0} \sum_{m\in\mathcal{M}_{N}} \sum_{m'\in\mathcal{M}_{N}^{\partial}(m;s)} Cov\big(x_{m,i}x_{m,j}, x_{m',i}x_{m',j}\big).
			\end{align*}
			Again from Assumption \ref{assm:ExogenousRegression:BoundedSupport}, there exists a nonnegative finite constant $C_{0,2}$ such that 
			\begin{align*}
				C_{0,2} = \sup_{N\geq 1} \max_{m,m'\in\mathcal{M}_{N}} Cov\left(x_{m,i}x_{m,j}, x_{m',i}x_{m',j}\right).
			\end{align*}
			Hence,
				\begin{align*}
				Var\Big( \frac{1}{M} \sum_{k\in\mathcal{M}_{N}^{}} x_{k,i}x_{k,j} \Big)
				&\leq \frac{1}{M^{2}} \sum_{s\geq 0} \sum_{m\in\mathcal{M}_{N}} \sum_{m'\in\mathcal{M}_{N}^{\partial}(m;s)} C_{0,2} \displaybreak[0]\\
				&= \frac{C_{0,2}}{M^{2}} \sum_{s\geq 0} \sum_{m\in\mathcal{M}_{N}} \sum_{m'\in\mathcal{M}_{N}^{\partial}(m;s)} 1 \displaybreak[0]\\
				&= \frac{C_{0,2}}{M^{2}} \sum_{s\geq 0} M \delta_{M}^{\partial}(s;1) \displaybreak[0]\\
				&= C_{0,2} \underbrace{\frac{1}{M^{}} \sum_{s\geq 0} \delta_{M}^{\partial}(s;1)}_{\rightarrow 0} \displaybreak[0]\\
				&\rightarrow 0,
			\end{align*}
			where the last implication is due to Assumption \ref{assm:ExogenousRegression:Sparseness}. By repeating the same argument for all $i,j=1,\ldots,K$, we obtain
			\begin{align*}
				Var\Big( \frac{1}{M} \sum_{k\in\mathcal{M}_{N}^{}} x_{M,k}x_{M,k}' \Big) \rightarrow 0.
			\end{align*}
			
			\quad Now, by the Chebyshev's inequality, we arrive at 
			\begin{align*}
				\Big\| \frac{1}{M} \sum_{k\in\mathcal{M}_{N}^{}} x_{M,k}x_{M,k}' - A^{} \Big\|_{F} \stackrel{p}{\rightarrow} 0.
			\end{align*}
			Furthermore, applying the Continuous Mapping Theorem yields 
			\begin{align*}
				\Big\| \Big( \frac{1}{M} \sum_{k\in\mathcal{M}_{N}^{}} x_{M,k}x_{M,k}' \Big)^{-1} - A^{-1} \Big\|_{F} \stackrel{p}{\rightarrow} 0,
			\end{align*}
			obtaining the result. Therefore,
	\begin{align*}
		\Big\| \Big( \frac{1}{M} \sum_{k\in\mathcal{M}_{N}^{}} x_{M,k}x_{M,k}' \Big)^{-1}
		- \Big( \frac{1}{M} \sum_{k\in\mathcal{M}_{N}^{}} E\big[ x_{M,k}x_{M,k}' \big] \Big)^{-1} \Big\|_{F}
		 \stackrel{p}{\rightarrow} 0,
	\end{align*}
	as desired. 
\end{proof}


\subsection{Asymptotic Normality of $\hat{\beta}$}\label{appndx:ProofOfMainTheorems:ExogenousLinearRegression:AsymptoticNormalityOfBetaHat}

In this subsection, we prove Theorem \ref{theo:ExogenousRegression:CLT} under a slightly milder condition than Assumption \ref{assm:AsymptRateOfVariance}.

\begin{assumption}[Growth Rates of Variances]\label{assm:AsymptRateOfVariance2}
	There exists a sequence of (possibly random) positive numbers, $\{\pi_{N,M}\}_{N>0}$, such that
	\begin{align*}
		\frac{\sigma_{M}^{2}}{\pi_{N,M}\tau_{M}^{2}} \stackrel{a.s.}{\rightarrow} 1 \quad\quad \text{as}\quad N\rightarrow\infty.
	\end{align*}
\end{assumption}

When $\pi_{N,M}=1$, this assumption simplifies to Assumption \ref{assm:AsymptRateOfVariance}, which is used for the results in the main text.

For our proof of the asymptotic distribution of $\hat \beta$, we require that its asymptotic variance is well-defined. The first assumption, Assumption \ref{assm:ExogenousRegression:BoundedSupport}(a)-(b), is necessary for one of the matrices in the expression to be well-defined.\footnote{As pointed out in \citet{Tabord-Meehan-2019}, the bounded support assumption can be relaxed by imposing an alternative condition on higher-order moments (boundedness of the 16th order moment of $x_{M, m}$, in our case).} Part (c) assures that the middle part of the asymptotic variance is non-trivial.

When Assumption \ref{assm:AsymptRateOfVariance} is replaced by Assumption \ref{assm:AsymptRateOfVariance2}, Assumption \ref{assm:ExogenousRegression:BoundedSupport} must also be modified accordingly.

\begin{assumption}\label{assm:WellDefinedVariance2}
	$\lim_{N\rightarrow\infty} \frac{N \pi_{N,M}}{M^{2}} \sum_{m\in\mathcal{M}_{N}^{}} \sum_{m'\in\mathcal{M}_{N}^{}} E\big[ \varepsilon_{M, m}\varepsilon_{M, m'} x_{M, m}x_{M, m'}' \big]$ exists with finite elements.
\end{assumption}

An important comparison of Assumption \ref{assm:WellDefinedVariance2} can be made to the variety of assumptions used in the literature.

\begin{remark}\label{rmk:ExpressionOfAsymptoticVariance}
	The requirement on the behavior of $AVar(\hat{\beta})$ mirrors Assumptions 2.4, 2.5 and 2.6 of \citet{Tabord-Meehan-2019}: Assumption \ref{assm:WellDefinedVariance2} boils down to his Assumption 2.4, if it is well-defined with $\pi_{N,M}=\frac{M}{N}$; it reduces to his Assumption 2.5, if it is compatible with $\pi_{N,M}=\frac{M}{N^{2}}$; and it coincides with Assumption 2.6, if it is maintained with $\pi_{N,M}=\frac{M}{N^{r+1}}$ for $r\in[0,1]$. Moreover, if $AVar(\hat{\beta})$ is well-defined for $\pi_{N,M}=1$, the expression (\ref{eq:ExogenousRegression:AsymptoticVariance}) simplifies to the assumption that appears in Lemma 1 of \citet{Aronow_et_al-2015}. 
\end{remark}

\noindent 
\textbf{Proof of Theorem \ref{theo:ExogenousRegression:CLT}:}
	From \eqref{eq:ExogenousRegression:ExpressionOfBetaHat},
	\begin{align*}
		\sqrt{N} (\hat{\beta} - \beta)
		&= \Big( \frac{1}{M} \sum_{j\in\mathcal{M}_{N}^{}} x_{M, j}x_{M, j}' \Big)^{-1}
			\frac{\sqrt{N}}{M} \sum_{m\in\mathcal{M}_{N}^{}} Y_{M, m} \\
		&= \Big( \frac{1}{M} \sum_{j\in\mathcal{M}_{N}^{}} x_{M, j}x_{M, j}' \Big)^{-1}
			\frac{\sqrt{N}}{M} S_{M}.
	\end{align*}
	Since $\big( \frac{1}{M} \sum_{j\in\mathcal{M}_{N}^{}} x_{M, j}x_{M, j}' \big)^{-1}$ converges to a well-defined limit (Lemma \ref{lemma:WellDefinedBread}), the asymptotic distribution of $\sqrt{N} (\hat{\beta} - \beta)$ is dictated by that of $\frac{\sqrt{N}}{M} S_{M}$.
	
	First of all, we prove 
	\begin{align*}
		\frac{S_{M}^{u}}{\sigma_{M}^{}} \stackrel{d}{\rightarrow} \mathcal{N}(0,1),
	\end{align*}
	as $N\rightarrow\infty$. 
Consider the scenario that  $N\rightarrow\infty$, in which Assumption \ref{assm:ExogenousRegression:Denseness} implies $M\rightarrow\infty$.
	Denote $\tilde{S}_{M}^{u}\coloneqq \frac{S_{M}^{u}}{\sigma_{M}}$.
	Let $X$ be the $M\times K$ matrix that records the observed dyad-specific characteristics as defined in Section \ref{subsubsec:Set-upAndIdentification}, but here the subscript $M$ is omitted for notational simplicity. The value that $X$ takes is denoted by $x$.
	
	Under Assumptions \ref{assm:ExogenousRegression:ConditionalFiniteMoment2} and \ref{assm:Kojevnikov_et_al-2021:Assm3.4}, it holds by Theorem 3.2 of \citet{Kojevnikov_et_al-2021} that for any $\epsilon>0$, there exists $M_{0}>0$ such that for each $M>M_{0}$ and for each $x\in\mathbbm{R}^{M\times K}$,
	\begin{align}\label{eq:Appendix:AsymptoticNormalityOfBetaHat0}
		\sup_{t\in\mathbbm{R}} \big| \mathrm{Pr}( \tilde{S}_{M}^{u}\leq t \mid X=x ) - \Phi(t) \big| < \epsilon,
	\end{align}
	where $\Phi(\cdot)$ is the CDF of a standard Normal distribution. Then, by the law of total probability, we have
	\begin{align}\label{eq:Appendix:AsymptoticNormalityOfBetaHat1}
		\big| \mathrm{Pr}( \tilde{S}_{M}^{u}\leq t ) - \Phi(t) \big|
		&= \Big| \int \mathrm{Pr}( \tilde{S}_{M}^{u}\leq t \mid X=x) dF_{X}(x) - \Phi(t) \Big| \displaybreak[0]\nonumber\\
		&= \Big| \int \mathrm{Pr}( \tilde{S}_{M}^{u}\leq t \mid X=x) - \Phi(t) dF_{X}(x) \Big| \displaybreak[0]\nonumber\\
		&\leq \int \big| \mathrm{Pr}( \tilde{S}_{M}^{u}\leq t \mid X=x) - \Phi(t) \big| dF_{X}(x) \displaybreak[0]\nonumber\\
		&\leq \int \sup_{t\in\mathbbm{R}} \big| \mathrm{Pr}( \tilde{S}_{M}^{u}\leq t \mid X=x) - \Phi(t) \big| dF_{X}(x),
	\end{align}
	where $F_{X}(\cdot)$ denotes the probability distribution function of $X$.
	Now pick arbitrarily $\epsilon>0$.
	Then there exists $M_{0}>0$ such that for each $M>M_{0}$
	\begin{align}\label{eq:Appendix:AsymptoticNormalityOfBetaHat2}
		\big| \mathrm{Pr}( \tilde{S}_{M}^{u}\leq t ) - \Phi(t) \big|
		&\leq \int \underbrace{\sup_{t\in\mathbbm{R}} \big| \mathrm{Pr}( \tilde{S}_{M}^{u}\leq t \mid X=x) - \Phi(t) \big|}_{< \epsilon} dF_{X}(x) \displaybreak[0]\nonumber\\
		&\leq \int  \epsilon dF_{X}(x) \displaybreak[0]\nonumber\\
		&\leq \epsilon,
	\end{align}
	where the first and second inequalities come from \eqref{eq:Appendix:AsymptoticNormalityOfBetaHat1} and \eqref{eq:Appendix:AsymptoticNormalityOfBetaHat0}, respectively.
	Since the right hand side of \eqref{eq:Appendix:AsymptoticNormalityOfBetaHat2} does not depend on $t$, we then have that for each $M>M_{0}$,
	\begin{align*}
		\sup_{t\in \mathbbm{R}} \big| \mathrm{Pr}( \tilde{S}_{M}^{u}\leq t ) - \Phi(t) \big|
		&\leq \epsilon,
	\end{align*}
	which implies
	\begin{align*}
		\sup_{t\in \mathbbm{R}} \big| \mathrm{Pr}( \tilde{S}_{M}^{u}\leq t ) - \Phi(t) \big|
		&\rightarrow 0 \qquad as \quad M\rightarrow\infty,
	\end{align*}
	We have then shown that
	\begin{align*}
		\sup_{t\in \mathbbm{R}} \big| \mathrm{Pr}( \tilde{S}_{M}^{u}\leq t ) - \Phi(t) \big|
		&\rightarrow 0 \qquad as \quad N\rightarrow\infty,
	\end{align*}
	from which we obtain
	\begin{align*}
		\frac{S_{M}^{u}}{\sigma_{M}^{}} \stackrel{d}{\rightarrow} \mathcal{N}(0,1) \qquad as \quad N\rightarrow \infty. 
	\end{align*}
	
	Next this can be combined with Assumption \ref{assm:AsymptRateOfVariance2} by using the Slutsky's Theorem, yielding that
	\begin{align*}
		\frac{S_{M}^{u}}{\tau_{M}\sqrt{\pi_{N,M}}} \stackrel{d}{\rightarrow} \mathcal{N}(0,1) \qquad as \quad N\rightarrow\infty.
	\end{align*}
	Moreover, applying the Cram\'er-Wold device gives
	\begin{align*}
		\frac{\tau_{M}^{-1}}{\sqrt{\pi_{N,M}}} S_{M} \stackrel{d}{\rightarrow} \mathcal{N}(0,I_{K}) \qquad as \quad N\rightarrow\infty,
	\end{align*}
	where $I_{K}$ is the $K\times K$ identity matrix and $\tau_{M}$ is understood as the variance-covariance matrix.\footnote{To save notation, we use the same $\tau_{M}$ to denote the case of one-dimensional parameter and the case of multiple-dimensional parameters.}
	
	Now notice that we have
	\begin{align*}
		\sqrt{N} (\hat{\beta} - \beta)
		&= \Big( \frac{1}{M} \sum_{j\in\mathcal{M}_{N}^{}} x_{M, j}x_{M, j}' \Big)^{-1}
			\frac{\sqrt{N}}{M} \tau_{M} \sqrt{\pi_{N,M}} \underbrace{\frac{\tau_{M}^{-1}}{\sqrt{\pi_{N,M}}} S_{M, m}}_{\stackrel{d}{\rightarrow} \mathcal{N}(0,I_{K})}.
	\end{align*}
	Hence we obtain
	\begin{align*}
		\sqrt{N} \big( \hat{\beta}-\beta \big) \stackrel{d}{\rightarrow} \mathcal{N}(0, AVar(\hat{\beta})) \qquad as \quad N\rightarrow\infty,
	\end{align*}
	where
	{\footnotesize
	\begin{align*}
		AVar(\hat{\beta})
		\coloneqq \lim_{N\rightarrow\infty} N\pi_{N,M}^{} \Big( \sum_{k\in\mathcal{M}_{N}^{}} E\big[ x_{M, k}x_{M, k}' \big] \Big)^{-1}
				\Big( \sum_{m\in\mathcal{M}_{N}^{}} \sum_{m'\in\mathcal{M}_{N}^{}} E\big[ \varepsilon_{M, m}\varepsilon_{M, m'} x_{M, m}x_{M, m'}' \big] \Big)
				\Big( \sum_{k\in\mathcal{M}_{N}^{}} E\big[ x_{M, k}x_{M, k}' \big] \Big)^{-1},
	\end{align*}}
	which is well-defined due to Lemma \ref{lemma:WellDefinedBread} (i) along with Assumption \ref{assm:WellDefinedVariance2}. When $\pi_{N,M}=1$, this is the result in the main text. \hfill$\square$ 

\subsection{Lemma}

In the proof of Theorem \ref{thm:ExogenousRegression:ConsistencyOfTheNetwork-HACVarianceEstimator}, we make use of the following lemma from \citet{Kojevnikov_et_al-2021}, p.903:
	
\begin{lemma}\label{lemma:Kojevnikov_et_al-2021:p903}
	Define
	\begin{align*}
		H_{M}(s,r) \coloneqq \{ (m,j,k,l)\in\mathcal{M}_{N}^{4}: j\in\mathcal{M}_{N}(m; r),\ l\in\mathcal{M}_{N}(k; r),\ \rho_{M}(\{m,j\}, \{k,l\}) = s \}.
	\end{align*}
	Then
	\begin{align*}
		| H_{M}(s, r) | \leq 4 M c_{M}(s, r; 2).
	\end{align*}
\end{lemma}

\subsection{Consistency of $\widehat{Var}(\hat{\beta})$}\label{appndx:subsec:ProofOfMainTheorems:ExogenousLinearRegression:ConsistencyOfVHat}

\textbf{Proof of Theorem \ref{thm:ExogenousRegression:ConsistencyOfTheNetwork-HACVarianceEstimator}:}
	Denote the variance of $\frac{S_{M}}{\sqrt{M}}$ as $V_{N,M}\coloneqq Var\big( \frac{S_{M}}{\sqrt{M}} \big)$. It can readily be shown that $V_{N,M}$ takes the form of $V_{N,M} = \sum_{s\geq 0} \Omega_{N,M}(s)$, where 
	\begin{align*}
		\Omega_{N,M}(s) \coloneqq \frac{1}{M} \sum_{m\in \mathcal{M}_{N}} \sum_{j\in \mathcal{M}_{N}^{\partial}(m;s)} E\big[ Y_{M,m} Y_{M,j}' \big].
	\end{align*}
	Following \citet{Kojevnikov_et_al-2021}, we define the kernel heteroskedasticity and autocorrelation consistent (HAC) estimator of $V_{N,M}$ as $\hat{V}_{N,M} \coloneqq \sum_{s\geq 0} \omega_{M}(s) \hat{\Omega}_{N,M}(s)$,
	where $\omega_{M}(s) \coloneqq \omega\big(\frac{s}{b_{M}}\big)$ and 
	\begin{align*}
		\hat{\Omega}_{N,M}(s) \coloneqq \frac{1}{M} \sum_{m\in \mathcal{M}_{N}} \sum_{j\in \mathcal{M}_{N}^{\partial}(m;s)} \hat{Y}_{M,m} \hat{Y}_{M,j}'.
	\end{align*}
	Moreover, we define an empirical analogue of $V_{N,M}$, though infeasible, by $\tilde{V}_{N,M} \coloneqq \sum_{s\geq 0} \omega_{M}(s) \tilde{\Omega}_{N,M}(s)$, where 
	\begin{align*}
		\tilde{\Omega}_{N,M}(s) \coloneqq \frac{1}{M} \sum_{m\in \mathcal{M}_{N}} \sum_{j\in \mathcal{M}_{N}^{\partial}(m;s)} Y_{M,m} Y_{M,j}^{'}.
	\end{align*}
	Additionally, we denote a conditional version of $V_{N,M}$ by $V_{N,M}^{c}\coloneqq Var\big(\frac{S_{M}}{\sqrt{M}} \mid \mathcal{C}_{M}\big)$, i.e., $V_{N,M}^{c} = \sum_{s\geq 0} \Omega_{N,M}^{c}(s)$, where 
	\begin{align*}
		\Omega_{N,M}^{c}(s) \coloneqq \frac{1}{M} \sum_{m\in \mathcal{M}_{N}} \sum_{j\in \mathcal{M}_{N}^{\partial}(m;s)} E\big[ Y_{M,m} Y_{M,j}^{'} \mid \mathcal{C}_{M} \big].
	\end{align*}
	Notice that since $E\big[Y_{M,m}\mid \mathcal{C}_{M} \big] = 0$ $a.s.$, it follows from the law of total variance that $V_{N,M} = E\big[ V_{N,M}^{c} \big]$.
	Notice furthermore that it holds that 
	\begin{align*}
		Var(\hat{\beta}) 
		= \frac{N}{M} \Big( \frac{1}{M} \sum_{k\in\mathcal{M}_{N}^{}} E\big[ x_{M, k}x_{M, k}' \big] \Big)^{-1} 
			V_{N,M}
			\Big( \frac{1}{M} \sum_{k\in\mathcal{M}_{N}^{}} E\big[ x_{M, k}x_{M, k}' \big] \Big)^{-1},
	\end{align*}
	and
	\begin{align*}
		\widehat{Var}(\hat{\beta}) 
		= \frac{1}{M} \Big( \frac{1}{M} \sum_{k\in\mathcal{M}_{N}^{}} x_{M, k}x_{M, k}' \Big)^{-1}
			\hat{V}_{N,M} 
			\Big( \frac{1}{M} \sum_{k\in\mathcal{M}_{N}^{}} x_{M, k}x_{M, k}' \Big)^{-1}.
	\end{align*}
	Since
	\begin{align*}
		\| N \widehat{Var}(\hat{\beta}) - Var(\hat{\beta}) \|_{F} 
		&= \frac{N}{M} 
			\Bigg\| \Big( \frac{1}{M} \sum_{k\in\mathcal{M}_{N}^{}} x_{M, k}x_{M, k}' \Big)^{-1}
				\hat{V}_{N,M} 
				\Big( \frac{1}{M} \sum_{k\in\mathcal{M}_{N}^{}} x_{M, k}x_{M, k}' \Big)^{-1} \displaybreak[0]\\
		&\phantom{=\frac{N}{M} \Bigg\| }	- \Big( \frac{1}{M} \sum_{k\in\mathcal{M}_{N}^{}} E\big[ x_{M, k}x_{M, k}' \big] \Big)^{-1} 
					V_{N,M}
					\Big( \frac{1}{M} \sum_{k\in\mathcal{M}_{N}^{}} E\big[ x_{M, k}x_{M, k}' \big] \Big)^{-1}
			\Bigg\|_{F}, 
	\end{align*}
	and $\frac{N}{M}$ is bounded due to Assumption \ref{assm:ExogenousRegression:Denseness}, it thus suffices to show that 
	\begin{itemize}
		\item[(i)] $\Big\| \big( \frac{1}{M} \sum_{k\in\mathcal{M}_{N}^{}} x_{M, k}x_{M, k}' \big)^{-1} - \big( \frac{1}{M} \sum_{k\in\mathcal{M}_{N}^{}} E\big[ x_{M, k}x_{M, k}' \big] \big)^{-1} \Big\|_{F} \stackrel{p}{\rightarrow} 0$;
		\item[(ii)] $\| \hat{V}_{N,M} - V_{N,M} \|_{F} \stackrel{p}{\rightarrow} 0$.
	\end{itemize}
	Part (i) is already shown in Lemma \ref{lemma:WellDefinedBread} (ii). Hence, it remains to prove Part (ii).
	
	To begin with, observe that by the technique of add and subtract as well as the triangular inequality,
	\begin{align*}
		\| \hat{V}_{N,M} - V_{N,M} \|_{F} 
		&= \| \hat{V}_{N,M} - \tilde{V}_{N,M} + \tilde{V}_{N,M} - V^{c}_{N,M} + V^{c}_{N,M} - V_{N,M} \|_{F} \displaybreak[0]\\
		&\leq \| \hat{V}_{N,M} - \tilde{V}_{N,M} \|_{F} + \| \tilde{V}_{N,M} - V^{c}_{N,M} \|_{F} + \| V^{c}_{N,M} - V_{N,M} \|_{F}.
	\end{align*}
	We thus aim to prove
	\begin{itemize}
		\item[(1)] $\| V^{c}_{N,M} - V_{N,M} \|_{F} \stackrel{p}{\rightarrow} 0$;
		\item[(2)] $\| \tilde{V}_{N,M} - V^{c}_{N,M} \|_{F} \stackrel{p}{\rightarrow} 0$;
		\item[(3)] $\| \hat{V}_{N,M} - \tilde{V}_{N,M} \|_{F} \stackrel{p}{\rightarrow} 0$.
	\end{itemize}
	We start with:
	\begin{itemize}
		\item[(1)] $\| V^{c}_{N,M} - V_{N,M} \|_{F} \stackrel{p}{\rightarrow} 0$:\\
			The proof proceeds in multiple steps:
			\begin{itemize}
				\item[(a)] $E\big[ \big\| V_{N,M}^{c} - V_{N,M} \big\|_{F}^{2} \big] \rightarrow 0$;
				\item[(b)] $\big\| V_{N,M}^{c} - V_{N,M} \big\|_{F} \stackrel{p}{\rightarrow} 0$.
			\end{itemize}
We begin with:
			\begin{itemize}
				\item[(a)] $E\big[ \big\| V_{N,M}^{c} - V_{N,M} \big\|_{F}^{2} \big] \rightarrow 0$:\\
					We prove this by showing the element-wise convergence. With a slight abuse of notation, we denote the $(a,b)$ entry of $V_{N,M}^{c}$ and $V_{N,M}$ as $V_{a,b}^{c}$ and $V_{a,b}$, respectively. Then it is enough to verify that 
					\begin{align*}
						E\big[ (V_{a,b}^{c} - V_{a,b})^{2} \big] \rightarrow 0.
					\end{align*} 
					Notice that $V_{a,b}^{c}$ and $V_{a,b}$ are given by 
					\begin{align*}
						V_{a,b}^{c} = \sum_{s\geq0} \frac{1}{M} \sum_{m\in \mathcal{M}_{N}} \sum_{j\in \mathcal{M}_{N}^{\partial}(m;s)} E\big[ Y_{m,a} Y_{j,b} \mid \mathcal{C}_{M} \big]
					\end{align*}
					and
					\begin{align*}
						V_{a,b}^{} = \sum_{s\geq0} \frac{1}{M} \sum_{m\in \mathcal{M}_{N}} \sum_{j\in \mathcal{M}_{N}^{\partial}(m;s)} E\big[ Y_{m,a} Y_{j,b} \big],
					\end{align*}
					where $Y_{m,a}$ and $Y_{m,b}$ stand for the $a$-th and $b$-th element of $Y_{M,m}$, respectively.
					Note moreover that $E[V_{a,b}^{c}]=V_{a,b}^{}$.
					Hence, we can write
					\begin{align*}
						E\big[ (V_{a,b}^{c} - V_{a,b})^{2} \big]
						&= Var(V_{a,b}^{c}) \displaybreak[0]\\
						&= E\big[ (V_{a,b}^{c})^{2} \big] - \big( V_{a,b} \big)^{2} \displaybreak[0]\\
						&\leq E\big[ (V_{a,b}^{c})^{2} \big].
					\end{align*}
					Observe that
					\begin{align*}
						E\big[ (V_{a,b}^{c})^{2} \big]
						&= E\Big[ \Big(\sum_{s\geq0} \frac{1}{M} \sum_{m\in \mathcal{M}_{N}} \sum_{j\in \mathcal{M}_{N}^{\partial}(m;s)} E\big[ Y_{m,a} Y_{j,b} \mid \mathcal{C}_{M} \big] \Big)^{2} \Big] \displaybreak[0]\\
						&= E\Big[ \frac{1}{M^{2}} \sum_{s\geq0} \sum_{m\in \mathcal{M}_{N}} \sum_{j\in \mathcal{M}_{N}^{\partial}(m;s)} \sum_{t\geq0} \sum_{k\in \mathcal{M}_{N}} \sum_{l\in \mathcal{M}_{N}^{\partial}(k;t)} E\big[ Y_{m,a} Y_{j,b} \mid \mathcal{C}_{M} \big] E\big[ Y_{k,a} Y_{l,b} \mid \mathcal{C}_{M} \big] \Big] \displaybreak[0]\\
						&= \frac{1}{M^{2}} \sum_{s\geq0} \sum_{m\in \mathcal{M}_{N}} \sum_{j\in \mathcal{M}_{N}^{\partial}(m;s)} \sum_{t\geq0} \sum_{k\in \mathcal{M}_{N}} \sum_{l\in \mathcal{M}_{N}^{\partial}(k;t)} E\Big[ E\big[ Y_{m,a} Y_{j,b} \mid \mathcal{C}_{M} \big] E\big[ Y_{k,a} Y_{l,b} \mid \mathcal{C}_{M} \big] \Big].
					\end{align*}
					
					\quad By the Cauchy-Schwartz inequality,
					\begin{align*}
						E\big[ \varepsilon_{m} \varepsilon_{j} \mid \mathcal{C}_{M} \big] 
						\leq \big( E\big[ \varepsilon_{m}^{2} \mid \mathcal{C}_{M} \big] \big)^{\frac{1}{2}}
							\big( E\big[ \varepsilon_{m}^{2} \mid \mathcal{C}_{M} \big] \big)^{\frac{1}{2}} ,
					\end{align*}
					it then follows from from Assumption \ref{assm:Kojevnikov_et_al-2021:Assm4.1} (a) that there exists an a.s.-bounded function $\bar{C}_{1}$ such that $E\big[ \varepsilon_{m} \varepsilon_{j} \mid \mathcal{C}_{M} \big] \leq \bar{C}_{1}\ a.s.$ Similarly, we have an a.s.-bounded function $\bar{C}_{2}$ such that $E\big[ \varepsilon_{k} \varepsilon_{l} \mid \mathcal{C}_{M} \big] \leq \bar{C}_{2}\ a.s.$ Then,
					\begin{align*}
						E\big[ Y_{m,a} Y_{j,b} \mid \mathcal{C}_{M} \big]
						&= E\big[ \varepsilon_{m}\varepsilon_{j} x_{m,a}x_{j,b} \mid \mathcal{C}_{M} \big] \displaybreak[0]\\
						&= x_{m,a}x_{j,b} \underbrace{E\big[ \varepsilon_{m}\varepsilon_{j} \mid \mathcal{C}_{M} \big]}_{\leq\bar{C}_{1}} \displaybreak[0]\\
						&\leq x_{m,a}x_{j,b}\bar{C}_{1} \quad a.s.,
					\end{align*}
					where $x_{m,a}$ represents the $a$-th element of $x_{M,m}$ and $x_{j,b}$ the $b$-th element of $x_{M,j}$.
					Analogously, one obtains $E\big[ Y_{k,a} Y_{l,b} \mid \mathcal{C}_{M} \big]\leq x_{k,a}x_{l,b} \bar{C}_{2}\ a.s.$ 
					Once again, through the multiple application of the Cauchy-Schwartz inequality, it follows that 
					\begin{align*}
						E\big[ E\big[ Y_{m,a} Y_{j,b} \mid \mathcal{C}_{M} \big] E\big[ Y_{k,a} Y_{l,b} \mid \mathcal{C}_{M} \big] \big]
						&\leq E\big[ \bar{C}_{1}\bar{C}_{2} x_{m,a}x_{j,b}x_{k,a}x_{l,b} \big] \displaybreak[0]\\
						&\leq \big( E\big[ (\bar{C}_{1}\bar{C}_{2})^{2}  \big] \big)^{\frac{1}{2}} \big( E\big[ (x_{m,a}x_{j,b}x_{k,a}x_{l,b})^{2} \big] \big)^{\frac{1}{2}} \displaybreak[0]\\
						&\leq \big( E\big[ (\bar{C}_{1})^{4}  \big] \big)^{\frac{1}{4}} \big( E\big[ (\bar{C}_{2})^{4}  \big] \big)^{\frac{1}{4}} \displaybreak[0]\\
						&\phantom{\leq\quad} \times \big( E\big[ x_{m,a}^{8} \big] \big)^{\frac{1}{8}} \big( E\big[ x_{j,b}^{8} \big] \big)^{\frac{1}{8}} 
							 \big( E\big[ x_{l,a}^{8} \big] \big)^{\frac{1}{8}} \big( E\big[ x_{l,b}^{8} \big] \big)^{\frac{1}{8}}.
					\end{align*}
					We note here that Assumption \ref{assm:ExogenousRegression:BoundedSupport} ensures that there exists a nonnegative finite constant $C_{m,a}$ such that $E\big[ x_{l,a}^{8} \big]<C_{m,a}$, with the same argument holding true for $x_{j,b}$, $x_{k,a}$ and $x_{l,b}$ as well. Hence,
					\begin{align*}
						E\big[ E\big[ Y_{m,a} Y_{j,b} \mid \mathcal{C}_{M} \big] E\big[ Y_{k,a} Y_{l,b} \mid \mathcal{C}_{M} \big] \big]
						&\leq \bar{C}_{},
					\end{align*}
					where $\bar{C}_{}$ is a nonnegative finite constant that is appropriately defined.
					
					\quad Substituting this into the inequality above, 
					\begin{align*}
						E\big[ (V_{a,b}^{c})^{2} \big]
						&\leq \frac{1}{M^{2}} \sum_{s\geq0} \sum_{m\in \mathcal{M}_{N}} \sum_{j\in \mathcal{M}_{N}^{\partial}(m;s)} \sum_{t\geq0} \sum_{k\in \mathcal{M}_{N}} \sum_{l\in \mathcal{M}_{N}^{\partial}(k;t)} \bar{C} \displaybreak[0]\\
						&= \frac{\bar{C}}{M^{2}} \sum_{s\geq0} \sum_{m\in \mathcal{M}_{N}} \sum_{j\in \mathcal{M}_{N}^{\partial}(m;s)} \sum_{t\geq0} \sum_{k\in \mathcal{M}_{N}} \sum_{l\in \mathcal{M}_{N}^{\partial}(k;t)} 1 \displaybreak[0]\\
						&= \frac{\bar{C}}{M^{2}} \sum_{s\geq0} \sum_{(m,j,k,l)\in H_{M}(s;b_{M})} 1 \displaybreak[0]\\
						&= \frac{\bar{C}}{M^{2}} \sum_{s\geq0} \underbrace{|H_{M}(s;b_{M})|}_{\leq 4 M c_{M}(s,b_{M}; 2)} \displaybreak[0]\\
						&\leq \frac{\bar{C}}{M^{2}} \sum_{s\geq0} 4 M c_{M}(s,b_{M}; 2) \displaybreak[0]\\
						&= 4\bar{C} \underbrace{\frac{1}{M} \sum_{s\geq0}c_{M}(s,b_{M}; 2)}_{\rightarrow 0} \displaybreak[0]\\
						&\rightarrow 0,
					\end{align*}
					where the second inequality comes from Lemma \ref{lemma:Kojevnikov_et_al-2021:p903}, and the last implication is due to Assumption \ref{assm:ExogenousRegression:Sparseness}. Therefore we have shown that
					\begin{align*}
						E\big[ (V_{a,b}^{c} - V_{a,b})^{2} \big] \rightarrow 0.
					\end{align*} 
					
					\quad By repeating the same argument for each $a,b=1,\ldots,K$, it follows that
					\begin{align*}
						E\big[ \|V_{N,M}^{c} - V_{M,M}\|_{F}^{2} \big] \rightarrow 0.
					\end{align*}

				\item[(b)] $\big\| V_{N,M}^{c} - V_{N,M} \big\|_{F} \stackrel{p}{\rightarrow} 0$:\\
					By the Chebyshev's inequality and the result of part (a), we complete part (1) as it follows that for any $\eta>0$,
					\begin{align*}
						\mathrm{Pr}(\big\| V_{N,M}^{c} - V_{N,M} \big\|_{F} > \eta)
						&< \frac{1}{\eta^{2}} \underbrace{E\Big[ \big\| V_{N,M}^{c} - V_{N,M} \big\|_{F}^{2} \Big]}_{\rightarrow 0} \displaybreak[0]\\
						&\rightarrow 0.
					\end{align*}
					
			\end{itemize}
		\item[(2)] $\| \tilde{V}_{N,M} - V^{c}_{N,M} \|_{F} \stackrel{p}{\rightarrow} 0$:\\
			This immediately follows from applying Proposition 4.1 of \citet{Kojevnikov_et_al-2021}\footnote{Notice that the definitions of $V_{N,M}$, $\hat{V}_{N,M}$, $\tilde{V}_{N,M}$ and $V_{N,M}^{c}$ are slightly different from those used in Proposition 4.1 of \citet{Kojevnikov_et_al-2021}.} and the Dominated Convergence Theorem in the Markov inequality: i.e., 
			\begin{align*}
				\mathrm{Pr}\big( \| \tilde{V}_{N,M} - V^{c}_{N,M} \|_{F} \geq \eta \big)
				&\leq \frac{1}{\eta} E\big[ \| \tilde{V}_{N,M} - V^{c}_{N,M} \|_{F} \big] \displaybreak[0]\\
				&= \frac{1}{\eta} E\big[ \underbrace{E\big[ \| \tilde{V}_{N,M} - V^{c}_{N,M} \|_{F} \mid \mathcal{C}_{M} \big]}_{\stackrel{a.s.}{\rightarrow} 0} \big] \displaybreak[0]\\
				&\rightarrow 0,
			\end{align*}
			for any $\eta>0$.
		\item[(3)] $\| \hat{V}_{N,M} - \tilde{V}_{N,M} \|_{F} \stackrel{p}{\rightarrow} 0$:\\
			First, 
			we have\footnote{To lighten the notational burden, we drop the $M$ subscript from $\{x_{M,m}\}_{m\in \mathcal{M}_{N}}$ and $\{\varepsilon_{M,m}\}_{m\in \mathcal{M}_{N}}$ in the rest of the proof.}
			{\footnotesize
			\begin{align*}
				\big\| \hat{V}_{N,M} - \tilde{V}_{N,M} \big\|_{F}
				&= \Bigg\| \sum_{s\geq 0} \omega_{M}(s) \frac{1}{M} \sum_{m\in \mathcal{M}_{N}} \sum_{j\in \mathcal{M}_{N}^{\partial}(m;s)} 
					\hat{\varepsilon}_{m} \hat{\varepsilon}_{j} x_{m} x_{j}' - \sum_{s\geq 0} \omega_{M}(s) \frac{1}{M} \sum_{m\in \mathcal{M}_{N}} \sum_{j\in \mathcal{M}_{N}^{\partial}(m;s)} \varepsilon_{m} \varepsilon_{j} x_{m} x_{j}' \Bigg\|_{F} \displaybreak[0]\\
				&= \Bigg\| \sum_{s\geq 0} \omega_{M}(s) \frac{1}{M} \sum_{m\in \mathcal{M}_{N}} \sum_{j\in \mathcal{M}_{N}^{\partial}(m;s)}   
					 \big( \hat{\varepsilon}_{m} \hat{\varepsilon}_{j} - \varepsilon_{m} \varepsilon_{j} \big) x_{m} x_{j}' \Bigg\|_{F} \displaybreak[0]\\
				&\leq \Bigg\| \sum_{s\geq 0} \underbrace{|\omega_{M}(s)|}_{\leq1} \frac{1}{M} \sum_{m\in \mathcal{M}_{N}} \sum_{j\in \mathcal{M}_{N}^{\partial}(m;s)}   
					 \big( \hat{\varepsilon}_{m} \hat{\varepsilon}_{j} - \varepsilon_{m} \varepsilon_{j} \big) x_{m} x_{j}' \Bigg\|_{F} \displaybreak[0]\\
				&\leq \Bigg\| \sum_{s\geq 0} \frac{1}{M} \sum_{m\in \mathcal{M}_{N}} \sum_{j\in \mathcal{M}_{N}^{\partial}(m;s)}   
					 \big( \hat{\varepsilon}_{m} \hat{\varepsilon}_{j} - \varepsilon_{m} \varepsilon_{j} \big) x_{m} x_{j}' \Bigg\|_{F} \displaybreak[0]\\
				&\leq \sum_{s\geq 0} \frac{1}{M} \sum_{m\in \mathcal{M}_{N}} \sum_{j\in \mathcal{M}_{N}^{\partial}(m;s)}   
					 \big| \hat{\varepsilon}_{m} \hat{\varepsilon}_{j} - \varepsilon_{m} \varepsilon_{j} \big| \big\| x_{m} x_{j}' \big\|_{F}.
 			\end{align*}
			}
			\quad Observe that, by definition, $\hat{\varepsilon}_{m}$ can be written as $\hat{\varepsilon}_{m} =  \varepsilon_{m} - x_{m}' (\hat{\beta} - \beta)$.
			Hence
			\begin{align*}
				\hat{\varepsilon}_{m} \hat{\varepsilon}_{j} - \varepsilon_{m} \varepsilon_{j}
				= - \varepsilon_{m} (\hat{\beta} - \beta)' x_{j}
					- x_{m}' (\hat{\beta} - \beta) \varepsilon_{j}
					+ x_{m}' (\hat{\beta} - \beta) (\hat{\beta} - \beta)' x_{j},
			\end{align*}
			so that by the triangular inequality,
			\begin{align*}
				\big| \hat{\varepsilon}_{m} \hat{\varepsilon}_{j} - \varepsilon_{m} \varepsilon_{j} \big| 
				&\leq \big\| \hat{\beta} - \beta \big\|_{2} \big\| x_{j} \big\|_{2} \big| \varepsilon_{m} \big| 
					+ \big\| \hat{\beta} - \beta \big\|_{2} \big\| x_{m} \big\|_{2} \big|\varepsilon_{j} \big|
					+ \big\| \hat{\beta} - \beta \big\|_{2}^{2} \big\| x_{m} \big\|_{2} \big\| x_{j} \big\|_{2},
			\end{align*}
			for each $m,j\in\mathcal{M}_{N}$.
			Hence $\| \hat{V}_{N,M} - \tilde{V}_{N,M} \|_{F}$ can be bounded as
			{\footnotesize
			\begin{align*}
				&\big\| \hat{V}_{N,M} - \tilde{V}_{N,M} \big\|_{F} \displaybreak[0]\\
				&\leq \big\| \hat{\beta} - \beta \big\|_{2} \frac{1}{M^{}} \sum_{s\geq 0} \sum_{m\in \mathcal{M}_{N}} \sum_{j\in \mathcal{M}_{N}^{\partial}(m;s)} \| x_{m} \|_{2} \| x_{j} \|_{2}^{2} | \varepsilon_{m} | 
					+ \big\| \hat{\beta} - \beta \big\|_{2} \frac{1}{M^{}} \sum_{s\geq 0} \sum_{m\in \mathcal{M}_{N}} \sum_{j\in \mathcal{M}_{N}^{\partial}(m;s)} \| x_{m} \|_{2}^{2} \| x_{j} \|_{2} |\varepsilon_{j} | \displaybreak[0]\\
				&\qquad	+ \big\| \hat{\beta} - \beta \big\|_{2}^{2} \frac{1}{M^{}} \sum_{s\geq 0} \sum_{m\in \mathcal{M}_{N}} \sum_{j\in \mathcal{M}_{N}^{\partial}(m;s)} \| x_{m} \|_{2}^{2} \| x_{j} \|_{2}^{2}.
			\end{align*}
			}
			
			\quad Denote
			\begin{align*}
				& R_{N,1} \coloneqq \frac{1}{M^{}} \sum_{s\geq 0} \sum_{m\in \mathcal{M}_{N}} \sum_{j\in \mathcal{M}_{N}^{\partial}(m;s)} \| x_{m} \|_{2} \| x_{j} \|_{2}^{2} | \varepsilon_{m} |; \displaybreak[0]\\
				& R_{N,2} \coloneqq \frac{1}{M^{}} \sum_{s\geq 0} \sum_{m\in \mathcal{M}_{N}} \sum_{j\in \mathcal{M}_{N}^{\partial}(m;s)} \allowbreak \| x_{m} \|_{2}^{2} \| x_{j} \|_{2} |\varepsilon_{j} |; \displaybreak[0]\\
				& R_{N,3} \coloneqq \frac{1}{M^{}} \sum_{s\geq 0} \sum_{m\in \mathcal{M}_{N}} \sum_{j\in \mathcal{M}_{N}^{\partial}(m;s)} \| x_{m} \|_{2}^{2} \| x_{j} \|_{2}^{2}.
			\end{align*}
			\quad Now, since by Theorem \ref{theo:ConsistencyofBetaHat}, $\| \hat{\beta} - \beta \|_{2}^{}\stackrel{p}{\rightarrow}0$,
			and the application of the Continuous Mapping Theorem yields $\| \hat{\beta} - \beta \|_{2}^{2}\stackrel{p}{\rightarrow}0$,
			it thus suffices to prove that each of $R_{N,1}$, $R_{N,2}$ and $R_{N,3}$ converges in probability to a finite number. In proving this, we follow a strategy employed in \citet{Aronow_et_al-2015} and \citet{Tabord-Meehan-2019}.
			
			\quad First let us study the expectation of $R_{N,1}$. By applying the Cauchy-Schwartz inequality repeatedly, we have that
			\begin{align*}
				E\big[ R_{N,1} \big] 
				\leq \frac{1}{M^{}} \sum_{s\geq 0} \sum_{m\in \mathcal{M}_{N}} \sum_{j\in \mathcal{M}_{N}^{\partial}(m;s)} \big( \big(E\big[ \| x_{m} \|_{2}^{2} \big]\big)^{\frac{1}{2}} \big(E\big[ \| x_{j} \|_{2}^{8} \big]\big)^{\frac{1}{2}} \big)^{\frac{1}{2}} \big( E\big[ E\big[ |\varepsilon_{m}|^{2} \mid \mathcal{C}_{M}\big]\big] \big)^{\frac{1}{2}}.
			\end{align*}
			Here, in light of Assumption \ref{assm:Kojevnikov_et_al-2021:Assm4.1}, there exists an a.s.-bounded function $C_{1}$ such that $C_{1} = \sup_{N\geq1} \max_{m\in\mathcal{M}_{N}} E\big[ | \varepsilon_{m} |^{2} \mid \mathcal{C}_{M} \big]$,
			and moreover by Assumption \ref{assm:ExogenousRegression:BoundedSupport}, there exists a nonnegative finite number $C_{2}>0$ such that $C_{2} = \sup_{N\geq1} \max_{m\in\mathcal{M}_{N}} E\big[ \| x_{m} \|_{2}^{8} \big]$.
			With a slight abuse of notation, we have for every $N>0$
			\begin{align*}
				E\left[ R_{N,1} \right]
				&\leq \frac{1}{M^{}} \sum_{s\geq 0} \sum_{m\in \mathcal{M}_{N}} \sum_{j\in \mathcal{M}_{N}^{\partial}(m;s)} C_{1}C_{2} \displaybreak[0]\\
				&= \frac{C_{1}C_{2}}{M^{}} \sum_{s\geq 0} \sum_{m\in \mathcal{M}_{N}} \sum_{j\in \mathcal{M}_{N}^{\partial}(m;s)} 1 \displaybreak[0]\\
				&= C_{1}C_{2} \sum_{s\geq 0} \underbrace{\frac{1}{M^{}} \sum_{m\in \mathcal{M}_{N}} |\mathcal{M}_{N}^{\partial}(m;s)|}_{\delta_{M}^{\partial}(s; 1)} \displaybreak[0]\\
				&= C_{1}C_{2} \underbrace{ \sum_{s\geq 0} \delta_{M}^{\partial}(s; 1)}_{< \infty} \displaybreak[0]\\
				&< C,
			\end{align*}
			for some constant $C\in(0,\infty)$, where the last inequality is because of Assumption \ref{assm:ExogenousRegression:Sparseness}.
			
			\quad Next let us study the variance of $R_{N,1}$. It suffices to show that $E\big[ R_{N,1}^{2} \big] \rightarrow 0$,
			By the Cauchy-Schwartz inequality, it holds that
			{\footnotesize
			\begin{align*}
				E\big[ R_{N,1}^{2} \big]
				&= E\Big[ \frac{1}{M^{2}} \sum_{s\geq 0} \sum_{m\in \mathcal{M}_{N}} \sum_{j\in \mathcal{M}_{N}^{\partial}(m;s)} \sum_{t\geq 0} \sum_{k\in \mathcal{M}_{N}} \sum_{l\in \mathcal{M}_{N}^{\partial}(k;t)} \| x_{m} \|_{2} \| x_{j} \|_{2}^{2} \| x_{k} \|_{2} \| x_{l} \|_{2}^{2} | \varepsilon_{m} | | \varepsilon_{k} |  \Big] \displaybreak[0]\\
				&\leq \frac{1}{M^{2}} \sum_{s\geq 0} \sum_{m\in \mathcal{M}_{N}} \sum_{j\in \mathcal{M}_{N}^{\partial}(m;s)} \sum_{t\geq 0} \sum_{k\in \mathcal{M}_{N}} \sum_{l\in \mathcal{M}_{N}^{\partial}(k;t)} \big( E\big[ \| x_{m} \|_{2}^{2} \| x_{j} \|_{2}^{4} \| x_{k} \|_{2}^{2} \| x_{l} \|_{2}^{4} \big] \big)^{\frac{1}{2}} \big( E\big[ | \varepsilon_{m} |^{2} | \varepsilon_{k} |^{2}  \big] \big)^{\frac{1}{2}}.
			\end{align*}
			}
			Here, by Assumption \ref{assm:ExogenousRegression:BoundedSupport} and the Cauchy-Schwartz inequality, there exists a nonnegative finite constant $C_{3}>0$ such that $C_{3} = \sup_{N\leq1} \max_{m,j,k,l \in\mathcal{M}_{N}} E\big[ \| x_{m} \|_{2}^{2} \| x_{j} \|_{2}^{4} \| x_{k} \|_{2}^{2} \| x_{l} \|_{2}^{4} \big]$.
			Then, with a slight abuse of notation in writing $C_{3}^{\frac{1}{2}}$ as $C_{3}$, we have
			\begin{align*}
				E\big[ R_{N,1}^{2} \big] 
				&= \frac{C_{3}}{M^{2}} \sum_{s\geq 0} \sum_{m\in \mathcal{M}_{N}} \allowbreak \sum_{j\in \mathcal{M}_{N}^{\partial}(m;s)} \sum_{t\geq 0} \sum_{k\in \mathcal{M}_{N}} \sum_{l\in \mathcal{M}_{N}^{\partial}(k;t)} \big( E\big[ | \varepsilon_{m} |^{2} | \varepsilon_{k} |^{2}  \big] \big)^{\frac{1}{2}} \displaybreak[0]\\
				&= \frac{C_{3}}{M^{2}} \sum_{s\geq 0} \sum_{(m,j,k,l)\in H_{M}(s;b_{M})} \allowbreak \big( E\big[ E\big[ | \varepsilon_{m} |^{2} | \varepsilon_{k} |^{2}  \mid \mathcal{C}_{M} \big]\big]\big)^{\frac{1}{2}}.
			\end{align*}	
			Corollary A.2 of \citet{Kojevnikov_et_al-2021} shows that there exists a nonnegative finite constant $C_{4}$ such that $E\big[ | \varepsilon_{m} |^{2} | \varepsilon_{k} |^{2} \mid \mathcal{C}_{M} \big] \leq C_{4} \bar{\theta} \theta_{M,s}^{1-\frac{4}{p}}$, where $\bar{\theta} \coloneqq \sup_{M\geq 1}\max_{s\geq 1} \theta_{M,s}$.
			Upon applying Lemma \ref{lemma:Kojevnikov_et_al-2021:p903} from the Appendix, we obtain 
			\begin{align*}
				E\big[ R_{N,1}^{2} \big] 
				\leq \frac{C_{3}C_{4}'}{M^{2}} \sum_{s\geq 0} \big( E\big[ \theta_{M,s}^{1-\frac{4}{p}} \big] \big)^{\frac{1}{2} } 4M c_{M}(s,b_{M};2) 
				= \frac{4C_{3}C_{4}''}{M^{}} \sum_{s\geq 0} c_{M}(s,b_{M};2) \rightarrow 0,
			\end{align*}
			where we apply Assumption \ref{assm:ExogenousRegression:Sparseness} for the last implication, and $C_{4}'$ and $C_{4}''$ are nonnegative finite constants defined appropriately.
			Hence we have shown that $R_{N,1}$ converges to a finite constant. 
			
			\quad The proof of $R_{N,2}$ is analogous.
			
			\quad It remains to show that $R_{N,3}$ converges in probability to a finite constant. Let us first study the expectation of $R_{N,3}$. Observe that 
			\begin{align*}
				E\big[ R_{N,3} \big] 
				= \frac{1}{M^{}} \sum_{s\geq 0} \sum_{m\in \mathcal{M}_{N}} \sum_{j\in \mathcal{M}_{N}^{\partial}(m;s)} E\big[ \| x_{m} \|_{2}^{2} \| x_{j} \|_{2}^{2} \big].
			\end{align*}
			By Assumption \ref{assm:ExogenousRegression:BoundedSupport}, there exists a nonnegative finite number $C_{5}>0$ such that $C_{5} = \sup_{N\geq 1} \max_{m\in\mathcal{M}_{N}} E\big[ \| x_{m} \|_{2}^{2} \| x_{j} \|_{2}^{2} \big]$.
			Hence for every $N>0$,
			\begin{align*}
				E\big[ R_{N,3} \big] 
				&= \frac{1}{M^{}} \sum_{s\geq 0} \sum_{m\in \mathcal{M}_{N}} \sum_{j\in \mathcal{M}_{N}^{\partial}(m;s)} \underbrace{E\Big[ \| x_{m} \|_{2}^{2} \| x_{j} \|_{2}^{2} \Big]}_{\leq C_{5}} \displaybreak[0]\\
				&\leq \frac{1}{M^{}} \sum_{s\geq 0} \sum_{m\in \mathcal{M}_{N}} \sum_{j\in \mathcal{M}_{N}^{\partial}(m;s)} C_{5} \displaybreak[0]\\
				&= \frac{C_{5}}{M^{}} \sum_{s\geq 0} \sum_{m\in \mathcal{M}_{N}} \sum_{j\in \mathcal{M}_{N}^{\partial}(m;s)} 1 \displaybreak[0]\\
				&= C_{5} \sum_{s\geq 0} \underbrace{\frac{1}{M} \sum_{m\in \mathcal{M}_{N}} |\mathcal{M}_{N}^{\partial}(m;s)|}_{\delta_{M}^{\partial}(s; 1)} \displaybreak[0]\\
				&= C_{5} \underbrace{\sum_{s\geq 0} \delta_{M}^{\partial}(s; 1)}_{< \infty} \displaybreak[0]\\
				&< C,
			\end{align*}
			where we apply Assumption \ref{assm:ExogenousRegression:Sparseness} in the last implication and a constant $C\in(0,\infty)$ is appropriately defined.
			
			\quad Next let us consider the variance of $R_{M,3}$:
			\begin{align*}
				E\big[ R_{N,3}^{2} \big] 
				&= \frac{1}{M^{2}} \sum_{s\geq 0} \sum_{m\in \mathcal{M}_{N}} \sum_{j\in \mathcal{M}_{N}^{\partial}(m;s)}  \sum_{t\geq 0} \sum_{k\in \mathcal{M}_{N}} \sum_{l\in \mathcal{M}_{N}^{\partial}(k;t)} E\big[ \| x_{m} \|_{2}^{2} \| x_{j} \|_{2}^{2} \| x_{k} \|_{2}^{2} \| x_{l} \|_{2}^{2} \big]. 
			\end{align*}
			Once again, Assumption \ref{assm:ExogenousRegression:BoundedSupport} and the Cauchy-Schwartz inequality imply that there exists a nonnegative finite number $C_{6}>0$ such that $C_{6} = \sup_{N\geq 1} \max_{m,j,k,l\in\mathcal{M}_{N}} \allowbreak E\big[ \| x_{m} \|_{2}^{2} \| x_{j} \|_{2}^{2} \| x_{k} \|_{2}^{2} \| x_{l} \|_{2}^{2} \big]$.
			Then by Lemma \ref{lemma:Kojevnikov_et_al-2021:p903}, 
			\begin{align*}
				E\big[ R_{N,3}^{2} \big] 
				\leq \frac{1}{M^{2}} \sum_{s\geq 0} \sum_{m\in \mathcal{M}_{N}} \sum_{j\in \mathcal{M}_{N}^{\partial}(m;s)}  \sum_{t\geq 0} \sum_{k\in \mathcal{M}_{N}} \sum_{l\in \mathcal{M}_{N}^{\partial}(k;t)} C_{6} 
				= \frac{4C_{6}}{M^{}} \sum_{s\geq 0} c_{M}(s,b_{M};2) \rightarrow 0,
			\end{align*}
			where the last implication is a consequence of Assumption \ref{assm:ExogenousRegression:Sparseness} (ii).
			
			\quad Therefore we have shown that
			{\footnotesize
			\begin{align*}
				&\big\| \hat{V}_{N,M} - \tilde{V}_{N,M} \big\|_{F} \displaybreak[0]\\
				&\leq \big\| \hat{\beta} - \beta \big\|_{2} \underbrace{ \frac{1}{M^{}} \sum_{s\geq 0} \sum_{m\in \mathcal{M}_{N}} \sum_{j\in \mathcal{M}_{N}^{\partial}(m;s)} \| x_{m} \|_{2} \| x_{j} \|_{2}^{2} | \varepsilon_{m} |}_{R_{M,1}} 
					+ \big\| \hat{\beta} - \beta \big\|_{2} \underbrace{ \frac{1}{M^{}} \sum_{s\geq 0} \sum_{m\in \mathcal{M}_{N}} \sum_{j\in \mathcal{M}_{N}^{\partial}(m;s)} \| x_{m} \|_{2}^{2} \| x_{j} \|_{2} |\varepsilon_{j} |}_{R_{M,2}} \displaybreak[0]\\
				&\qquad	+ \big\| \hat{\beta} - \beta \big\|_{2}^{2} \underbrace{\frac{1}{M^{}} \sum_{s\geq 0} \sum_{m\in \mathcal{M}_{N}} \sum_{j\in \mathcal{M}_{N}^{\partial}(m;s)} \| x_{m} \|_{2}^{2} \| x_{j} \|_{2}^{2}}_{R_{M,3}} \displaybreak[0]\\
				&= \underbrace{\big\| \hat{\beta} - \beta \big\|_{2}}_{\stackrel{p}{\rightarrow}0} \underbrace{R_{M,1}}_{< \infty} + \underbrace{\big\| \hat{\beta} - \beta \big\|_{2}}_{\stackrel{p}{\rightarrow}0} \underbrace{R_{M,2}}_{< \infty} + \underbrace{\big\| \hat{\beta} - \beta \big\|_{2}^{2}}_{\stackrel{p}{\rightarrow}0} \underbrace{R_{M,3}}_{< \infty} \stackrel{p}{\rightarrow}0,
			\end{align*}
			}
			which proves $\| \hat{V}_{N,M} - \tilde{V}_{N,M} \|_{F}\stackrel{p}{\rightarrow} 0$.

		\end{itemize}

		To sum up, combining parts (1), (2) and (3), we have 
		\begin{align*}
			\| \hat{V}_{N,M} - V_{N,M} \|_{F} 
			&\leq \underbrace{\| \hat{V}_{N,M} - \tilde{V}_{N,M} \|_{F}}_{\stackrel{p}{\rightarrow} 0} + \underbrace{\| \tilde{V}_{N,M} - V^{c}_{N,M} \|_{F}}_{\stackrel{p}{\rightarrow} 0} + \underbrace{\| V^{c}_{N,M} - V_{N,M} \|_{F}}_{\stackrel{p}{\rightarrow} 0} \displaybreak[0]\\
			&\stackrel{p}{\rightarrow} 0,
		\end{align*}
		which completes the proof.	\hfill$\square$

\subsection{Corollary \ref{coro:ExogenousRegression:NetworkHACEstimators:InconsistencyResult}}\label{appndx:subsec:ProofOfMainTheorems:ExogenousLinearRegression:InConsistencyOfVHatDyad}

\textbf{Proof of Corollary \ref{coro:ExogenousRegression:NetworkHACEstimators:InconsistencyResult}:}
	For simplicity we denote
	\begin{align*}
		\hat{V}_{N,M}^{Dyad}
		\coloneqq \Big( \sum_{m\in\mathcal{M}_{N}^{}} x_{m}x_{m}' \Big)^{-1}
				\Big( \sum_{m\in\mathcal{M}_{N}^{}} \sum_{m'\in\mathcal{M}_{N}^{}} \mathbbm{1}_{m,m'} \hat{\varepsilon}_{m}\hat{\varepsilon}_{m'} x_{m}x_{m'}' \Big)
				\Big( \sum_{m\in\mathcal{M}_{N}^{}} x_{m}x_{m}' \Big)^{-1},
	\end{align*}
	and
	\begin{align*}
		\hat{V}_{N,M}^{Network}
		\coloneqq \Big( \sum_{m\in\mathcal{M}_{N}^{}} x_{m}x_{m}' \Big)^{-1}
				\Big( \sum_{m\in\mathcal{M}_{N}^{}} \sum_{m'\in\mathcal{M}_{N}^{}} h_{m,m'} \hat{\varepsilon}_{m}\hat{\varepsilon}_{m'} x_{m}x_{m'}' \Big)
				\Big( \sum_{m\in\mathcal{M}_{N}^{}} x_{m}x_{m}' \Big)^{-1},\footnotemark
	\end{align*}
	where we choose the kernel function and the lag truncation parameter so that the weights become equal one for all active dyads: namely, we use the mean-shifted rectangular kernel with the lag truncation being the length of the longest path in the network.\footnotetext{For the sake of brevity, we suppress the $M$ from subscript throughout this proof.}
	Define moreover $\widetilde{Var}(\hat{\beta})$ to be the same variance as in the main text --- \\ $N \Big( \frac{1}{M}\sum_{k\in\mathcal{M}_{N}} E\big[x_{M,k}x_{M,k}' \big]\Big)^{-1} \Big( \frac{1}{M^{2}} \sum_{s\geq 0}\sum_{m\in\mathcal{M}_{N}} \sum_{m'\in\mathcal{M}_{N}^{\partial}(m;s)} E\big[Y_{M,m}Y_{M,m'}' \big] \Big) \Big( \frac{1}{M}\sum_{k\in\mathcal{M}_{N}} E\big[x_{M,k}x_{M,k}' \big]\Big)^{-1}$, but now applied to the network-regression model (\ref{eq:ExogenousRegression:PopulationNetworkRegression_a}) and (\ref{eq:ExogenousRegression:PopulationNetworkRegression_b}).
	By the triangular inequality,
	\begin{align*}
		\big\| N \hat{V}_{N,M}^{Network} - N \hat{V}_{N,M}^{Dyad} \big\|_{F}
		&= \big\| N \hat{V}_{N,M}^{Network} - \widetilde{Var}(\hat{\beta}) + \widetilde{Var}(\hat{\beta}) - N \hat{V}_{N,M}^{Dyad} \big\|_{F} \displaybreak[0]\\
		&\leq \big\| N \hat{V}_{N,M}^{Network} - \widetilde{Var}(\hat{\beta}) \big\|_{F} + \big\| \widetilde{Var}(\hat{\beta}) - N \hat{V}_{N,M}^{Dyad} \big\|_{F}.
	\end{align*}
	
	Since Theorem \ref{thm:ExogenousRegression:ConsistencyOfTheNetwork-HACVarianceEstimator} implies $\big\| N \hat{V}_{N,M}^{Network} - \widetilde{Var}(\hat{\beta}) \big\|_{F}\stackrel{p}{\rightarrow} 0$, then in the limit we are left with
	\begin{align}\label{eq:ExogenousRegression:NetworkHACEstimators:Corollary:InconsistencyResult}
		\big\| N \hat{V}_{N,M}^{Network} - N \hat{V}_{N,M}^{Dyad} \big\|_{F}
		&\leq \big\| \widetilde{Var}(\hat{\beta}) - N \hat{V}_{N,M}^{Dyad} \big\|_{F}.
	\end{align}
	Now we prove the statement by way of contradiction. Assume for the sake of contradiction that the dyadic-robust variance estimator $\hat{V}_{N,M}^{Dyad}$ is consistent, i.e., $\big\| \widetilde{Var}(\hat{\beta}) - N \hat{V}_{N,M}^{Dyad} \big\|_{F} \stackrel{p}{\rightarrow} 0$. This, combined with the inequality (\ref{eq:ExogenousRegression:NetworkHACEstimators:Corollary:InconsistencyResult}), implies $\big\| N \hat{V}_{N,M}^{Network} - N \hat{V}_{N,M}^{Dyad} \big\|_{F}\stackrel{p}{\rightarrow} 0$. Now, observe that
	{\footnotesize
	\begin{align*}
		&\big\| N \hat{V}_{N,M}^{Network} - N \hat{V}_{N,M}^{Dyad} \big\|_{F} \displaybreak[0]\nonumber\\
		&= \Big\| N \Big( \sum_{m\in\mathcal{M}_{N}^{}} x_{m}x_{m}' \Big)^{-1}
				\Big( \sum_{m\in\mathcal{M}_{N}^{}} \sum_{m'\in\mathcal{M}_{N}^{}} h_{m,m'} \hat{\varepsilon}_{m}\hat{\varepsilon}_{m'} x_{m}x_{m'}' \Big)
				\Big( \sum_{m\in\mathcal{M}_{N}^{}} x_{m}x_{m}' \Big)^{-1} \nonumber\\
		& \qquad  - N \Big( \sum_{m\in\mathcal{M}_{N}^{}} x_{m}x_{m}' \Big)^{-1}
				\Big( \sum_{m\in\mathcal{M}_{N}^{}} \sum_{m'\in\mathcal{M}_{N}^{}} \mathbbm{1}_{m,m'} \hat{\varepsilon}_{m}\hat{\varepsilon}_{m'} x_{m}x_{m'}' \Big)
				\Big( \sum_{m\in\mathcal{M}_{N}^{}} x_{m}x_{m}' \Big)^{-1} \Big\|_{F} \displaybreak[0]\nonumber\\
		&= \Big\| N \Big( \sum_{m\in\mathcal{M}_{N}^{}} x_{m}x_{m}' \Big)^{-1}
				\Big( \sum_{s\geq 0}\sum_{m\in\mathcal{M}_{N}^{}} \sum_{m'\in\mathcal{M}_{N}^{\partial}(m; s)} h_{m,m'} \hat{\varepsilon}_{m}\hat{\varepsilon}_{m'} x_{m}x_{m'}' \Big)
				\Big( \sum_{m\in\mathcal{M}_{N}^{}} x_{m}x_{m}' \Big)^{-1} \nonumber\\
		& \qquad  - N \Big( \sum_{m\in\mathcal{M}_{N}^{}} x_{m}x_{m}' \Big)^{-1}
				\Big( \sum_{s\geq 0}\sum_{m\in\mathcal{M}_{N}^{}} \sum_{m'\in\mathcal{M}_{N}^{\partial}(m; s)} \mathbbm{1}_{m,m'} \hat{\varepsilon}_{m}\hat{\varepsilon}_{m'} x_{m}x_{m'}' \Big)
				\Big( \sum_{m\in\mathcal{M}_{N}^{}} x_{m}x_{m}' \Big)^{-1} \Big\|_{F} \displaybreak[0]\nonumber\\
		&= \Big\| N \Big( \sum_{m\in\mathcal{M}_{N}^{}} x_{m}x_{m}' \Big)^{-1}
				\Big( \sum_{s=\{0,1\}}\sum_{m\in\mathcal{M}_{N}^{}} \sum_{m'\in\mathcal{M}_{N}^{\partial}(m; s)} \underbrace{\Big(h_{m,m'} -\mathbbm{1}_{m,m'}\Big)}_{0} \hat{\varepsilon}_{m}\hat{\varepsilon}_{m'} x_{m}x_{m'}'  \Big)
				\Big( \sum_{m\in\mathcal{M}_{N}^{}} x_{m}x_{m}' \Big)^{-1} \nonumber\\
		& \qquad + N \Big( \sum_{m\in\mathcal{M}_{N}^{}} x_{m}x_{m}' \Big)^{-1}
				\Big( \sum_{s\geq 2}\sum_{m\in\mathcal{M}_{N}^{}} \sum_{m'\in\mathcal{M}_{N}^{\partial}(m; s)} \Big( \underbrace{h_{m,m'}}_{1} -\underbrace{\mathbbm{1}_{m,m'}}_{0} \Big) \hat{\varepsilon}_{m}\hat{\varepsilon}_{m'} x_{m}x_{m'}' \Big)
				\Big( \sum_{m\in\mathcal{M}_{N}^{}} x_{m}x_{m}' \Big)^{-1} \Big\|_{F}  \displaybreak[0]\nonumber\\
		&= \Big\| N \Big( \sum_{m\in\mathcal{M}_{N}^{}} x_{m}x_{m}' \Big)^{-1}
				\Big( \sum_{s\geq 2}\sum_{m\in\mathcal{M}_{N}^{}} \sum_{m'\in\mathcal{M}_{N}^{\partial}(m; s)} \hat{\varepsilon}_{m}\hat{\varepsilon}_{m'} x_{m}x_{m'}' \Big)
				\Big( \sum_{m\in\mathcal{M}_{N}^{}} x_{m}x_{m}' \Big)^{-1} \Big\|_{F} \displaybreak[0]\nonumber\\
		&= \frac{N}{M} \Big\| \Big( \frac{1}{M} \sum_{m\in\mathcal{M}_{N}^{}} x_{m}x_{m}' \Big)^{-1}
				\Big( \frac{1}{M} \sum_{s\geq 2}\sum_{m\in\mathcal{M}_{N}^{}} \sum_{m'\in\mathcal{M}_{N}^{\partial}(m; s)} \hat{\varepsilon}_{m}\hat{\varepsilon}_{m'} x_{m}x_{m'}' \Big)
				\Big( \frac{1}{M} \sum_{m\in\mathcal{M}_{N}^{}} x_{m}x_{m}' \Big)^{-1} \Big\|_{F}, 
	\end{align*}
	}
We prove that the inside the Frobenius norm does not converge in probability to zero. 
	
	First it can immediately be shown, by Lemma \ref{lemma:WellDefinedBread} (ii), that the ``bread'' part $\big( \frac{1}{M} \sum_{m\in\mathcal{M}_{N}^{}} x_{m}x_{m}' \big)^{-1}$ converges to $\big( \frac{1}{M} \sum_{m\in\mathcal{M}_{N}^{}} E\big[ x_{m}x_{m}' \big] \big)^{-1}$.
	
	Next plugging the definition of $\hat{\varepsilon}$ into the middle part, we have
	{\footnotesize
	\begin{align*}
		&\frac{1}{M} \sum_{s\geq 2}\sum_{m\in\mathcal{M}_{N}^{}} \sum_{m'\in\mathcal{M}_{N}^{\partial}(m; s)} \hat{\varepsilon}_{m}\hat{\varepsilon}_{m'} x_{m}x_{m'}'  \displaybreak[0]\\
		&= \frac{1}{M} \sum_{s\geq 2}\sum_{m\in\mathcal{M}_{N}^{}} \sum_{m'\in\mathcal{M}_{N}^{\partial}(m; s)} \big\{\varepsilon_{m}+x_{m}'(\beta-\hat{\beta}) \big\} \big\{\varepsilon_{m'}+x_{m'}'(\beta-\hat{\beta}) \big\} x_{m}x_{m'}' \displaybreak[0]\\
		&= \frac{1}{M} \sum_{s\geq 2}\sum_{m\in\mathcal{M}_{N}^{}} \sum_{m'\in\mathcal{M}_{N}^{\partial}(m; s)}
			\varepsilon_{m}\varepsilon_{m'} x_{m}x_{m'}'
			+ \varepsilon_{m} (\beta-\hat{\beta}) x_{m'} x_{m}x_{m'}'
			+ x_{m}'(\beta-\hat{\beta})\varepsilon_{m'} x_{m}x_{m'}'
			+ x_{m}'(\beta-\hat{\beta})(\beta-\hat{\beta})' x_{m'} x_{m}x_{m'}' \displaybreak[0]\\
		&= \frac{1}{M} \sum_{s\geq 2}\sum_{m\in\mathcal{M}_{N}^{}} \sum_{m'\in\mathcal{M}_{N}^{\partial}(m; s)}
				\varepsilon_{m}\varepsilon_{m'} x_{m}x_{m'}'
			+ \frac{1}{M} \sum_{s\geq 2}\sum_{m\in\mathcal{M}_{N}^{}} \sum_{m'\in\mathcal{M}_{N}^{\partial}(m; s)}
				\varepsilon_{m} (\beta-\hat{\beta}) x_{m'} x_{m}x_{m'}' \displaybreak[0]\\
		&\quad + \frac{1}{M} \sum_{s\geq 2}\sum_{m\in\mathcal{M}_{N}^{}} \sum_{m'\in\mathcal{M}_{N}^{\partial}(m; s)}
				x_{m}'(\beta-\hat{\beta})\varepsilon_{m'} x_{m}x_{m'}'
			+ \frac{1}{M} \sum_{s\geq 2}\sum_{m\in\mathcal{M}_{N}^{}} \sum_{m'\in\mathcal{M}_{N}^{\partial}(m; s)}
				x_{m}'(\beta-\hat{\beta})(\beta-\hat{\beta})' x_{m'} x_{m}x_{m'}'.
	\end{align*}
	}
	Denote
	\begin{align*}
		& Q_{M,1} \coloneqq \frac{1}{M} \sum_{s\geq 2}\sum_{m\in\mathcal{M}_{N}^{}} \sum_{m'\in\mathcal{M}_{N}^{\partial}(m; s)}
							\varepsilon_{m}\varepsilon_{m'} x_{m}x_{m'}' \displaybreak[0]\\
		& Q_{M,2} \coloneqq \frac{1}{M} \sum_{s\geq 2}\sum_{m\in\mathcal{M}_{N}^{}} \sum_{m'\in\mathcal{M}_{N}^{\partial}(m; s)}
							\varepsilon_{m} (\beta-\hat{\beta}) x_{m'} x_{m}x_{m'}' \displaybreak[0]\\
		& Q_{M,3} \coloneqq \frac{1}{M} \sum_{s\geq 2}\sum_{m\in\mathcal{M}_{N}^{}} \sum_{m'\in\mathcal{M}_{N}^{\partial}(m; s)}
							x_{m}'(\beta-\hat{\beta})\varepsilon_{m'} x_{m}x_{m'}' \displaybreak[0]\\
		& Q_{M,4} \coloneqq \frac{1}{M} \sum_{s\geq 2}\sum_{m\in\mathcal{M}_{N}^{}} \sum_{m'\in\mathcal{M}_{N}^{\partial}(m; s)}
							x_{m}'(\beta-\hat{\beta})(\beta-\hat{\beta})' x_{m'} x_{m}x_{m'}'.
	\end{align*}
	
	From Theorem \ref{theo:ConsistencyofBetaHat}, it can be seen that $Q_{M,2}$, $Q_{M,3}$ and $Q_{M,4}$ either converge to zero or diverge as $N$ goes to infinity. When it comes to $Q_{M,1}$, observe that
	\begin{align*}
		E\Big[ Q_{M,1} \Big]
		&= E\Big[ \frac{1}{M} \sum_{s\geq 2}\sum_{m\in\mathcal{M}_{N}^{}} \sum_{m'\in\mathcal{M}_{N}^{\partial}(m; s)}
							\varepsilon_{m}\varepsilon_{m'} x_{m}x_{m'}' \Big] \displaybreak[0]\\
		&= \frac{1}{M} \sum_{s\geq 2}\sum_{m\in\mathcal{M}_{N}^{}} \sum_{m'\in\mathcal{M}_{N}^{\partial}(m; s)}
							E\Big[ \varepsilon_{m}\varepsilon_{m'} x_{m}x_{m'}' \Big],		
	\end{align*}
	which never equals to zero due to the hypothesis (\ref{eq:ExogenousRegression:NetworkHACEstimators:InconsistencyResult}) of this corollary. In either case, the middle part does not converge in probability to zero, meaning that $\big\| N \hat{V}_{N,M}^{Network} - N \hat{V}_{N,M}^{Dyad} \big\|_{F} \stackrel{p}{\rightarrow} 0$ is not true. This, however, contradicts the implication of the assumption that the dyadic-robust variance estimator is consistent. Hence, by means of contradiction, we conclude that the dyadic-robust variance estimator is not consistent, which completes the proof. \hfill$\square$

\subsection{Example \ref{example:MaximumAdmissibleBiasInTheDyadic-RobustVarianceEstimator}}\label{appndx:example_proof}

\textbf{Proof of Example \ref{example:MaximumAdmissibleBiasInTheDyadic-RobustVarianceEstimator}}


By the inequality of arithmetic and geometric means, the left-hand side of \eqref{eq:ExogenousRegression:NetworkHACEstimators:InconsistencyResult} can be bounded by
	\begin{align*}
		 \frac{1}{M} \sum_{s\geq 2} \sum_{m\in\mathcal{M}_{N}} \sum_{m'\in\mathcal{M}_{N}^{\partial}(m;s)} E\big[ \varepsilon_{M, m} \varepsilon_{M, m'} x_{M, m} x_{M, m'} \big]
		 &=  \sum_{s\geq 2} \gamma^{s} \delta_{M}^{\partial}(s) \\
		 &\geq (S-1) \bigg( \prod_{s\geq 2} \gamma^{s} \bigg)^{1/(S-1)} \bigg( \prod_{s\geq 2} \delta_{M}^{\partial}(s) \bigg)^{1/(S-1)},
	\end{align*}
	where $S\geq 2$ denotes the length of the longest path in the network. As the first and third terms in both estimators are the same, then using the proposed network-robust estimator will be desirable if the middle term (i.e., the left-hand side above) is larger than the tolerated threshold, $B$:
	\begin{align*}
		 (S-1) \bigg( \prod_{s\geq 2} \gamma^{s} \bigg)^{1/(S-1)} \bigg( \prod_{s\geq 2} \delta_{M}^{\partial}(s) \bigg)^{1/(S-1)}
		 > B.
	\end{align*}
	Passing logs on both sides yields the results. The lower bound is attained if $\gamma^{s} \delta_{M}^{\partial}(s)=\gamma^{s'} \delta_{M}^{\partial}(s')$ for all $s,s'=2,\ldots,S$.  Note that $S$ and the network densities $\{\delta_{M}^{\partial}(s)\}_{s\geq 2}$ can be estimated following the definition \eqref{eq:DefinitionsOfNetworkDensities}, as a (sample) network is observable.
	
	\bigskip
\section{Additional Monte Carlo Simulation Results}\label{appdx:AdditionalResults}


\subsection{Summary Statistics}\label{appdx:AdditionalResults:SummaryStatistics}

Table \ref{tbl:MonteCarloSimulation:DegreeCharacteristics:AverageMaxDegree_Nodes} shows summary statistics (i.e., the average and maximum degrees) of the networks across nodes that  are used in our simulation study. By construction, the maximum degree and the average degree monotonically increase in the parameters for both specifications. 

\begin{table}[H]
	\centering
	\caption{Summary Statistics of Networks among Nodes in the Simulations}
	\begin{threeparttable}
		\centering
		{\small
		\begin{tabular}{llccccccc} \toprule\toprule
		 		& 			& \multicolumn{3}{c}{Specification 1}	& &\multicolumn{3}{c}{Specification 2} \\ \cmidrule{3-5}\cmidrule{7-9} 
			$N$	&       		& $\nu=1$	& $\nu=2$ 	& $\nu=3$		& & $\lambda=1$	& $\lambda=2$	&	$\lambda=3$           \\ \midrule
			500 	& $d_{max}$	& 23 		& 40 		& 41 			& & 5			& 7 			& 8 \\ 
   				& $d_{ave}$	& 0.8020 	& 1.5780 	& 2.3540 		& & 0.4760		& 0.9680		& 1.4800 \\ \midrule
			1000	& $d_{max}$	& 26 		& 36 		& 47 			& & 4 			& 7 			& 8 \\ 
   				& $d_{ave}$	& 0.8590 	& 1.7000 	& 2.5410 		& & 0.4980 		& 0.9810 		& 1.5010 \\ \midrule
			5000 & $d_{max}$	& 53 		& 125 	& 130		& & 6 			& 9 			& 10 \\ 
   				& $d_{ave}$	& 0.9326 	& 1.8618 	& 2.7910 		& & 0.4952 		& 1.0016 		& 1.5114 \\ \bottomrule
			\end{tabular}
		}
		\begin{tablenotes}[para,flushleft]
			\item[] {\footnotesize {\sl Notes}: Observation units in this table are nodes (individuals) as usual in the literature. The maximum degree, $d_{max}$, means the maximum number of nodes that are adjacent to a node, and the average degree, $d_{ave}$, is the average number of nodes adjacent to each node of the network.}
		\end{tablenotes}
	\end{threeparttable}
	\label{tbl:MonteCarloSimulation:DegreeCharacteristics:AverageMaxDegree_Nodes}
\end{table}

Table \ref{tbl:MonteCarloSimulation:DegreeCharacteristics:AverageMaxDegree_Edges} reports the degree characteristics of the networks when viewed as networks over the active edges. The table provides the average degree, the maximum degree, and the number of active edges (i.e., dyads). 

\begin{table}[H]
	\centering
	\caption{Summary Statistics of Networks among Dyads in the Simulations}
	\begin{threeparttable}
		\centering
		{\small
		\begin{tabular}{llccccccc} \toprule \toprule
			$$ 	& 			& \multicolumn{3}{c}{Specification 1} & & \multicolumn{3}{c}{Specification 2} \\ \cmidrule{3-5}\cmidrule{7-9} 
			$N$ 	& 			& $\nu=1$	& $\nu=2$	& $\nu=3$ 	& & $\lambda=1$ & $\lambda=2$	& $\lambda=3$ \\ \midrule 
			500	& $d_{act}$	& 401	& 788 	& 1175 		& & 238 		& 484 	& 740 \\ 
   				& $d_{max}$	& 32 		& 45 		& 70 			& & 4 		& 9 		& 14 \\ 
   				& $d_{ave}$  	& 3.6858	& 6.0063 	& 9.1881 		& & 0.9580 	& 2.0248 	& 3.0027 \\ 
			\midrule 
			1000	& $d_{act}$	& 859	& 1699 	& 2540 		& & 498		& 981 	& 1501 \\ 
   				& $d_{max}$	& 35 		& 55 		& 76 			& & 5 		& 8 		& 10 \\
   				& $d_{ave}$	& 3.9581 	& 6.7810 	& 9.2047 		& & 1.0341 	& 1.9888 	& 2.9594 \\ 
				\midrule
			5000	& $d_{act}$	& 4663 	& 9305 	& 13952 		& & 2476 		& 5008 	& 7557 \\ 
   				& $d_{max}$	& 74 		& 161 	& 210 		& & 7 		& 12 		& 15 \\ 
   				& $d_{ave}$	& 5.2989 	& 9.5159 	& 12.8521 	& & 1.0137 	& 2.0228 	& 3.0341 \\ \bottomrule
			\end{tabular}
		}
		\begin{tablenotes}[para,flushleft]
			\item[] {\footnotesize {\sl Note}: Observation units in this table are active edges (dyads), which departs from the convention. Active edges are edges that are at work in the original network over the nodes. The number of active edges is denoted by $d_{act}$. The maximum degree, $d_{max}$, expresses the maximum number of edges that are adjacent to an edge, and the average degree, $d_{ave}$, is the average number of edges adjacent to each edge of the network. }
		\end{tablenotes}
	\end{threeparttable}
	\label{tbl:MonteCarloSimulation:DegreeCharacteristics:AverageMaxDegree_Edges}
\end{table}

\subsection{$S=2$ and $\gamma=0.8$}\label{appdx:AdditionalResults:S2gamma0.8}

In this section, we further discuss the results of the Monte Carlo simulations presented in the main text. The asymptotic behaviors of the three variance estimators are illustrated in Figure \ref{fig:MonteCarloSimulation:Results:Boxplots08}, where the horizontal axes represent the sample size and the vertical axes indicate the standard error of the regression coefficient. The boxplots show the 25th and 75th percentiles across simulations, as well as the median, with the whiskers indicating the bounds that are not considered as outliers. The whisker length is set to cover $\pm 2.7$ times the standard deviation of the standard-error estimates. The light-, medium- and dark-gray boxplots describe the distribution of the Eicker-Huber-White, the dyadic-robust and our proposed network-robust variance estimates across simulations, respectively. The diamonds indicate the empirical standard errors of the estimates of the regression coefficients, what \citet{Aronow_et_al-2015} call the true standard error. It is unsurprising that the empirical standard errors are the same across different variance estimators, as we use the same $\hat{\beta}$. The boxplots show that as the sample size increases, the variation of the network-robust variance estimator shrinks, reaching the empirical standard error (the diamonds). This is as expected since this estimator is consistent for the true variance (Theorem \ref{thm:ExogenousRegression:ConsistencyOfTheNetwork-HACVarianceEstimator}). The estimates appear to vary little for moderate sample sizes (e.g., $N=1000$). However, the other variance estimators (the light- and medium-gray boxplots) converge to lower values than the empirical standard errors (the diamonds), verifying their inconsistency in this environment with network spillovers, as shown by Corollary \ref{coro:ExogenousRegression:NetworkHACEstimators:InconsistencyResult}. As we make such spillovers very small (e.g., $\gamma = 0.2$ in Appendix \ref{appdx:AdditionalResults:S2gamma0.2}), all estimators have similar performance. This highlights the role of condition (\ref{eq:ExogenousRegression:NetworkHACEstimators:InconsistencyResult}): namely, the dyadic-robust variance estimator might perform satisfactorily well as long as higher-order correlations beyond immediate neighbors are negligible.

\setcounter{figure}{1}
\begin{figure}[htbp]
	\centering
	\caption{Boxplots of Standard Errors for Specifications 1 and 2 ($S=2$, $\gamma=0.8$)}
	\begin{threeparttable}
		\begin{tabular}{p{\textwidth}}
			\centering
			\includegraphics[keepaspectratio, scale=.70]{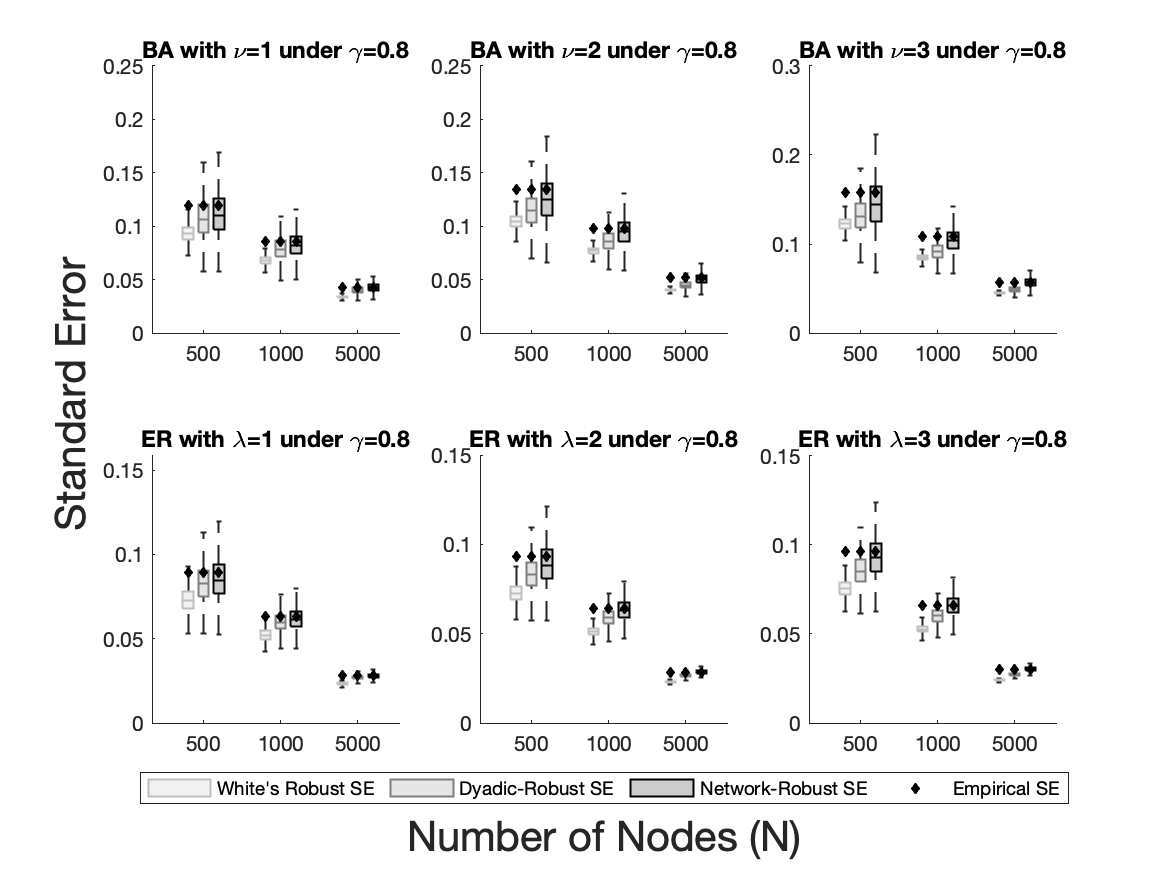}
		\end{tabular}

		\begin{tablenotes}[para,flushleft]
			\item[] {\footnotesize {\sl Note}: This figure shows boxplots describing the estimated standard errors and the empirical standard errors for various combinations of parameters under Specification 1 (Barab\'asi-Albert networks) and Specification 2 (Erd\"os-Renyi networks). The horizontal axis shows the number of nodes and the vertical axis represents the the standard error of the coefficient. The shaded boxes represent the 25th, 50th and 75th percentiles of estimated standard errors with the whiskers indicating the most extreme values that are not considered as outliers. The light-gray box illustrates the Eicker-Huber-White standard error, the medium-gray one the dyadic-robust standard error and the dark-gray one the network-robust standard error. The diamonds stand for the empirical standard error, defined as the standard deviation of the estimates of the regression coefficient. The estimator is considered as not covering the true standard error when the diamond is outside of the shaded area.}
		\end{tablenotes}
	\end{threeparttable}
	\label{fig:MonteCarloSimulation:Results:Boxplots08}
\end{figure}

Table \ref{tbl:MonteCarloSimulation:Results:MeanAndBiasOfTheStandardErrors:N=5000Gamma=0.8} describes the standard deviations of the estimated regression coefficients (what \citet{Aronow_et_al-2015} calls the true standard errors) and the means of the estimated standard errors for each variance estimator. The round brackets indicate the biases of each estimate relative to the true standard error in percentage ($\%$). For instance, the Eicker-Huber-White variance estimator and the dyadic-robust variance estimator, when applied to Specification 1 with $\nu=3$, underestimate the true standard error by 21.45\% and 14.14\%, respectively.


\begin{table}[htbp]
	\centering
	\caption{Means and Biases of the Standard Errors: $N=5000$, $S=2$, $\gamma=0.8$.} 
	\begin{threeparttable}
		{\small
		\begin{tabular}{lccccccc}\toprule \toprule 
 								& \multicolumn{3}{c}{Specification 1}  & & \multicolumn{3}{c}{Specification 2}  \\ \cmidrule{2-4} \cmidrule{6-8} 
 								& $\nu=1$	& $\nu=2$	& $\nu=3$	& & $\lambda=1$	& $\lambda=2$		& $\lambda=3$  \\ \midrule 
			True 					& 0.0430 	& 0.0518	& 0.0570	& & 0.0285		& 0.0283	&   0.0302     \\ \\ 
			
			Eicker-Huber-White 		& 0.0337	& 0.0404	& 0.0448 	& & 0.0234  		& 0.0229	& 0.0239     \\ 
			\qquad (Bias \%) 		& (-21.61) & (-21.92) & (-21.45)	& & (-17.91)   		& (-19.14)	& (-20.68)    \\ \\ 
			
			Dyadic-robust 			& 0.0403	& 0.0453 	& 0.0490 	& & 0.0270  		& 0.0266	& 0.0275     \\ 
			\qquad (Bias \%) 		& (-6.19)	& (-12.54) & (-14.14)	& & (-5.01)   		& (-6.12)	& (-8.93)    \\ \\ 

			Network-robust 		& 0.0425	& 0.0509 	& 0.0565 	& & 0.0280  		& 0.0285	& 0.0302     \\ 
			\qquad (Bias \%) 		& (-1.09)	& (-1.78) 	& (-0.92) 	& & (-1.70)   		& (0.58)	& (0.09)    \\ \bottomrule 
		\end{tabular}
		
		\begin{tablenotes}[para,flushleft]
			\item[] {\sl Note}: This table shows the standard deviations of the estimated regression coefficients (the true standard error) and the means of the estimated standard errors for each variance estimator with the round brackets indicating the biases relative to the true standard error in percentage ($\%$). To facilitate the comparison, the biases are rounded off to the second decimal places.
		\end{tablenotes}
		}
	\end{threeparttable}
	\label{tbl:MonteCarloSimulation:Results:MeanAndBiasOfTheStandardErrors:N=5000Gamma=0.8}
\end{table}


\subsection{$S=2$ and $\gamma=0.8$ with Higher Density Parameters}\label{appdx:AdditionalResults:S2gamma0.8WithHigherDensityParameters}


This subsection examines how an increase in the number of connected dyads affects the performance of the dyadic-robust variance estimator. Table \ref{tbl:MonteCarloSimulation:Results:CPACI_S2gamma0.8woKernelwHigherDensenessParameters} reports the results for the case of $S=2$ with $\gamma=0.8$, i.e., the same combination as the main text (Table \ref{tbl:MonteCarloSimulation:Results:CPACI_S2gamma0.8woKernel}), but for denser networks which set $\nu=4,5$ for Specification 1 and $\lambda=4,5$ for Specification 2. We find that, while our estimator performs well (with coverage close to the nominal level), the bias in the Eicker-Huber-White and dyadic-robust estimators variance estimators are present and increase as the network becomes denser. 

\begin{table}[htbp]
	\centering
	\caption{The empirical coverage probability and average length of confidence intervals for $\beta$ at $95\%$ nominal level: $S=2$, $\gamma=0.8$, higher denseness parameters.}
	\begin{threeparttable}
		\begin{tabular}{lcccccc} \toprule \toprule
							& $$ 		& \multicolumn{2}{c}{Specification 1} && \multicolumn{2}{c}{Specification 2} \\ \cmidrule{3-4}\cmidrule{6-7}
							& $N$ 	& $\nu=4$ & $\nu=5$ & & $\lambda=4$ & $\lambda=5$    \\ \midrule 
							&		& \multicolumn{5}{c}{Coverage Probability} \\ 
			Eicker-Huber-White	& 500 	& 0.8758	& 0.8688	& & 0.8744	& 0.8706     \\ 
  							& 1000 	& 0.8752 	& 0.8688 	& & 0.8694	& 0.8784     \\ 
  							& 5000 	& 0.8658 	& 0.8808 	& & 0.8694	& 0.8750     \\
			Dyadic-robust		& 500 	& 0.8912 	& 0.8808 	& & 0.9142	& 0.9058     \\ 
  							& 1000 	& 0.8936 	& 0.8852 	& & 0.9124	& 0.9160     \\ 
  							& 5000 	& 0.8940 	& 0.9020 	& & 0.9152	& 0.9176     \\
			Network-robust		& 500 	& 0.9084 	& 0.8928 	& & 0.9394	& 0.9368     \\ 
  							& 1000 	& 0.9246 	& 0.9190 	& & 0.9424	& 0.9494     \\ 
 							& 5000 	& 0.9404 	& 0.9436 	& & 0.9450	& 0.9512     \\ 
		\\					
							&		& \multicolumn{5}{c}{Average Length of the C.I.} \\ 
			Eicker-Huber-White	& 500 	& 0.5282	& 0.5577 	& & 0.3088	& 0.3230     \\ 
  							& 1000 	& 0.3964 	& 0.4155 	& & 0.2183	& 0.2323     \\ 
  							& 5000 	& 0.1944 	& 0.2132 	& & 0.0992	& 0.1045     \\
			Dyadic-robust 		& 500 	& 0.5580 	& 0.5841 	& & 0.3471	& 0.3601     \\ 
							& 1000 	& 0.4211 	& 0.4380 	& & 0.2465	& 0.2595     \\ 
  							& 5000 	& 0.2085 	& 0.2254 	& & 0.1124	& 0.1172     \\
			Network-robust		& 500 	& 0.6099 	& 0.6428 	& & 0.3825	& 0.4022     \\ 
  							& 1000 	& 0.4751 	& 0.4966 	& & 0.2743	& 0.2927     \\ 
  							& 5000 	& 0.2449 	& 0.2660 	& & 0.1259	& 0.1331     \\ \bottomrule
		\end{tabular}
		\begin{tablenotes}[para,flushleft]
			\item[] {\footnotesize {\sl Note}: The upper-half of the table displays the empirical coverage probability of the asymptotic confidence interval for $\beta_{}$, and the lower-half showcases the average length of the estimated confidence intervals. As the sample size ($N$) increases, the empirical coverage probability approaches $0.95$, the nominal level. This convergence is accompanied by the shrinking average length of confidence intervals.}
		\end{tablenotes}
	\end{threeparttable}
	\label{tbl:MonteCarloSimulation:Results:CPACI_S2gamma0.8woKernelwHigherDensenessParameters}
\end{table}


\subsection{$S=2$ and $\gamma=0.2$}\label{appdx:AdditionalResults:S2gamma0.2}


Table \ref{tbl:MonteCarloSimulation:Results:CPACI_S2gamma0.2woKernel} presents the empirical coverage probability and average length of confidence intervals for $\beta$ at $5\%$ nominal size when $S=2$ and $\gamma=0.2$. The associated boxplots are given in Figure \ref{fig:MonteCarloSimulation:Results:Boxplots_S2gamma0.2}. Since the magnitude of spillovers is now much smaller than the case of $\gamma=0.8$, there are only minor differences in performance between the network-robust variance estimator and the other two existing methods (namely, the Eicker-Huber-White and dyadic-robust variance estimators). In terms of convergence, the comparable performance of the dyadic-robust-variance estimator is evident in Figure \ref{fig:MonteCarloSimulation:Results:Boxplots_S2gamma0.2}. 

Comparing Table \ref{tbl:MonteCarloSimulation:Results:CPACI_S2gamma0.2woKernel} to Table \ref{tbl:MonteCarloSimulation:Results:CPACI_S2gamma0.8woKernel} highlights the impact of spillovers on the variance estimators. When the spillovers are substantially weak (e.g., $\gamma=0.2$), the dyadic-robust variance estimator can serve as a good substitute for the network-robust one. In the case of relatively high spillovers (e.g., $\gamma=0.8$), on the other hand, there are evident biases (around 4 percentage points for Specification 1 and 3 percentage points for Specification 2 when $N=5000$). Based on this comparison, we suggest that the network-robust variance estimator be used when the correlations are expected to be relatively strong.

\begin{table}[htbp]
	\centering
	\caption{The empirical coverage probability and average length of confidence intervals for $\beta$ at $95\%$ nominal level: $S=2$, $\gamma=0.2$.}
	\begin{threeparttable}
		\begin{tabular}{lcccccccc} \toprule \toprule
						& $$ 		& \multicolumn{3}{c}{Specification 1} && \multicolumn{3}{c}{Specification 2} \\ \cmidrule{3-5}\cmidrule{7-9}
								& $N$	& $\nu=1$	& $\nu=2$	& $\nu=3$	&& $\lambda=1$	& $\lambda=2$	& $\lambda=3$ \\ \midrule
								&		& \multicolumn{6}{c}{Coverage Probability} \\ 
		Eicker-Huber-White	& 500	& 0.9286 	& 0.9186 	& 0.9108 	& & 0.9434	& 0.9292	& 0.9308     \\ 
  						& 1000	& 0.9320 	& 0.9158 	& 0.9148 	& & 0.9338	& 0.9322	& 0.9322     \\ 
 						& 5000	& 0.9200 	& 0.9110 	& 0.9124 	& & 0.9434	& 0.9382	& 0.9308     \\
		Dyadic-robust 		& 500	& 0.9342 	& 0.9350 	& 0.9336 	& & 0.9454	& 0.9368	& 0.9422     \\ 
  						& 1000	& 0.9454 	& 0.9376 	& 0.9458 	& & 0.9398	& 0.9422	& 0.9486     \\ 
  						& 5000	& 0.9446 	& 0.9448 	& 0.9432 	& & 0.9486	& 0.9490	& 0.9472     \\
		Network-robust 	& 500	& 0.9284 	& 0.9246 	& 0.9162 	& & 0.9428	& 0.9360	& 0.9384     \\ 
  						& 1000	& 0.9414 	& 0.9294 	& 0.9370 	& & 0.9392	& 0.9410	& 0.9456     \\ 
  						& 5000	& 0.9454 	& 0.9476 	& 0.9418 	& & 0.9494	& 0.9492	& 0.9470     \\
  					\\
								&		& \multicolumn{6}{c}{Average Length of the Confidence Intervals} \\ 
		Eicker-Huber-White	& 500	& 0.1578	& 0.1214 	& 0.1092 	& & 0.1860	& 0.1360	& 0.1141     \\ 
  						& 1000	& 0.1088 	& 0.0846 	& 0.0743 	& & 0.1290	& 0.0955	& 0.0799     \\ 
  						& 5000	& 0.0486 	& 0.0388 	& 0.0346 	& & 0.0579	& 0.0423	& 0.0357     \\
		Dyadic-robust 		& 500	& 0.1648 	& 0.1316 	& 0.1213 	& & 0.1890	& 0.1410	& 0.1205     \\ 
  						& 1000	& 0.1158 	& 0.0931 	& 0.0833 	& & 0.1319	& 0.0994	& 0.0848     \\ 
  						& 5000	& 0.0532 	& 0.0439 	& 0.0398 	& & 0.0594	& 0.0443	& 0.0381     \\
		Network-robust 	& 500	& 0.1637 	& 0.1291 	& 0.1174 	& & 0.1885	& 0.1404	& 0.1196     \\ 
  						& 1000	& 0.1154 	& 0.0922 	& 0.0825 	& & 0.1318	& 0.0993	& 0.0848     \\ 
  						& 5000	& 0.0533 	& 0.0440	& 0.0401	& & 0.0594	& 0.0444	& 0.0382     \\ \bottomrule
		\end{tabular}
		\begin{tablenotes}[para,flushleft]
			\item[] {\footnotesize {\sl Note}: The upper-half of the table displays the empirical coverage probability of the asymptotic confidence interval for $\beta_{}$, and the lower-half showcases the average length of the estimated confidence intervals. One computational issue that plagues the Monte Carlo simulation is the potential lack of positive-semi-definiteness of the estimated variance-covariance matrix. In general, this problem prevails only when the sample size ($N$) is small. In our case, when $N=500$, four variance estimates out of five thousands take negative values. We deal with this issue by first applying the eigenvalue decomposition to the estimated variance-covariance matrix and then augmenting the diagonal matrix of eigenvalues by a small constant, followed by pre- and post-multiplications by the matrix of eigenvectors to obtain the updated estimate for the variance-covariance matrix. As the sample size ($N$) increases, the empirical coverage probability approaches $0.95$, the nominal level. This convergence is accompanied by the shrinking average length of confidence intervals.}
		\end{tablenotes}
	\end{threeparttable}
	\label{tbl:MonteCarloSimulation:Results:CPACI_S2gamma0.2woKernel}
\end{table}


\begin{figure}[htbp]
	\centering
	\caption{Boxplots of Standard Errors for Specifications 1 and 2 ($S=2$, $\gamma=0.2$)}
	\begin{threeparttable}
		\begin{tabular}{p{\textwidth}}
			\centering
			\includegraphics[keepaspectratio, scale=.70]{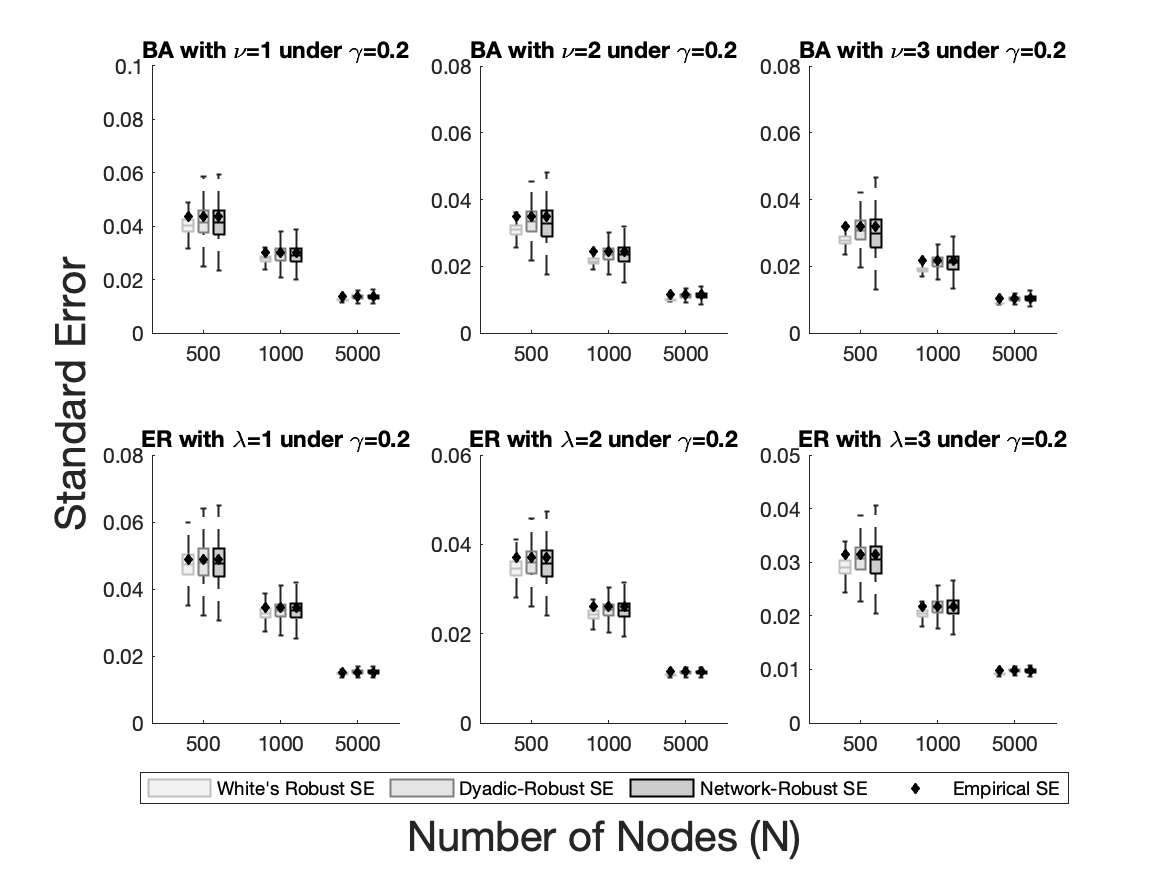}
		\end{tabular}

		\begin{tablenotes}[para,flushleft]
			\item[] 
			{\footnotesize {\sl Note}: This figure shows boxplots describing the estimated standard errors and the empirical standard errors for various combinations of parameters under Specification 1 (Barab\'asi-Albert networks) and Specification 2 (Erd\"os-Renyi networks). The horizontal axis shows the number of nodes and the vertical axis represents the the standard error of the coefficient. The shaded boxes represent the 25th, 50th and 75th percentiles of estimated standard errors with the whiskers indicating the most extreme values that are not considered as outliers. The light-gray box illustrates the Eicker-Huber-White standard error, the medium-gray one the dyadic-robust standard error and the dark-gray one the network-robust standard error. The diamonds stand for the empirical standard error, defined as the standard deviation of the estimates of the regression coefficient. This figure showcases the boxplots for the case when $\gamma=0.2$.}
		\end{tablenotes}
	\end{threeparttable}
	\label{fig:MonteCarloSimulation:Results:Boxplots_S2gamma0.2}
\end{figure}


\subsection{$S=1$}\label{appdx:AdditionalResults:S1}

For comparison purposes, this subsection explores the results for $S=1$. If $S=1$, there are no higher-order correlations beyond direct (adjacent) neighbors. Then, the network-robust variance estimator ought to coincide with the dyadic-robust variance estimator by definition, for any $\gamma$, as pointed out in Example \ref{ex:ExogenousRegression:Dyadic-RobustVarianceEstimator}. This is verified below for the case of $\gamma = 0.8$. Table \ref{tbl:MonteCarloSimulation:Results:CPACI_S1gamma0.8woKernel} shows the simulation results.


\begin{table}[htbp]
	\centering
	\caption{The empirical coverage probability and average length of confidence intervals for $\beta$ at $95\%$ nominal level: $S=1$, $\gamma=0.8$.}             
	\begin{threeparttable}
		\begin{tabular}{lcccccccc} \toprule\toprule
							& $$ 		& \multicolumn{3}{c}{Specification 1} && \multicolumn{3}{c}{Specification 2} \\ \cmidrule{3-5}\cmidrule{7-9}
							& $N$ 	& $\nu=1$	& $\nu=2$	& $\nu=3$	& & $\lambda=1$	& $\lambda=2$	& $\lambda=3$ \\ \midrule
							&		& \multicolumn{6}{c}{Coverage Probability}						\\ 
			Eicker-Huber-White 	& 500 	& 0.8804 	& 0.8676 	& 0.8734 	& & 0.8906	& 0.8768	& 0.8692     	\\ 
  							& 1000	& 0.8678 	& 0.8810 	& 0.8710 	& & 0.8984	& 0.8864	& 0.8856		\\ 
  							& 5000 	& 0.8752 	& 0.8652 	& 0.8742 	& & 0.8996	& 0.8910	& 0.8778     	\\
			Dyadic-robust 		& 500 	& 0.9292 	& 0.9304 	& 0.9384 	& & 0.9366	& 0.9416	& 0.9368     	\\ 
  							& 1000 	& 0.9364 	& 0.9426 	& 0.9432 	& & 0.9428	& 0.9454	& 0.9484     	\\ 
  							& 5000 	& 0.9474 	& 0.9414 	& 0.9498 	& & 0.9452	& 0.9518	& 0.9506     	\\
			Network-robust 	& 500 	& 0.9292 	& 0.9304 	& 0.9384 	& & 0.9366	& 0.9416	& 0.9368     	\\ 
  							& 1000 	& 0.9364 	& 0.9426 	& 0.9432 	& & 0.9428	& 0.9454	& 0.9484     	\\ 
  							& 5000 	& 0.9474	& 0.9414 	& 0.9498 	& & 0.9452	& 0.9518	& 0.9506     	\\
  			\\
							&		& \multicolumn{6}{c}{Average Length of the Confidence Intervals} 		\\ 
			Eicker-Huber-White 	& 500 	& 0.3282 	& 0.2901 	& 0.2881 	& & 0.2664	& 0.2377	& 0.2235     	\\ 
  							& 1000 	& 0.2321 	& 0.2088 	& 0.1964 	& & 0.1887	& 0.1665	& 0.1564     	\\ 
  							& 5000 	& 0.1131 	& 0.1042 	& 0.0980 	& & 0.0844	& 0.0742	& 0.0704     	\\
			Dyadic-robust 		& 500 	& 0.3934 	& 0.3591 	& 0.3603 	& & 0.3104	& 0.2888	& 0.2776     	\\ 
  							& 1000 	& 0.2853 	& 0.2625 	& 0.2500 	& & 0.2227	& 0.2037	& 0.1950     	\\ 
  							& 5000 	& 0.1428 	& 0.1330 	& 0.1259 	& & 0.0998	& 0.0913	& 0.0882     	\\
			Network-robust 	& 500 	& 0.3934 	& 0.3591 	& 0.3603 	& & 0.3104	& 0.2888	& 0.2776     	\\ 
  							& 1000 	& 0.2853 	& 0.2625 	& 0.2500 	& & 0.2227	& 0.2037	& 0.1950     	\\ 
  							& 5000 	& 0.1428	& 0.1330	& 0.1259	& & 0.0998	& 0.0913	& 0.0882     	\\ \bottomrule
		\end{tabular}
	
		\begin{tablenotes}[para,flushleft]
			\item[] {\footnotesize {\sl Note}: The upper-half of the table displays the empirical coverage probability of the asymptotic confidence interval for $\beta_{}$, and the lower-half showcases the average length of the estimated confidence intervals. As the sample size ($N$) increases, the empirical coverage probability approaches $0.95$, the nominal level. This convergence is accompanied by the shrinking average length of confidence intervals.}
		\end{tablenotes}
	\end{threeparttable}
	\label{tbl:MonteCarloSimulation:Results:CPACI_S1gamma0.8woKernel}
\end{table}

\subsection{$S=\infty$ with the Parzen kernel}\label{appdx:AdditionalResults:SInfty}

In this subsection, we investigate the consequences of adaptively choosing the value of the lag-truncation parameter following the rule outlined in the main text. To this end, we set $S=\infty$ (i.e., spillovers may propagate to all neighbors), with the magnitude of the spillovers controlled by $\gamma = 0.8$ (the same as in the main text). In this environment, the spillovers are never truncated while decaying as they propagate farther. With regards to estimation, we consider the Parzen kernel, letting the lag-truncation parameter be chosen on the basis of \citet{Kojevnikov_et_al-2021}. The simulation results are given in Table \ref{tbl:MonteCarloSimulation:Results:CPACI_SInftygamma0.8wKernel}, while the selected lag-truncation parameters are shown in Table \ref{tbl:MonteCarloSimulation:Results:Lag-TruncationParameters}. 

The empirical coverage probability based on the network-robust variance estimator approaches to 95\%, as expected. On the other hand, both the Eicker-Huber-White and dyadic-robust variance estimator understate the targeted nominal level, as claimed in the main text. It should be noted that these biases can become larger when the decay rate is slower. We focus on Specification 2, as it likely satisfies the assumptions above under $S=\infty$. After all, with $S=\infty$ and a very dense network, Assumption 3.4 is violated. 


\begin{table}[htbp]
	\centering
	\caption{The empirical coverage probability and average length of confidence intervals for $\beta$ at $95\%$ nominal level, Specification 2: $S=\infty$, $\gamma=0.8$, the Parzen kernel.}             
	\begin{threeparttable}
		\begin{tabular}{lcccc} \toprule\toprule
							& $N$ & $\lambda=1$	& $\lambda=2$	& $\lambda=3$ \\ \midrule
							&		& \multicolumn{3}{c}{Coverage Probability} \\ 
			Eicker-Huber-White	& 500  & 0.8884	& 0.8768	& 0.8718     \\ 
  							& 1000 	 & 0.8892	& 0.8832	& 0.8806     \\ 
  							& 5000 	 & 0.8966	& 0.8820	& 0.8806     \\
			Dyadic-robust 		& 500  & 0.9282	& 0.9150	& 0.8964     \\ 
  							& 1000	 & 0.9300	& 0.9214	& 0.9126     \\ 
  							& 5000 & 0.9384	& 0.9272	& 0.9118     \\
			Network-robust 	& 500 	 & 0.9366	& 0.9302	& 0.9180     \\ 
  							& 1000  & 0.9382	& 0.9386	& 0.9362     \\ 
  							& 5000  & 0.9480	& 0.9510	& 0.9462     \\
  			\\
							&		& \multicolumn{3}{c}{Average Length of C.I.} \\ 
			Eicker-Huber-White	& 500  & 0.2890	& 0.3085	& 0.3658     \\ 
  							& 1000 & 0.2103	& 0.2241	& 0.2570     \\ 
  							& 5000 	& 0.0933	& 0.0994	& 0.1188     \\
			Dyadic-robust & 500		 & 0.3293	& 0.3483	& 0.3966     \\ 
  							& 1000  & 0.2405	& 0.2526	& 0.2809     \\ 
  							& 5000 	& 0.1075	& 0.1127	& 0.1300     \\
			Network-robust		& 500 	& 0.3387	& 0.3705	& 0.4233     \\ 
  							& 1000 	 & 0.2502	& 0.2729	& 0.3070     \\ 
  							& 5000 	 & 0.1120	& 0.1233	& 0.1457     \\ \bottomrule
		\end{tabular}
	
		\begin{tablenotes}[para,flushleft]
			\item[] {\footnotesize {\sl Note}: The upper-half of the table displays the empirical coverage probability of the asymptotic confidence interval for $\beta_{}$, and the lower-half showcases the average length of the estimated confidence intervals. As the sample size ($N$) increases, the empirical coverage probability approaches $0.95$, the nominal level. This convergence is accompanied by the shrinking average length of confidence intervals.}
		\end{tablenotes}
	\end{threeparttable}
	\label{tbl:MonteCarloSimulation:Results:CPACI_SInftygamma0.8wKernel}
\end{table}


\begin{table}[htbp]
	\centering
	\caption{The lag-truncation parameters for Table \ref{tbl:MonteCarloSimulation:Results:CPACI_SInftygamma0.8wKernel} based on the \citepos{Kojevnikov_et_al-2021} rule.}             
	\begin{threeparttable}
		\begin{tabular}{lccc}\toprule \toprule 
  			$N$	  & $\lambda=1$	& $\lambda=2$	& $\lambda=3$  \\ \midrule
  			500	 & 224.3186		& 17.5262		& 12.0174  \\
  			1000	& 254.5841		& 20.0388		& 13.4822  \\ 
  			5000	&  320.3268		& 24.1851		& 16.0915  \\ \bottomrule 
		\end{tabular} 
		\end{threeparttable}
					\parbox{6.2in}{\footnotesize Note: This table displays the lag-truncation parameters $b_M$ for the simulations in Table \ref{tbl:MonteCarloSimulation:Results:CPACI_SInftygamma0.8wKernel}, selected using the rule: $b_M = 2 \log(M)/\log(\max(average~degree, 1.05))$, with $M$ denoting the number of active dyads.}
			\label{tbl:MonteCarloSimulation:Results:Lag-TruncationParameters}
\end{table}

\newpage

\section{Additional Information for the Empirical Illustration}\label{appndx:EmpiricalIllustration}


\subsection{Seating Arrangement at the European Parliament}\label{appndx:EmpiricalIllustration:SeatingArrangement}

Figure \ref{fig:SeatingPlanAtTheEuropeanParliament} exhibits an example of the seating arrangement at the European Parliament, and describes how we construct an adjacency relationship among MEPs within their EPG groups.

\begin{figure}[htbp]
	\caption{Seating Plan at the European Parliament: Strasbourg, September 14, 2009}
	\begin{threeparttable}
	\begin{tabular}{p{\textwidth}}
	\centering
	\includegraphics[keepaspectratio, scale=0.50]{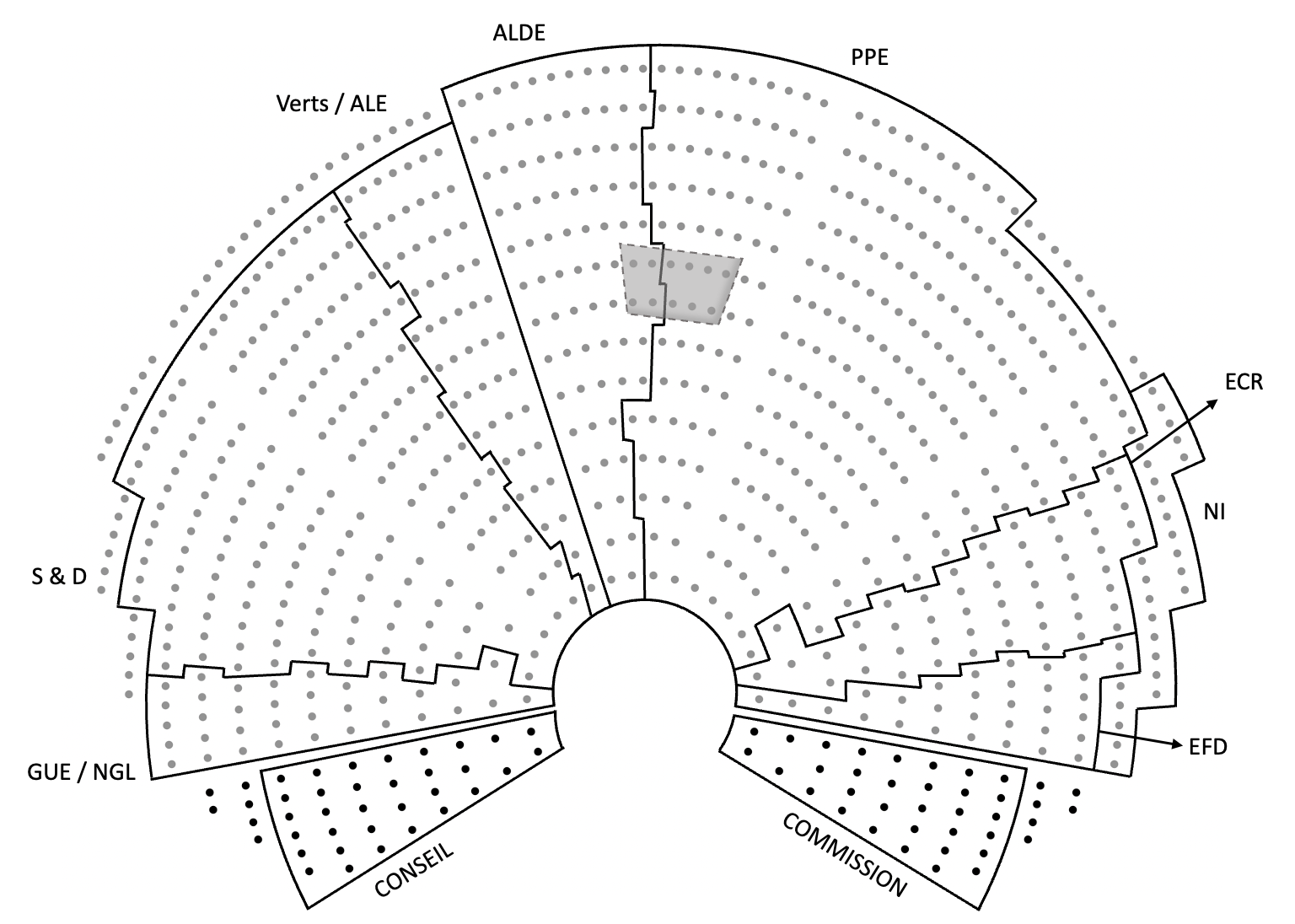}
	\centering
	\vspace{5mm}
	\scalebox{0.85}{
      	\begin{tabular}{lcccccccccccccc} \toprule \toprule
			& \multicolumn{5}{c}{ALDE} & \multicolumn{8}{c}{PPE} \\
		Row 10	& $\cdots$ & \multicolumn{2}{c}{$\vdots$} & & \multicolumn{1}{|c}{$\vdots$} & \multicolumn{2}{c}{$\vdots$} & \multicolumn{2}{c}{$\vdots$} & \multicolumn{2}{c}{$\vdots$} & \multicolumn{2}{c}{$\vdots$} & $\cdots$ \\ \cline{6-6}
			&  & & & & & \multicolumn{1}{|c}{\phantom{Seat}} & & & & & && & \\ 
		\multirow{2}{*}{Row 9}	& \multirow{2}{*}{$\cdots$} & \multicolumn{2}{c}{Seat 36:} & \multicolumn{2}{c}{Seat 37:} & \multicolumn{2}{|c}{Seat 38:} & \multicolumn{2}{c}{Seat 39:} & \multicolumn{2}{c}{Seat 40:} & \multicolumn{2}{c}{Seat 41:} & \multirow{2}{*}{$\cdots$} \\ 
			&  & \multicolumn{2}{c}{ALFANO} & \multicolumn{2}{c}{ALVARO} & \multicolumn{2}{|c}{CAVADA} & \multicolumn{2}{c}{COELHO} & \multicolumn{2}{c}{COLLINO} & \multicolumn{2}{c}{COMI} &   \\ 	
			&  & & & & & \multicolumn{1}{|c}{\phantom{Seat}} & & & & & && & \\ \cline{7-7}
			& 			&						& \multicolumn{2}{c}{} & \multicolumn{2}{c}{} & \multicolumn{2}{|c}{} & \multicolumn{2}{c}{} & \multicolumn{2}{c}{} & \\	
		\multirow{2}{*}{Row 8}	& \multicolumn{2}{c}{\multirow{2}{*}{$\cdots$}} & \multicolumn{2}{c}{Seat 32:} & \multicolumn{2}{c}{Seat 33:} & \multicolumn{2}{|c}{Seat 34:} & \multicolumn{2}{c}{Seat 35:} & \multicolumn{2}{c}{Seat 36:} & \multirow{2}{*}{$\cdots$} \\ 
						& 			&						& \multicolumn{2}{c}{OJULAND} & \multicolumn{2}{c}{OVIIR} & \multicolumn{2}{|c}{BACH} & \multicolumn{2}{c}{BALDASSARRE} & \multicolumn{2}{c}{BALZ} & \\
				&  & & & & & & \multicolumn{1}{|c}{\phantom{Seat}} & & & & && & \\ \cline{7-7}
		Row 7	& $\cdots$ & \multicolumn{2}{c}{$\vdots$} & \multicolumn{2}{c}{$\vdots$} & \multicolumn{2}{|c}{$\vdots$} & \multicolumn{2}{c}{$\vdots$} & \multicolumn{2}{c}{$\vdots$} & \multicolumn{2}{c}{$\vdots$}  & $\cdots$\\ 
		\bottomrule
	\end{tabular}
	}
	\end{tabular}
		\begin{tablenotes}[para,flushleft]
			\item[] {\footnotesize {\sl Note}: The upper panel illustrates a zoomed-out view of a seating plan for the European parliament in Strasbourg on September 14, 2009. gray circles are individual MEPs, while black circles embody members of conseil and commission. The associated party (EPG) is denoted at the top. The lower panel provides a zoomed-in view elaborating on the part of the upper panel marked by the dotted trapezoid shaded in gray. Alafano and Alvaro are treated as adjacent because they are sitting next to each other and belong to the same political party, i.e., ALE. Similarly, Ojuland and Oviir are considered to be adjacent. On the other hand, following the original authors, Alvaro and Cavada are not regarded as adjacent though they are seated together because they belong to different political parties, i.e., ALE and PPE, respectively. In terms of dyad-level adjacency, Cavada-Coelho and Coelho-Collino are adjacent dyads as they share Coelho, whereas Cavada-Coelho and Collino-Comi are not adjacent, but they are still connected as they have indirect paths to one another along the dyadic network.}
		\end{tablenotes}
	\end{threeparttable}
	
	\label{fig:SeatingPlanAtTheEuropeanParliament}
\end{figure}

\subsection{Summary Statistics of the Seating Arrangement}\label{appndx:EmpiricalIllustration:SummaryStastisticsOfTheSeatingArrangement}

Table \ref{tbl:EmpiricalIllustration:SummaryStastisticsOfSeatingArrangement} lists the summary statistics of the seating arrangement (for Strasbourg at term 7) when viewed as a network over pairs of MEPs. Its summary statistics are consistent with those from the Erd\"os-Renyi random network with $\lambda=1$ to $\lambda=3$ (see Table \ref{tbl:MonteCarloSimulation:DegreeCharacteristics:AverageMaxDegree_Edges}). This suggests that our empirical illustration should perform well with the Parzen kernel and bandwidth choice proposed in \citet{Kojevnikov_et_al-2021}.

\begin{table}[htbp]    
	\centering             
 	\caption{Summary Statistics of The Seating Arrangement: Strasbourg, Term 7}
	\begin{threeparttable}
		\centering
		\begin{tabular}{ccccc} \toprule \toprule   
			$d_{act}$	& $d_{max}$	& $d_{ave}$	& $e_{direct}$	& $e_{indirect}$	\\ \midrule 
			602		& 2			& 1.7076		& 514		& 3136			\\ \bottomrule 
		\end{tabular}
	\end{threeparttable}
			\parbox{6.2in}{\footnotesize See Table \ref{tbl:MonteCarloSimulation:DegreeCharacteristics:AverageMaxDegree_Edges} for the definition of the first three indicators. The last two represent the number of adjacent and connected dyads, respectively.}
	\label{tbl:EmpiricalIllustration:SummaryStastisticsOfSeatingArrangement}
\end{table} 

\subsection{Data Construction}\label{appndx:EmpiricalIllustration:DataConstruction}

Data construction for our empirical exercise in Section \ref{sec:EmpiricalIllustration} proceeds in multiple steps:
\begin{description}
	\item[Step 1:] Our subsample consists of the location of interest (i.e., Strasbourg) for the period of interest (i.e., Term 7). We select a further subset of the extracted data by seating arrangement (i.e., we focus on Pattern 1 for the present analysis - see Table \ref{tbl:EmpiricalApplication:Results:PatternsSeatingArrangements:StrasbourgTerm7}). 
	\item[Step 2:] Since our analysis is concerned with voting concordance, we follow the original authors in dropping entries with missing data or ``abstain'' in the variable ``vote.''\footnote{This amounts to assuming that those observations are missing completely at random (MCAR).} 
	\item[Step 3:] The resulting data still contains individuals belonging to ``Identity, Tradition and Sovereignty (ITS),'' one of the European Political Groups that dissolved in November 7, during the sixth term. We drop such MEPs from our analysis. 
	\item[Step 4:] The selected data is used to form the dyadic data registering the pair-of-MEPs-specific information. When pairing two MEPs, we follow \citet{Harmon_et_al-2019} in focusing on those pairs of MEPs, both of whom are
		\begin{itemize}
			\item[(i)] in the same EPG;
			\item[(ii)] from an alphabetically-seated EPG; and
			\item[(iii)] non-leaders at the time of voting.
		\end{itemize} 
	Our dyadic data consists of two types of variables: binary variables and numerical variables. The dyad-level binary (i.e., indicator) variables are defined to be one if the individual-level binary variables are the same, and zero otherwise. The dyadic-specific numerical variables in our analysis are the differences between the individual-level numerical variables, such as age and tenure. When calculating the differences in ages and tenures, we take the absolute values as we do not consider directional dyads, and we then rescale them into ten-year units. See the note below Table \ref{tbl:EmpiricalApplication:Results:MainAnalysis:StrasbourgTerm7} for details.
\end{description}

\begin{table}[htbp]
	\centering
	\caption{Patterns of Seating Arrangements: Strasbourg, Term 7}
	\begin{threeparttable}
		\centering
		{\small
		\begin{tabular}{lrcrc} \toprule \toprule
			Pattern	& \multicolumn{3}{c}{Date}			& Number of Proposals	\\ \midrule
			1		& 7/14/2009	& $\sim$	& 7/16/2009	& 116 \\ 
			2		& 8/18/2009	& $\sim$	& 8/21/2009	& 72 \\ 
			3		& 9/23/2009	& $\sim$	& 9/25/2009	& 114 \\ 
			4		& 10/13/2009	& $\sim$	& 10/16/2009	& 40 \\ 
			5		& 11/19/2009	& $\sim$	& 12/11/2009	& 94 \\ 
			6		& 1/5/2010	& $\sim$	& 1/8/2010 	& 79 \\ 
			7		& 3/17/2010	& $\sim$	& 3/19/2010	& 45 \\ 
			8		& 4/14/2010	& $\sim$	& 4/16/2010	& 120 \\ 
			9		& 5/5/2010	& $\sim$	& 5/7/2010	& 79 \\ 
			10		& 7/7/2010	& $\sim$	& 7/9/2010	& 34 \\ 
			11		& 7/21/2010	& $\sim$	& 7/22/2010	& 50 \\ 
			12		& 8/18/2010	& $\sim$	& 8/20/2010	& 118 \\ \bottomrule
		\end{tabular}
		}

		\begin{tablenotes}[para,flushleft]
			\item[] {\sl Note}: This table presents patterns of seating arrangements with the corresponding dates and the number of total observations for each pattern. Since voting may be taken place for multiple proposals within the same day, the total number proposals tends to be higher than that of days in a single pattern. For example, the first line indicates that 116 proposals were discussed and votes were cast over the three days (from the 14th of July, 2009 to the 16th of July, 2009).
		\end{tablenotes}
	\end{threeparttable}
	\label{tbl:EmpiricalApplication:Results:PatternsSeatingArrangements:StrasbourgTerm7}
\end{table}

\subsection{Full Results}\label{appndx:EmpiricalIllustration:FullResults}

Table \ref{tbl:EmpiricalApplication:Results:MainAnalysis:StrasbourgTerm7_2} reports the detailed result of the empirical illustration. As explained in Section \ref{sec:EmpiricalIllustration:Results}, Panel A reports the estimates of the parameter of interest, while Panel B lists the standard errors based on the different variance estimators. In particular, we carry out the estimation using both the network-robust variance estimator with the mean-shifted rectangular kernel, and the one with the Parzen kernel, generating the same estimates. Panel C collects the parameter estimates for other covariates accompanied by the standard errors obtained from our proposed variance estimator from equation (\ref{eq:ExogenousRegression:NetworkVarianceEstimator}), and indicates the presence or absence of day-level fixed effects.

\begin{table}[htbp]
	\centering
	\caption{Spillovers in Legislative Voting -- Main Analysis}
	\begin{threeparttable}
		\centering
		{\scriptsize
		\begin{tabular}{lccc} \toprule \toprule
									& Specification (I) 	& Specification (II) 	& Specification (III) \\ \midrule
			{\sl Panel A: Parameter estimates for Seat neighbors}&&&					 \\
			\quad Seat neighbors		& 0.0069 		& 0.0060		& 0.0060  \\ \\
			
			{\sl Panel B: Standard errors for Seat neighbors} &&& \\
			\quad Eicker-Huber-White			& 0.0031 		& 0.0030	& 0.0030 \\
			\quad Dyadic-robust			& 0.0075	 	& 0.0082 	& 0.0087 \\
			\quad Network-robust (with the rectangular kernel)    & 0.0095    & 0.0104    & 0.0112    \\ 
			\quad Network-robust (with the Parzen kernel)    & 0.0095    & 0.0104    & 0.0112    \\ \\
			
			{\sl Panel C: Parameter estimates for other covariates} &&& \\
			\quad Same country			&			& 0.0561		& 0.0562 \\
									& 			& (0.0008)		& (0.0008) \\
			\quad Same quality education 	&			& 0.0030		& 0.0028 \\
									& 			& (0.0007)		& (0.0007)	 \\
			\quad Same freshman status	&			& -0.0070		& -0.0070 \\
									&			& (0.0008)		& (0.0008) \\
			\quad Same gender			&			& 0.0004		& 0.0004 \\
									& 			& (0.0007)		& (0.0006) \\
			\quad Age difference			&			& 0.0007		& 0.0004 \\
									& 			& (0.0004)		& (0.0004) \\
			\quad Tenure difference		&			& -0.0149		& -0.0149 \\ 
									& 			& (0.0006)		& (0.0006) \\ \\
		
			\quad Day-level FE			& No			& No			& Yes \\ \\ \bottomrule
		\end{tabular}
		}
		\begin{tablenotes}[para,flushleft]
			\item[] {\footnotesize {\sl Note}: Panel A displays the parameter estimates for the three different specifications; Panel B shows the standard errors for the regression coefficient of {\sl SeatNeighbors} using different variance estimators; and Panel C collects the parameter estimates for other covariates accompanied by the standard errors obtained from our proposed variance estimator from equation (\ref{eq:ExogenousRegression:NetworkVarianceEstimator}), and indicates the presence or absence of day-level fixed effects. Adjacency of MEPs is defined at the level of a row-by-EP-by-EPG. (See the note below Figure \ref{fig:SeatingPlanAtTheEuropeanParliament}.) Independent variables are as follows: {\sl Seat neighbors} is an indicator variable denoting whether both MEPs sit together; {\sl Same country} represents an indicator for whether both MEPs are from the same country; {\sl Same quality education} is an indicator showing whether both MEPs have the same quality of education background, measured by if both have the degree from top 500 universities; {\sl Same freshman status} encodes whether both MEPs are freshman or not; {\sl Age difference} is the difference in the MEPs' ages; and {\sl Tenure difference} measures the difference in the MEPs' tenures.}
		\end{tablenotes}
	\end{threeparttable}
	\label{tbl:EmpiricalApplication:Results:MainAnalysis:StrasbourgTerm7_2}
\end{table}

\end{document}